\documentclass[10pt, english]{article}

\usepackage{bbm}

\usepackage[utf8]{inputenc}
\usepackage[T1]{fontenc}
\usepackage[a4paper, margin=1.985cm]{geometry}
\usepackage{libertine} 
\usepackage{inconsolata} 
\usepackage{amsmath, amsthm, amssymb}
\usepackage{graphicx}

\usepackage[shortlabels]{enumitem}
\usepackage{authblk}
\usepackage{thmtools}
\usepackage{thm-restate}
\usepackage{authblk}
\usepackage[hyphens]{xurl}
\usepackage{float}
\usepackage{xcolor}
\usepackage{mathtools}
\usepackage{setspace}
\usepackage{thm-restate}

\usepackage{algorithm}
\PassOptionsToPackage{noend}{algpseudocode}
\usepackage{algpseudocode}
\usepackage[colorlinks=true, linkcolor = bordeaux, citecolor = darkblue, urlcolor = darkblue]{hyperref}

\makeatletter
\@addtoreset{ALG@line}{algorithm}
\makeatother

\usepackage{etoolbox}
\usepackage{tikz}
\usetikzlibrary{tikzmark}
\usetikzlibrary{calc}

\errorcontextlines\maxdimen

\newcommand{\ALGtikzmarkcolor}{black}
\newcommand{\ALGtikzmarkextraindent}{4pt}
\newcommand{\ALGtikzmarkverticaloffsetstart}{-.5ex}
\newcommand{\ALGtikzmarkverticaloffsetend}{-.5ex}
\makeatletter
\newcounter{ALG@tikzmark@tempcnta}

\newcommand\ALG@tikzmark@start{%
    \global\let\ALG@tikzmark@last\ALG@tikzmark@starttext%
    \expandafter\edef\csname ALG@tikzmark@\theALG@nested\endcsname{\theALG@tikzmark@tempcnta}%
    \tikzmark{ALG@tikzmark@start@\csname ALG@tikzmark@\theALG@nested\endcsname}%
    \addtocounter{ALG@tikzmark@tempcnta}{1}%
}

\def\ALG@tikzmark@starttext{start}
\newcommand\ALG@tikzmark@end{%
    \ifx\ALG@tikzmark@last\ALG@tikzmark@starttext
    \else
        \tikzmark{ALG@tikzmark@end@\csname ALG@tikzmark@\theALG@nested\endcsname}%
        \tikz[overlay,remember picture] \draw[\ALGtikzmarkcolor] let \p{S}=($(pic cs:ALG@tikzmark@start@\csname ALG@tikzmark@\theALG@nested\endcsname)+(\ALGtikzmarkextraindent,\ALGtikzmarkverticaloffsetstart)$), \p{E}=($(pic cs:ALG@tikzmark@end@\csname ALG@tikzmark@\theALG@nested\endcsname)+(\ALGtikzmarkextraindent,\ALGtikzmarkverticaloffsetend)$) in (\x{S},\y{S})--(\x{S},\y{E});%
    \fi
    \gdef\ALG@tikzmark@last{end}%
}

\apptocmd{\ALG@beginblock}{\ALG@tikzmark@start}{}{\errmessage{failed to patch}}
\pretocmd{\ALG@endblock}{\ALG@tikzmark@end}{}{\errmessage{failed to patch}}
\makeatother

\declaretheorem[name=Lemma, numberwithin = section]{lemma}
\declaretheorem[name=Theorem,sibling = lemma]{theorem}
\declaretheorem[name=Theorem,sibling = lemma]{thm}
\declaretheorem[name=Proposition, sibling=lemma]{proposition}
\declaretheorem[name=Observation, sibling=lemma]{observation}
\declaretheorem[name=Definition, sibling=lemma]{definition}
\declaretheorem[name=Corollary, sibling=lemma]{corollary}

\declaretheorem[name=Remark, sibling=lemma]{remark}
\declaretheorem[name=Claim]{claim}

\newtheorem*{claim*}{Claim}

\usepackage[noabbrev,capitalise,nameinlink]{cleveref}
\crefname{claim}{Claim}{Claims}
\crefname{lemma}{Lemma}{Lemmas}
\crefname{theorem}{Theorem}{Theorems}
\crefname{proposition}{Proposition}{Propositions}
\crefname{question}{Question}{Questions}
\crefname{definition}{Definition}{Definitions}
\crefname{conjecture}{Conjecture}{Conjectures}
\crefname{observation}{Observation}{Observations}
\crefname{corollary}{Corollary}{Corollaries}
\crefformat{equation}{(#2#1#3)}
\Crefformat{equation}{Equation #2(#1)#3}
\crefname{remark}{Remark}{Remarks}
\crefname{scenario}{Scenario}{Scenarios}

\def\cqedsymbol{\ifmmode$\lrcorner$\else{\unskip\nobreak\hfil
\penalty50\hskip1em\null\nobreak\hfil$\lrcorner$
\parfillskip=0pt\finalhyphendemerits=0\endgraf}\fi}

\newenvironment{subproof}[1][\proofname]{%
  \begin{proof}[#1]%
}{%
  \end{proof}%
}

\setlist[itemize]{topsep=0ex,itemsep=0ex,parsep=0.25ex}
\setlist[enumerate]{topsep=0ex,itemsep=0ex,parsep=0.25ex}

\definecolor{bordeaux}{RGB}{100,0,50}
\definecolor{darkblue}{RGB}{25, 25, 112}

\newcommand{\mcm}[3]{\newcommand{#1}[#2]{{\ensuremath{#3}}}} 

\mcm{\Acal}{0}{\mathcal A}
\mcm{\Bcal}{0}{\mathcal B}
\mcm{\Ccal}{0}{\mathcal C}
\mcm{\Dcal}{0}{\mathcal D}
\mcm{\Ecal}{0}{\mathcal E}
\mcm{\Fcal}{0}{\mathcal F}
\mcm{\Gcal}{0}{\mathcal G}
\mcm{\Hcal}{0}{\mathcal H}
\mcm{\Ical}{0}{\mathcal I}
\mcm{\Jcal}{0}{\mathcal J}
\mcm{\Kcal}{0}{\mathcal K}
\mcm{\Lcal}{0}{\mathcal L}
\mcm{\Mcal}{0}{\mathcal M}
\mcm{\Ncal}{0}{\mathcal N}
\mcm{\Ocal}{0}{{\mathcal O}}
\mcm{\Pcal}{0}{{\mathcal P}}
\mcm{\Qcal}{0}{{\mathcal Q}}
\mcm{\Rcal}{0}{{\mathcal R}}
\mcm{\Scal}{0}{{\mathcal S}}
\mcm{\Tcal}{0}{{\mathcal T}}
\mcm{\Ucal}{0}{{\mathcal U}}
\mcm{\Vcal}{0}{{\mathcal V}}
\mcm{\Wcal}{0}{{\mathcal W}}
\mcm{\Xcal}{0}{{\mathcal X}}
\mcm{\Ycal}{0}{{\mathcal Y}}
\mcm{\Zcal}{0}{{\mathcal Z}}

\newcommand*\samethanks[1][\value{footnote}]{\footnotemark[#1]}

\title{A Structural Linear-Time Algorithm for Computing the Tutte Decomposition}

\author[1]{Romain Bourneuf\thanks{Email: \texttt{romain.bourneuf@ens-lyon.fr, tim.planken@math.tu-freiberg.de}}}
\affil[1]{Univ. Bordeaux, CNRS, Bordeaux INP, LaBRI, UMR 5800, F-33400 Talence, France.}
\author[2]{Tim Planken\samethanks\thanks{Funded by DFG, project number 546892829.}}
\affil[2]{TU Freiberg, Germany.}

\date{}

\begin{document}
\maketitle

\begin{abstract}
    The block–cut tree decomposes a connected graph along its cutvertices, displaying its 2-connected components.
    The Tutte-decomposition extends this idea to 2-separators in 2-connected graphs, yielding a canonical tree-decomposition that decomposes the graph into its triconnected components.
    In 1973, Hopcroft and Tarjan introduced a linear-time algorithm to compute the Tutte-decomposition.
    Cunningham and Edmonds later established a structural characterization of the Tutte-decomposition via totally-nested 2-separations. 
    We present a conceptually simple algorithm based on this characterization, which computes the Tutte-decomposition in linear time.
    Our algorithm first computes all totally-nested 2-separations and then builds the Tutte-decomposition from them.
    
    Along the way, we derive new structural results on the structure of totally-nested 2-separations in 2-connected graphs using a novel notion of \emph{stability}, which may be of independent interest.
\end{abstract}

\section{Introduction}

A recurring theme in graph theory is to decompose graphs into simpler, well-structured pieces.
The most basic example is the decomposition of a graph into its connected components, which can be computed in linear time using depth-first search.
For connected graphs, a natural next step is the \emph{block-cut-tree} decomposition, which is roughly obtained by cutting the graph at its cutvertices into maximal 2-connected subgraphs, or \emph{blocks}.
Hopcroft and Tarjan \cite{Hopcroft-Tarjan-1-sep} proved that the block-cut-tree decomposition of a connected graph can be computed in linear time.
Decomposing 2-connected graphs into 3-connected parts by cutting at their 2-separators is slightly more complicated, as one 2-separator can separate the two vertices of another 2-separator.
From the structural perspective, this problem was solved in 1966 by Tutte \cite{Tutte}, who proved that there is a canonical way to decompose a graph in a tree-like way by cutting it at 2-separators, such that the basic parts, the \emph{triconnected components}\footnote{See \hyperref[para:Tutte]{later} for a formal definition. As an example, if $H$ is obtained from a 2-connected graph $G$ by subdividing every edge once then the triconnected components of $H$ are triangles, triple edges and $G$ itself.}, are either 3-connected, cycles or edges.
This decomposition is known as the \emph{Tutte-decomposition}.

A central motivation behind decomposing 2-connected graphs into 3-connected parts is a result of Mac Lane \cite{MacLane37}, stating that a graph is planar if and only if its triconnected components are planar.
Another key motivation is a theorem of Whitney \cite{Whitney}, which states that a 3-connected planar graph has a unique embedding embedding in the plane (up to mirroring). 
Therefore, an efficient algorithm to divide a graph into 3-connected parts can be very useful as a subroutine for planarity testing \cite{BSW70} or isomorphism testing in planar graphs \cite{HT71,HT74,HW74}.

For this reason, the problem of efficiently decomposing a 2-connected graph into 3-connected parts has been studied extensively. In 1973, Hopcroft and Tarjan  \cite{Hopcroft-Tarjan-2-sep} presented a linear-time algorithm to compute the Tutte-decomposition.
Their algorithm is not known to be efficiently parallelizable so in 1992, Miller and Ramachandran \cite{MR92} gave an efficient parallel algorithm for this problem, based on ear decompositions.
Building on this paper, in 1993, Fussell, Ramachandran and Thurimella \cite{FRT93} improved this parallel algorithm using the technique of local replacement. Another linear-time sequential algorithm for computing the Tutte-decomposition can be derived from their results.

For 2-connected planar graphs, the decomposition into 3-connected parts is often stored using a data structure known as the \emph{SPQR tree}. SPQR trees were first introduced in the context of planar graphs by Di Battista and Tamassia \cite{DT96}, in relation to planarity testing, and then for general graphs in \cite{DT90} for various on-line algorithms. Di~Battista and Tamassia \cite{DT962} further showed that they can be used to efficiently maintain the triconnected components of a dynamic graph. SPQR trees have found multiple other applications, notably in graph drawing, see for instance \cite{DL98,Gutwenger10,BHR19,BR20}.

In their seminal paper introducing SPQR trees \cite{DT96}, Di Battista and Tamassia suggest that the SPQR trees can be constructed in linear time using the algorithm of Hopcroft and Tarjan to compute the Tutte-decomposition \cite{Hopcroft-Tarjan-2-sep}. The first linear-time implementation of such an algorithm was provided in 2001 by Gutwenger and Mutzel \cite{Gutwenger-Mutzel}, who also corrected some errors in the work of Hopcroft and Tarjan. 

In 1980, Cunningham and Edmonds \cite{Cunningham_Edmonds_1980} gave an alternative, more structural, characterization of the Tutte-decomposition. Informally, two 2-separations of a graph are \emph{nested} if they do not cut through each other. A 2-separation of a graph $G$ is \emph{totally-nested} if it is nested with every other 2-separation of $G$.
Cunningham and Edmonds \cite{Cunningham_Edmonds_1980} prove that the Tutte-decomposition of a graph can be obtained by cutting it precisely along its totally-nested 2-separations.

In this paper, we show how to use the structural characterization of Cunningham and Edmonds to efficiently compute the Tutte-decomposition of a 2-connected graph. Our main result is an alternative algorithm to that of Hopcroft and Tarjan \cite{Hopcroft-Tarjan-2-sep}. We later provide a detailed comparison of the two algorithms.

\begin{restatable}{theorem}{main} \label{thm:main}
There is an $O(n+m)$-time algorithm which, given as input a 2-connected graph $G$, returns the Tutte-decomposition of $G$ and all its triconnected components.
\end{restatable}

\paragraph{Overview of our Algorithm.}

Suppose that we are given a 2-connected graph $G$.
First, we compute a set $\mathcal{S}$ of 2-separations of $G$, which contains all totally-nested 2-separations of $G$.
Then, we test the separations in~$\mathcal{S}$ to determine which of them are actually totally-nested.
Finally, we build the Tutte-decomposition directly from the set of totally-nested 2-separations of $G$.
It that sense, our algorithm corresponds to the structural characterization of the Tutte-decomposition from Cunningham and Edmonds.
The high-level roadmap of our algorithm is very natural, however the complexity lies in computing efficiently all the totally-nested 2-separations, and only them.

More precisely, say that a separation $(A, B)$ of $G$ is \emph{half-connected} if $G[A \setminus B]$ or $G[B \setminus A]$ is connected.
It is straightforward to see that every totally-nested 2-separation is half-connected, so we simply focus on computing half-connected 2-separations of $G$.

\begin{figure}[h]
    \centering
    \includegraphics[width=0.75\linewidth]{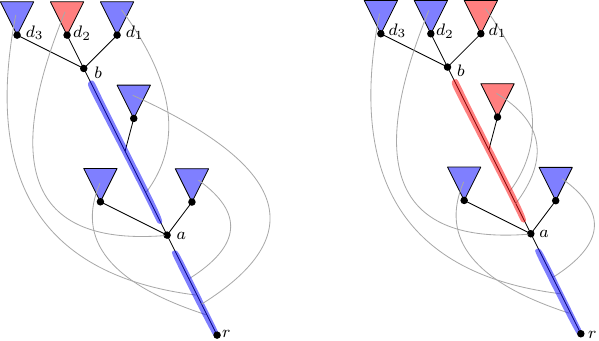}
    \caption{A type-1 separation (left) and a type-2 separation (right) of a graph~$G$ with a DFS tree with root~$r$.
    The separator is $\{a,b\}$ in both cases, while red and blue represent the proper sides of the separation.}
    \label{fig:types}
\end{figure}

Inspired by the algorithm of Hopcroft and Tarjan, our algorithm is based on a DFS tree~$T$ of $G$ and relies on a particular ordering of the vertices of $G$ according to~$T$. We consider two different types of 2-separations, called type-1 and type-2, like in the Hopcroft-Tarjan algorithm, see \cref{fig:types}.
After some easy precomputation, the half-connected type-1 separations of $G$ can be computed efficiently.
To compute type-2 separations efficiently, we introduce the concept of \emph{stability}, which we think could be useful in related problems.
Consider the type-2 separator $\{a,b\}$ in \cref{fig:types} (right) and let $a'$ be the child of~$a$ which is an ancestor of~$b$.
Then, this separator has the property that no vertex visited between~$a'$ and~$b$ in the DFS traversal sends an edge to a proper ancestor of~$a$.
Whenever this property holds for a pair $\{a,b\}$, we call the path $T[a,b]$ \emph{stable}.
We show that, after some linear-time precomputation, we can answer in constant time queries of the form: Given two vertices $a, b \in V(G)$ such that $a$ is an ancestor of $b$ in $T$, which is the ancestor $\alpha$ of $a$ which is closest to $a$ such that $T[\alpha, b]$ is stable? 
Using the concept of stability, we give a precise characterization of half-connected type-2 separations.
This characterization can be efficiently checked algorithmically.
Using this, we compute a maximal set of half-connected type-2 separations which are pairwise nested. Such a set must contain all totally-nested type-2 separations.
At that point of the algorithm, we computed a set $\mathcal{S}$ of half-connected 2-separations, which contains all totally-nested separations.

In \cref{sec:stru-tot-nest}, we give a characterization of which half-connected 2-separations are totally-nested, again algorithmically efficient. 
More precisely, using the concept of stability, after some precomputation we can compute in constant time for any given half-connected 2-separation $(A,B)$, a 2-separation that crosses $(A,B)$ if it exists, or otherwise decide that $(A, B)$ is totally-nested.
Therefore, we can extract from~$\mathcal{S}$ the set of all totally-nested 2-separations of $G$.

It is a folklore result that every set of nested separations induces a tree-decomposition.
The proof of this result is algorithmic, and can be turned into a linear-time algorithm. 
We believe that such a result already appears in the literature but we could not find it explicitly, so we prove it in detail for completeness.
Applying this result to the set of totally-nested separations of $G$, we finally obtain the Tutte-decomposition of $G$.

\paragraph{Tutte's definition of the Tutte-decomposition} \label{para:Tutte}

We now present Tutte's construction of the Tutte-decomposition.
Consider a 2-connected multigraph $G$ (meaning that there could be several edges connecting the same two vertices).
Consider two vertices $a \neq b \in V(G)$, and partition the edges of $G$ in such a way that two edges are in the same part if and only if they lie in a common path which contains neither $a$ nor $b$ as an internal vertex.
The parts of this partition are the \emph{separation classes}.
For instance, if $e \in E(G)$ connects $a$ and $b$ then $\{e\}$ is a separation class.
The pair $\{a, b\}$ is a \emph{separation pair} if there are at least two separation classes, unless one of the following holds:
\begin{itemize}
    \item there are exactly two separation classes, one of which is a single edge, or
    \item there are exactly three separation classes, each consisting of a single edge.
\end{itemize}
If $G$ has no separating pair, then $G$ is \emph{triconnected}.

Let $G$ be a 2-connected multigraph and $\{a, b\}$ be a separating pair of $G$, with separation classes $E_1, \ldots, E_s$.
Consider a bipartition $(E', E'')$ of $E$ which is a coarsening of the partition $(E_1, \ldots, E_s)$, and such that $|E'| \geq 2$ and $|E''| \geq 2$.
Let $G_1 = (V(E'), E' \,\dot\cup\, \{a, b\})$ and $G_2 = (V(E''), E'' \,\dot\cup\, \{a, b\})$.
Then, $G_1$ and $G_2$ are obtained by \emph{splitting} $G$ with respect to $\{a, b\}$, and are \emph{split multigraphs}.
Note that they are not uniquely defined, since they depend on the bipartition $(E', E'')$ of $E$.
The new edges $\{a, b\}$ of $G_1$ and $G_2$ are \emph{virtual edges}, and they are \emph{associated}.
If $G$ is 2-connected then all its split multigraphs are also 2-connected.
Suppose that $G$ is split, its split multigraphs are split, and so on until no more splits are possible.
The graphs constructed in this way are called \emph{split components} of $G$, and they are all triconnected by construction. 
Again, they are not uniquely defined.
Every split component is either a 3-connected graph, a $K_3$ or a triple edge between two vertices.

Let $G_1 = (V_1, E_1)$ and $G_2 = (V_2, E_2)$ be two split components of $G$, which contain two associated edges $e_1$ and $e_2$.
The graph $G = (V_1 \cup V_2, (E_1 \cup E_2) \setminus \{e_1, e_2\})$ is obtained by \emph{merging} $G_1$ and $G_2$.
Note that we can recover $G$ from its split components by merging them recursively.
The \emph{triconnected components} of $G$ are obtained from its split components by merging recursively the triangles into cycles, and the triple edges into multiple parallel edges.
The set of triconnected components is unique and canonical.
The graph $T$ whose vertex set is the set of triconnected components of $G$, and where two triconnected components are adjacent if they share associated virtual edges is a tree.
For every triconnected component~$t$, denote this component by $G_t = (V_t, E_t)$.
Then, the tree-decomposition $(T, (V_t)_{t \in V(T)})$ is the Tutte-decomposition of $G$, and the graphs $G_t$ are its torsos.

\paragraph{Comparison to the Hopcroft-Tarjan algorithm.}

The Hopcroft-Tarjan algorithm \cite{Hopcroft-Tarjan-2-sep} is based on Tutte's definition.
At a very high level, their algorithm builds the Tutte-decomposition following the construction of Tutte.
More precisely, starting from a 2-connected graph $G$, they recursively split $G$ to obtain split components of $G$.
Then, they perform the ``merging step'' recursively on these split components to obtain the triconnected components, and they finally build the Tutte-decomposition using the virtual edges.

Their algorithm to find split components of $G$ is based on a DFS tree~$T$ of $G$ and relies on the same particular ordering of the vertices of $G$ according to~$T$ as our algorithm. 
It is a recursive algorithm which traverses the tree~$T$ following the ordering of the vertices, and performs split operations whenever possible.
This algorithm is fairly long, and its analysis is nice but intricate, as witnessed by the mistake highlighted by Gutwenger and Mutzel \cite{Gutwenger-Mutzel}.

Our algorithm has the advantage of being much more modular: it can naturally be divided into subroutines which can be written and checked independently.
Our key subroutines are the computation of some values at each vertex of the graph, the computation of half-connected type-1 separations, the computation of a maximal nested set of half-connected type-2 separations, the tests whether they are totally-nested, and the final step to build the Tutte-decomposition from the set of totally-nested 2-separations.
Also, because of the difference in elegance between Tutte's definition and the characterization from Edmonds and Cunningham, our approach is arguably conceptually simpler than that of Hopcroft and Tarjan. It is also more direct, since it consists in finding the totally-nested separations and directly cutting at them, instead of splitting greedily and then having a merging step.
Furthermore, as it follows a more structural approach based on totally-nested separations, there is hope that our algorithm might be adapted to compute $(k+1)$-connected components by cutting at totally-nested $k$-separations, for $k \geq 3$.

\paragraph{Open problems and related work.}

More than 50 years ago, Aho, Hopcroft and Ullman \cite{AHU74} conjectured that for every integer $k \geq 1$, there exists a linear-time algorithm to compute the $k$-connected components of a graph.
For $k=1$, this can simply be done by a DFS, and for $k=2$, this is done by an algorithm of Hopcroft and Tarjan \cite{Hopcroft-Tarjan-1-sep}.
For $k \geq 3$, there is no clear notion of $k$-connected components, although arguably the most natural notion for $k=3$ is that of triconnected components, and computing them amounts to computing the Tutte-decomposition, which is done by the algorithm of Hopcroft and Tarjan \cite{Hopcroft-Tarjan-2-sep}.
For a long time, it was unclear what should be the correct notion of $4$-connected components and $5$-connected components.
One could expect that the characterization of the Tutte-decomposition by Cunningham and Edmonds extends to 3-separations, namely that the parts obtained after cutting a 3-connected graph at its totally-nested 3-separations are 4-connected, wheels or triangles.
Somewhat surprisingly, this is not true. 
In 2023, Carmesin and Kurkofka \cite{Carmesin-tridecomposition} found a way to define a canonical decomposition of 3-connected graphs whose basic parts are quasi 4-connected, wheels or so-called thickened $K_{3, m}$'s. This is done by introducing a notion of \emph{tri-separations}.
Even more recently, Kurkofka and Planken \cite{KP25} proved a similar result about canonically decomposing 4-connected graphs into basic parts which are either quasi-5-connected, double-wheels, thickened $K_{4,m}'s$ or other simple parts. They do it by generalizing the notion of tri-separations to \emph{tetra-separations}.
These works suggest natural candidates for the definition of $4$-connected components and $5$-connected components, and it would be interesting to see whether these canonical decompositions can be computed in linear time.
Note that even the problem of deciding in linear time whether a graph is $k$-connected for $k \geq 4$ was open until early 2025.
This problem was solved in a remarkable paper of Korhonen \cite{Korhonen25}.

Similar questions also hold for edge-connected components.
Again, for $k=1$, this can simply be done by a DFS, and for $k=2$, this follows from the algorithm of Hopcroft and Tarjan for $2$-connected components.
For $k=3$, Galil and Italiano \cite{GI93} observed that this also follows from the algorithm of Hopcroft and Tarjan for triconnected components.
For $k=4$, a linear-time algorithm was proposed independently by Nadara, Radecki, Smulewicz, and Sokołowski \cite{NRSS21}, and by Georgiadis, Italiano, and Kosinas \cite{GIK21}.
For $k=5$, Kosinas \cite{Kosinas23} gave a linear-time algorithm.
Finally, in the same remarkable paper \cite{Korhonen25}, Korhonen showed that for every integer $k \geq 1$, there is a linear-time algorithm to compute the $k$-edge-connected components of a graph.
More precisely, it is shown in \cite{Korhonen25} that there is a $k^{O(k^2)}m$ algorithm for computing a $k$-lean tree-decomposition of a given graph.

In the context of directed graphs, there are several standard algorithms based on DFS to compute strongly connected components in linear time.
For 2-vertex connectivity in directed graphs, Georgiadis, Italiano, Laura and Parotsidis \cite{GILP18} gave a linear-time algorithm which computes the 2-vertex-connected blocks.
To the best of our knowledge, there is no such result for $k$-vertex connectivity in directed graphs for any $k \geq 3$. 
It would be interesting to see whether a directed variant of stability can be a useful tool in this context.

\paragraph{Organization of the paper.}

We start by reviewing some background and setting up some notations in \cref{sec:prelim}.
We then investigate the structure of (half-connected) 2-separations in \cref{sec:2sep}.
In \cref{sec:4}, we show how this can be used to compute all half-connected type-1 separations, and a maximal nested set of half-connected type-2 separations.
In \cref{sec:stru-tot-nest}, we characterize totally-nested half-connected 2-separations.
Finally, we show how to implement these characterizations algorithmically in \cref{sec:build-tutte}, from where we can easily conclude the proof of \cref{thm:main}.
\section{Preliminaries} \label{sec:prelim}

\subsection{Separations}

A \emph{separation} of a graph $G$ is a pair $(A, B)$ such that $A \cup B = V(G)$, both $A \setminus B$ and $B \setminus A$ are nonempty and there is no $(A \setminus B)$--$(B\setminus A)$ edge.
The set $A \cap B$ is the \emph{separator} of the separation $(A, B)$, and its size is the \emph{order} of the separation. A separation of order $k$ is a \emph{$k$-separation}, and its separator is a \emph{$k$-separator}. Observe that if $(A, B)$ is a $k$-separation then $(B, A)$ is also a $k$-separation. Throughout this paper, we shall make no distinction between the separations $(A, B)$ and $(B, A)$.
A separation $(A, B)$ is \emph{minimal} if every vertex in the separator $A \cap B$ has a neighbor in both $A \setminus B$ and $B \setminus A$.
For $k \in \mathbb{N}$, a graph $G$ is \emph{$k$-connected} if $|V(G)| > k$ and $G$ has no separation of order at most $k-1$.

Two separations $(A, B)$ and $(C, D)$ of a graph $G$ are \emph{nested} if, after possibly switching $A$ and $B$, and possibly switching $C$ and $D$, we have $A \subseteq C$ and $D \subseteq B$. If they are not nested, $(A, B)$ and $(C, D)$ \emph{cross}.
A set $\mathcal{S}$ of separations of $G$ is \emph{nested} if the separations of $\mathcal{S}$ are pairwise nested.
A 2-separation of~$G$ is \emph{totally-nested} if it is nested with every 2-separation of~$G$. 

The way two 2-separations can cross in a 2-connected graph is very restricted, and can be described as follows.

\begin{lemma}[\protect{\cite[Lemma A.2.2]{Carmesin-tridecomposition}}] \label{lem:crossing-vertices}
    Two 2-separations $(A, B)$ and $(C, D)$ of a 2-connected graph $G$ cross if and only if one of the following assertions holds.
    \begin{itemize}
        \item $(A, B)$ separates the two vertices of $C \cap D$ while $(C, D)$ separates the two vertices of $A \cap B$.
        \item $A \cap B = C \cap D$ and there are four components $H_1, \ldots, H_4$ of $G - (A \cap B)$ such that $H_1, H_2 \subseteq G[A]$ and ${H_3, H_4 \subseteq G[B]}$, while $H_1, H_3 \subseteq G[C]$ and $H_2, H_4 \subseteq G[D]$.
    \end{itemize}
\end{lemma}

A separation $(A, B)$ of a graph $G$ is \emph{half-connected} if $G[A \setminus B]$ or $G[B \setminus A]$ is connected. The following result is folklore.

\begin{lemma} \label{lem:tot-nested-half-con}
    Every totally-nested 2-separation $(A, B)$ of a graph $G$ is half-connected.
\end{lemma}

\begin{proof}
    By contraposition, suppose that $(A, B)$ is not half-connected.
    Let $H_1, H_2$ be two connected components of ${G[A \setminus B]}$ and let $H_3, H_4$ be two connected components of $G[B \setminus A]$.
    Let $C = V(H_1) \cup V(H_3) \cup (A \cap B)$ and $D = V(G) \setminus (C \setminus (A \cap B))$.
    Then, $(C, D)$ is a 2-separation of $G$, and $(A, B)$ and $(C, D)$ cross by \cref{lem:crossing-vertices}.
\end{proof}


Let us recall some basic facts about tree decompositions and their relation to nested sets of separations.
Let ${(T, \mathcal{V})}$ be a tree-decomposition of a graph $G$  with $\mathcal{V}=(V_t)_{t \in V(T)}$ and $e = (t_1, t_2)$ be an edge of $T$. The graph $T - e$ is a forest with exactly two components, a component $T_1$ containing $t_1$ and a component $T_2$ containing $t_2$. 
Let $A_e = \bigcup_{t \in V(T_1)} V_t$ and $B_e = \bigcup_{t \in V(T_2)} V_t$. Then, $(A_e, B_e)$ is a separation of $G$, with separator $V_{t_1} \cap V_{t_2}$. The separation $(A_e, B_e)$ is \emph{induced} by the edge $e = (t_1, t_2)$ of $T$, and more generally is induced by the tree-decomposition $(T, \mathcal{V})$.
Conversely, a set $\mathcal{S}$ of separations \emph{induces} a tree-decomposition $(T, \mathcal{V})$ of $G$ if $e \mapsto (A_e, B_e)$ is a bijection from $E(T)$ to $\mathcal{S}$. That is, $\mathcal{S}$ induces a tree-decomposition $(T, \mathcal{V})$ if and only if $\mathcal{S}$ is the set of separations induced by the edges of $T$.

The following result is folklore, see e.g. \cite{GraphsMinorsX}.

\begin{lemma} \label{lem:nested-iff-tree-decomp}
    A set $\mathcal{S}$ of separations of a graph $G$ is nested if and only if there exists a tree-decomposition of $G$ which is induced by $\mathcal{S}$. If so, there is a unique such tree decomposition up to isomorphism.
\end{lemma}

Given an edge $(t_1, t_2)$ of a tree-decomposition $(T, \mathcal{V})$, the \emph{adhesion set} of $V_{t_1}$ and $V_{t_2}$ is the set $V_{t_1} \cap V_{t_2}$. The \emph{adhesion} of $(T, \mathcal{V})$ is the maximum size of an adhesion set. If $T$ only has one node then $(T, \mathcal{V})$ has adhesion 0.
The \emph{torso} at a node $t \in V(T)$ is the supergraph of $G[V_t]$ obtained by turning each adhesion set into a complete graph. More formally, the torso at $t$ has vertex set $V_t$, and $(u, v)$ is an edge of the torso at $t$ if and only if either $(u, v) \in E(G)$, or there exists an edge $(t, t')$ of $T$ such that $\{u, v\} \subseteq V_t \cap V_{t'}$. The edges of the torso that are not edges of $G$ are called \emph{torso edges}.

In 1961, Tutte \cite{Tutte} proved that every 2-connected graph admits a canonical tree-decomposition of adhesion 2 in which every torso is either 3-connected, a cycle or a $K_2$. In 1980, Cunnigham and Edmonds \cite{Cunningham_Edmonds_1980} gave various equivalent characterizations of this canonical tree-decomposition. Some of these characterizations also appear in the work of Hopcroft and Tarjan \cite{Hopcroft-Tarjan-2-sep}. We shall focus on the following characterization, due to Cunnigham and Edmonds.

\begin{theorem}[\cite{Cunningham_Edmonds_1980}]
If $G$ is $2$-connected, its totally-nested 2-separations induce a tree-decomposition of $G$ all whose torsos are minors of $G$ and are either 3-connected, cycles or $K_2$'s. This tree-decomposition is unique up to isomorphism.
\end{theorem}

For a 2-connected graph $G$, we call this tree-decomposition the \emph{Tutte-decomposition} of $G$.

\subsection{Notions on trees}

For every graph $G$, we shall denote its number of vertices by $n$ and its number of edges by $m$.
Let $G$ be a connected graph and $T$ be a rooted spanning tree of $G$, with root $r \in V(G)$.
If $v \neq r$, we denote by $p(v)$ the \emph{parent} of $v$ in $T$.
If $u \neq v$ and $p(u) = p(v)$ then $u$ and $v$ are \emph{siblings} in $T$.
If $u, v \in V(T)$, we denote by $T[u, v]$ the unique simple path between $u$ and $v$ in $T$. We will also use $T[u, v)$, $T(u, v]$ and $T(u, v)$ to denote the path $T[u, v]$ minus the vertex on the side of the parenthesis.
If $v \in T[r, u]$ then $v$ is an \emph{ancestor} of $u$, and $u$ is a \emph{descendant} of $v$. 
If $v \in T[r, u)$ then $v$ is a \emph{proper ancestor} of $u$ and $u$ is a \emph{proper descendant} of $v$. 
We denote by $Desc(v)$ the set of descendants of $v$ in $T$.
For every vertex $v \in V(T)$, we let $T(v)$ denote the subtree of $T$ rooted in $v$. 
The vertex set of $T(v)$ is exactly $Desc(v)$. 
We denote by $ND(v)$ the number of descendants of $v$, so $ND(v) = |Desc(v)|$.
Two vertices $u$ and $v$ are \emph{comparable} in $T$ if one of them is an ancestor of the other. 
Note that for every vertex $v$, all ancestors of $v$ are pairwise comparable in $T$. 

We define an ancestor-descendant relation on the edges of $T$ in the natural way: an edge $e = (a, b)$ is an \emph{ancestor} of an edge $e' = (a', b')$ if $a$ and $b$ are both ancestors of $a'$ and $b'$ in $T$. 
We view $T$ as a rooted tree with its root at the bottom, hence if $v$ is an ancestor of $u$, we might say that $v$ is \emph{lower} than $u$ and that $u$ is \emph{higher} than $v$. 
Given a nonempty set $X$ of pairwise comparable vertices, the \emph{lowest} element of $X$ is the unique $x \in X$ that is an ancestor of all elements of $X$. 
We similarly define the \emph{highest} element of $X$ as the unique $x \in X$ that is a descendant of all elements of $X$.
The \emph{second lowest} element of $X$ is the lowest element of $X \setminus \{x\}$ if this set is non-empty, and is not defined otherwise. 
For any integer $k$ and any set $X$ of pairwise comparable vertices, we define the \emph{$k$-th lowest} element of $X$ analogously. 
If $|X| < k$, the \emph{$k$-th lowest} element of $X$ is not defined.

A spanning tree $T$ of a graph~$G$ is \emph{normal} if the endpoints of every edge of $G$ are comparable in $T$. The edges of $G$ that are not edges of $T$ are called \emph{back-edges}. When considering a back-edge $(x, y)$, we will always implicitly assume that $x$ is a descendant of $y$.
Observe that every DFS spanning tree is a normal spanning tree, and therefore every connected graph has a normal spanning tree. 

For a vertex $v \in V(G)$, let $L(v)$ be the set of proper ancestors of $v$ that are adjacent in $G$ to a descendant (not necessarily proper) of $v$. 
Since $L(v)$ only contains proper ancestors of $v$, all elements of $L(v)$ are pairwise comparable.
For any integer $k$, the \emph{$k$-th lowpoint} of $v$ is the $k$-th lowest element of $L(v)$ if $|L(v)| \geq k$, and is $v$ otherwise. 
We denote it by $lwpt_k(v)$. 
The first lowpoint of $v$ will also simply be called the \emph{lowpoint} of $v$ and denoted $lwpt(v)$.
The \emph{highpoint} of $v$ is the highest element of $L(v) \setminus \{p(v)\}$ if this set is nonempty, and is $v$ otherwise. We denote it by $hgpt(v)$.

Suppose that the vertices of $G$ are numbered from 1 to $n$.
We will identify the vertices with their number.
For every vertex $v$, we denote by $N_{\min}(v)$ the smallest neighbour of $v$ in $G$ and by $N_{\max}(v)$ its largest neighbour.
For every vertex $v$, the children of $v$ in $T$ are totally ordered by their numbering.
Let $v \in V(G)$ and $c$ be a child of $v$. We denote by $maxhgpt(c)$ the maximum of $hgpt(c')$ over all children $c'$ of $v$ such that $c' \geq c$.
Two children $c_1 \neq c_2$ of $v$ are \emph{consecutive} if there is no other child $c$ of $v$ such that $c_1 < c < c_2$ or $c_2 < c < c_1$. Note that this does not necessarily imply $|c_2 - c_1| = 1$.
For every vertex $v$, we shall think of the children of $v$ as being placed in the tree from right to left by increasing number.
If $v$ is not a leaf, the \emph{left child} of $v$ is the largest child of $v$. We shall denote it by $\ell(v)$.
A vertex $w$ is a \emph{leftmost descendant} of $v$ if there exists a sequence $(v = v_0, v_1, \ldots, v_k = w)$ such that $v_i$ is the left child of $v_{i-1}$ for every $i \in [k]$. 
If $a, a', b$ are three vertices such that $a'$ is a child of $a$ and $b$ is a leftmost descendant of $a'$ then the path $T[a, b]$ is a \emph{leftmost path}. Note that we don't require $a'$ to be the left child of $a$, thus $b$ might not be a leftmost descendant of $a$.
Observe that the maximal leftmost paths of $T$ partition of the edges of $T$.

We define a total ordering $\leq'$ on pairs of integers as follows. Given two pairs of integers $(x, y)$ and $(x', y')$, we set $(x, y) \leq' (x', y')$ if and only if $(-x, y) \leq_{lex} (-x', y')$, where $\leq_{lex}$ is the lexicographic order.

\begin{definition}
\label{dfn:compatible-numbering}
    Let $G$ be a connected graph and $T$ be a normal spanning tree of $G$.
    A numbering of the vertices of $G$ is \emph{compatible with $T$} if the following conditions hold. \begin{enumerate}[(i)]
        \item The vertices are numbered from 1 to $n$. \label{item:numbered1-n}
        \item\label{itm:compatible-2} For every $j \in [n]$, the descendants of $j$ form exactly the interval $[j, j + ND(j) - 1]$. \label{item:descendants-interval}
        \item\label{itm:compatible-3} If $j < k$ are siblings then $(lwpt_1(j), lwpt_2(j)) \leq' (lwpt_1(k), lwpt_2(k))$. \label{item:order-siblings}
    \end{enumerate}
\end{definition}

This definition might seem obscure at first, we now try to give some intuition for it. \cref{item:descendants-interval} basically says that the numbering corresponds to the order in which vertices are discovered during a DFS of $T$. The idea behind \cref{item:order-siblings} is that the children of every vertex $v$ are ordered by decreasing lowpoint, so that the child of $v$ with the lowest lowpoint is the left child of $v$. However, we will need that children of $v$ with the same lowpoint are ordered by increasing second lowpoint, hence the definition of $\leq'$. The rationale behind this choice of ordering should be made clear by \cref{cor:cyclic-intervals}.

\begin{observation} \label{obs:numbering}
    Let $G$ be a connected graph and $T$ be a normal spanning tree of $G$. Suppose that the vertices of $G$ are numbered with a numbering compatible with $T$. The following are immediate consequences of \cref{item:descendants-interval,,item:order-siblings}. \begin{itemize}
        \item If $j < k$ are siblings and $j'$ is a descendant of $j$ then $j' < k$.
        \item If $i \leq j \leq k$ and $k$ is a descendant of $i$ then $j$ is also a descendant of $i$. 
        \item If $i$ is a proper ancestor of $j$ then $i < j$. 
        \item The root of $T$ is numbered 1.
        \item If $j < k$ are siblings then $lwpt(j) \geq lwpt(k)$.
    \end{itemize}
\end{observation}

\begin{definition}
    A numbering of the vertices of a graph $G$ is \emph{compatible with $G$} if the vertices are numbered from 1 to $n$ and for every totally-nested 2-separation $(A, B)$ of $G$, both $A \setminus B$ and $B \setminus A$ are cyclic intervals of $[n] \setminus (A \cap B)$. 
\end{definition}

As we shall see later (see \cref{rem:compatible}), if $G$ is a 2-connected graph with a normal spanning tree $T$, every numbering of $V(G)$ that is compatible with $T$ is also compatible with $G$.

Given a connected graph~$G$ and a normal spanning tree~$T$ of~$G$ with numbering of the vertices of~$G$ that is compatible with~$T$, we define $LD(v)$ to be the largest descendant of~$v \in V(G)$.

\subsection{Background on 1-separations}

We now review some basic results about 1-separations.
The following is a classical result from Hopcroft and Tarjan.

\begin{lemma}[\cite{Hopcroft-Tarjan-1-sep}] \label{lem:algo-1-sep}
    There is an $O(n+m)$-time algorithm which, given a connected graph $G$, determines whether $G$ is 2-connected, and if not returns a 1-separator of $G$.
\end{lemma}

The main structural result in \cite{Hopcroft-Tarjan-1-sep} behind this linear-time algorithm can be stated as follows.

\begin{lemma} \label{lem:charac-1-sep}
    Let $G$ be a connected graph, $T$ a normal spanning tree of $G$ with root $r$, and $v$ be any vertex of $G$. Then, $\{v\}$ is a separator of $G$ if and only if one of the following holds.
    \begin{itemize}
        \item $v = 1$ and $v$ has at least two children in $T$.
        \item $v \neq 1$ and $v$ has a child $w$ such that $lwpt(w) = v$.
    \end{itemize}
\end{lemma} 

\subsection{Representing separations}
\label{subsec:representing-separations}

Let $G$ be a 2-connected graph, $T$ be a normal spanning tree of $G$ and assume that the vertices of $G$ are numbered with a numbering compatible with~$T$. Let $(A, B)$ be a half-connected 2-separation of $G$ with separator $\{a, b\}$. 
Let $c_1 < \ldots < c_k$ be the children of $b$. 
By \cref{cor:cyclic-intervals}, the sets $A \setminus \{a, b\}$ and $B \setminus \{a, b\}$ are cyclic intervals of $[n] \setminus \{a, b\}$. Up to renaming $A$ and $B$, we can assume that $B \setminus \{a, b\}$ is an interval of $[n] \setminus \{a, b\}$. Let \textbf{f} be the first vertex of the interval of $B \setminus \{a, b\}$ and \textbf{l} be its last vertex. Then, the tuple $(a, b, \textbf{f}, \textbf{l})$ represents $(A, B)$ unambiguously. 
To simplify later computations, we shall also store the minimum child $c_m$ of $b$ in $B$ (we store 0 if there is no such child) and the maximum child $c_M$ of $b$ in $B$ (again, we store 0 if there is no such child).
Hence, we define the \emph{encoding} of the 2-separation $(A, B)$ to be the tuple $(a, b, \textbf{f}, \textbf{l}, c_m, c_M)$.
\section{The structure of 2-separations} \label{sec:2sep}

The algorithm introduced by Hopcroft and Tarjan to compute the Tutte-decomposition in linear time \cite{Hopcroft-Tarjan-2-sep} distinguishes between two types of 2-separations. Our algorithm will also consider these two types of separations separately, although they will handled differently from the Hopcroft-Tarjan algorithm. We first define these two types and prove that every 2-separation of a 2-connected graph is of one of these two types. We then introduce the notion of stability, before characterizing half-connected 2-separations of both types.

\subsection{A typography of 2-separations} \label{subsec:typo-2-sep}

In this section, we present the characterization of the 2-separations of a 2-connected graph due to Hopcroft and Tarjan~\cite{Hopcroft-Tarjan-2-sep}. We shall rephrase some of the statements to make them more suitable for later use.

We start with a simple lemma that we will use repeatedly.

\begin{lemma} \label{lem:follow-leftmost}
    Let $G$ be a 2-connected graph equipped with a normal spanning tree $T$ and a numbering of its vertices that is compatible with $T$.
    For every vertex $v \in V(G)$, there is a leftmost descendant of $v$ which is adjacent to $lwpt(v)$.
\end{lemma}

\begin{proof}
    We prove it by induction on $P_v$, the path rooted at $v$ which contains the leftmost descendants of $v$.
    If $v$ is adjacent to $lwpt(v)$ then the property holds for $P_v$. 
    If not, there is a descendant $x$ of $v$ which is adjacent to $lwpt(v)$. Let $c$ be the unique child of $v$ that is an ancestor of $x$. Note that $lwpt(c) \leq N_{\min}(x) \leq lwpt(v)$.
    Let $w$ be the left child of $v$. Since the numbering of the vertices is compatible with $T$, then $lwpt(w) \leq lwpt(c) \leq lwpt(v)$ by \cref{obs:numbering}. However, we also have $lwpt(v) \leq lwpt(w)$ since $v$ is an ancestor of $w$, so $lwpt(w) = lwpt(v)$. By the induction hypothesis applied to $P_w$, there is a leftmost descendant of $w$ (hence of $v$) that is adjacent to $lwpt(w) = lwpt(v)$.
\end{proof}

\begin{lemma}[\protect{\cite[Lemma 12]{Hopcroft-Tarjan-2-sep}}] \label{lem:a-b-comp}
    Let $G$ be a 2-connected graph, $T$ a normal spanning tree of $G$. If $\{u, v\}$ is a 2-separator of $G$ then $u$ and $v$ are comparable in $T$.
\end{lemma}

\begin{proof}
    Let $r$ denote the root of $T$.
    If $r \in \{u, v\}$, the statement holds so we can assume $r \notin \{u, v\}$. 
    By contradiction, assume that $u$ and $v$ are not comparable.
    Let $T_1, \ldots, T_k$ be the connected components of $T - \{u, v\}$. 
    Without loss of generality, we can assume $r \in V(T_1)$. 
    The other $T_i$ are the subtrees of $T$ rooted in the children of $u$ and of $v$, and $T_1$ contains all the vertices which are not descendants of $u$ or $v$. 
    Consider any $T_i$ for $i \geq 2$ and denote its root by $r_i$. 
    Up to renaming $u$ and $v$, we can assume that $p(r_i) = u$. 
    Since $u$ and $v$ are not comparable then $v$ is not a descendant of $r_i$ so $T_i = T(r_i)$.
    Since $G$ is 2-connected, $\{u\}$ is not a separator, so $lwpt(r_i) \neq u$ by \cref{lem:charac-1-sep}. 
    However, $(r_i, u) \in E(G)$ so $lwpt(r_i)$ is an ancestor of $u$.
    Therefore, $lwpt(r_i)$ is a proper ancestor of $u$, so $lwpt(r_i) \in V(T_1)$. 
    By definition, there is a descendant $x$ of $r_i$ that is adjacent to $lwpt(r_i)$.
    Thus, $x$ is a vertex of $T(r_i) = T_i$ and is adjacent to $lwpt(r_i) \in V(T_1)$. 
    Thus, all vertices of $T_1$ and of $T_i$ are in the same component of $G - \{u, v\}$. 
    Since this holds for every $i \geq 2$, it follows that $G - \{u, v\}$ is connected, a contradiction. 
    Thus, $u$ and $v$ are comparable.
\end{proof}

We are now ready to define the two types of 2-separations of a 2-connected graph, following the names given in \cite{Hopcroft-Tarjan-2-sep}.

\begin{definition} \label{def:t1-t2}
    Let $G$ be a 2-connected graph equipped with a normal spanning tree $T$ with root $r$. Let $(A, B)$ be a 2-separation of $G$ with separator $\{a, b\}$.
    We say that the separation $(A, B)$ is \emph{type-2} if the following holds.
    \begin{itemize}
        \item $r \notin \{a, b\}$.
        \item $V(T(a, b)) \neq \emptyset$.
        \item Up to exchanging $A$ and $B$, we have $r \in A \setminus B$ and $V(T(a, b)) \subseteq B \setminus A$.
    \end{itemize}
    Otherwise, $(A, B)$ is \emph{type-1}.
\end{definition}

\subsection{Stability}

We now introduce the notion of \emph{stability}, a key concept to our algorithm which will be extremely useful at various points in the algorithm. The next definition is supported by \cref{fig:stability}.

\begin{definition}\label{def:stability}
    Let $G$ be a graph with normal spanning tree $T$. Suppose that the vertices of $G$ are numbered with a numbering compatible with $T$.
    Let $a < b \in V(G)$ be such that $T[a, b]$ is a leftmost path. Let $a'$ be the child of $a$ such that $b$ is a leftmost descendant of $a'$.
    We say that $T[a, b]$ is \emph{stable} if every back-edge $(x, y)$ with $a' \leq x < b$ satisfies $y \geq a$.
    If $T[a, b]$ is not a leftmost path, then $T[a, b]$ is not stable.

    If $a \leq b$ and $T[a, b]$ is a leftmost path, we denote by $stab(a, b)$ the largest ancestor $\alpha$ of $a$ such that $T[\alpha, b]$ is stable if there is one, and we set $stab(a, b) = 0$ otherwise. In particular, if $T[a, b]$ is stable then $stab(a, b) = a$.
\end{definition}

\begin{figure}[h]
    \centering
    \includegraphics[width=0.25\linewidth]{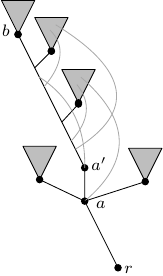}
    \caption{Every back-edge $(x,y)$ with $a' \leq x < b$ satisfies $y \geq a$. Hence, the path $T[a,b]$ is stable according to \cref{def:stability}.}
    \label{fig:stability}
\end{figure}

\begin{definition}
    Let $G$ be a graph with normal spanning tree $T$. Suppose that the vertices of $G$ are numbered with a numbering compatible with $T$.
    If $v \in V(G)$ has at least 2 children in $T$, and $w$ is its second largest child, we define $witn(v) = \min (N_{\min}(v), lwpt(w))$. Otherwise, we set $witn(v) = N_{\min}(v)$.
\end{definition}

 These $witn(\cdot)$ \emph{witness} the stability or non-stability of leftmost paths, as follows.

\begin{lemma} \label{lem:witn-stab}
    Let $G$ be a graph with normal spanning tree $T$. Suppose that the vertices of $G$ are numbered with a numbering compatible with $T$.
    Let $T[a, b]$ be a leftmost path in $T$ with $a < b$.
    Then, $T[a, b]$ is stable if and only if $witn(v) \geq a$ for every $v \in T(a, b)$.
\end{lemma}

\begin{proof}
    We denote by $a'$ the child of $a$ such that $b$ is a leftmost descendant of $a'$.
    Suppose first that $witn(v) \geq a$ for every $v \in T(a, b)$.
    By contradiction, suppose that $T[a, b]$ is not stable. Then, there exists a back-edge $(x, y)$ with $a' \leq x < b$ and $y < a$.
    Note that $witn(x) \leq N_{\min}(x) \leq y < a$ so $x \notin T(a, b)$. 
    Since $a' \leq x < b$ and $b$ is a descendant of $a'$ then $x$ is a descendant of $a'$ by \cref{obs:numbering}.
    Let $v$ be the largest vertex of $T(a, b)$ that is an ancestor of $x$ and let $u$ be the unique child of $v$ that is an ancestor of $x$ ($u$ exists since $x \notin T(a, b)$). 
    Then, $lwpt(u) \leq y < a$.
    By maximality of $v$, we have $u \notin T(a, b)$. 
    Furthermore, $u$ is an ancestor of $x$ so $u \leq x < b$. Thus, $u$ is not the left child of $v$. Let $w$ be the second largest child of $v$. 
    Note that $u$ and $w$ are siblings and $u \leq w$, so $lwpt(w) \leq lwpt(u) < a$ by \cref{obs:numbering}.
    Thus, $witn(v) \leq lwpt(w) < a$, a contradiction.

    Suppose now that $witn(v) < a$ for some $v \in T(a, b)$. If $witn(v) = N_{\min}(v)$ then there is a back-edge $(v, N_{\min}(v))$ with $a' \leq v < b$ and $N_{\min}(v) = witn(v) < a$ so $T[a, b]$ is not stable. 
    If $witn(v) = lwpt(w)$ where $w$ is the second largest child of $v$, let $v' \leq b$ be the left child of $v$.
    By definition of $lwpt(w)$, there is a descendant $w'$ of $w$ that is adjacent to $lwpt(w)$. By \cref{obs:numbering}, $v < w \leq w' < v' \leq b$ so $a' \leq w' < b$ and $w'$ is adjacent to $lwpt(w) = witn(v) < a$. Thus, $T[a, b]$ is not stable.
\end{proof}

We now state the characterization of 2-separations of a 2-connected graph from Hopcroft and Tarjan, which we rephrase using stability. See \cref{fig:types} for an illustration.

\begin{restatable}{proposition}{charactwosep}(\protect{\cite[Lemma 13]{Hopcroft-Tarjan-2-sep}}) \label{prop:charac-2-sep}
    Let $G$ be a 2-connected graph, and $T$ be a normal spanning tree of $G$ with root $r$. Suppose that the vertices of $G$ are numbered with a numbering compatible with $T$. 
    Then, $(A, B)$ is a 2-separation of $G$ with separator $\{a < b\}$ if and only if one of the following holds.
     
    \begin{enumerate}
        \item There exists a nonempty set $\mathcal{C}$ of children of $b$ such that, up to exchanging $A$ and $B$, we have $B = \{a, b\} \cup \bigcup_{c \in \mathcal{C}} Desc(c)$ and $A = V(G) \setminus (B \setminus \{a, b\})$. Furthermore, for every $c \in \mathcal{C}$, we have $lwpt(c) = a$, $lwpt_2(c) = b$ ; and some vertex $v \notin \{a,b\}$ is not a descendant of any $c \in \mathcal{C}$.
        \item There exists a child $a'$ of $a$ such that $b$ is a proper leftmost descendant of $a'$, and there exists a set $\mathcal{C}$ of children of $b$ such that, if $T_2$ denotes the connected component of $T - \{a, b\}$ containing $a'$, up to exchanging $A$ and $B$, we have ${B = \{a, b\} \cup V(T_2) \cup \bigcup_{c \in \mathcal{C}} Desc(c)}$ and $A = V(G) \setminus (B \setminus \{a, b\})$. Furthermore, $a \neq r$  ; $T[a, b]$ is stable ; and for every $c \in \mathcal{C}$, we have $lwpt(c) \geq a$ and for every child $d$ of $b$ which is not in $\mathcal{C}$, we have $hgpt(d) \leq a$. 
    \end{enumerate}
    Furthermore, $(A, B)$ is type-1 in the first case and type-2 in the second.
\end{restatable}

The proof of \cref{prop:charac-2-sep} is merely a reformulation of the proof of Hopcroft and Tarjan with our definitions. We defer it to \cref{sec:proof-charac-2-sec}. We will later need the following result, whose proof we also defer to \cref{sec:proof-charac-2-sec}.

\begin{restatable}{lemma}{cctypetwo}\label{lem:cc-type-2}
    Let $G$ be a 2-connected graph and $T$ be a normal spanning tree of $G$ with root $r$.
    Suppose that the vertices of $G$ are numbered with a numbering compatible with $T$. 
    Let $a, b \in V(G)$ and let $\mathcal{D}$ be the set of children $d$ of $b$ such that $lwpt(d) = a$ and $lwpt_2(d) = b$. Suppose that the following holds.
    \begin{itemize}
        \item there exists a child $a'$ of $a$ such that $b$ is a proper leftmost descendant of $a'$;
        \item $a \neq r$;
        \item $T[a, b]$ is stable;
        \item Every child $c$ of $b$ satisfies either $lwpt(c) \geq a$ or $hgpt(c) \leq a$. 
    \end{itemize}
    Then, the connected components of $G - \{a, b\}$ are exactly the following (and they are pairwise distinct).
    \begin{itemize}
        \item $G[Desc(d)]$ for every $d \in \mathcal{D}$;
        \item The connected component that contains $a'$;
        \item The connected component that contains $r$.
    \end{itemize}
    
\end{restatable}

\subsection{Half-connected 2-separations}

To construct the Tutte-decomposition, we do not need to consider all 2-separations, but only those that are totally-nested. By \cref{lem:tot-nested-half-con}, every totally-nested 2-separation is half-connected.
In this section, we give a characterization of the half-connected 2-separations of a 2-connected graph.

We start with type-1 separations.

\begin{proposition} \label{prop:charac-t1-half-con}
    Let $G$ be a 2-connected graph and $T$ be a normal spanning tree of $G$. Suppose that the vertices of $G$ are numbered with a numbering compatible with $T$.
    Let $a, b \in V(G)$ be such that $a$ is a proper ancestor of $b$, and let $a' \in V(G)$ be the only child of $a$ which is an ancestor of $b$. Let $\mathcal{D}$ be the set of children $d$ of $b$ such that $lwpt(d) = a$ and $lwpt_2(d) = b$. Let $\alpha$ be the smallest element of $\mathcal{D}$ and $\omega$ be the maximum vertex which is the descendant of some $d \in \mathcal{D}$.
    Then $(A, B)$ is a half-connected type-1 separation of $G$ with separator $\{a, b\}$ if and only if one of the following holds:

    \begin{itemize}
        \item There exists $d \in \mathcal{D}$ such that, up to exchanging $A$ and $B$, we have $B = \{a, b\} \cup Desc(d)$ and $A = V(G) \setminus (B \setminus \{a, b\})$, and some vertex $v \notin \{a, b\}$ is not a descendant of $d$.
        \item Up to exchanging $A$ and $B$, we have $B = \{a, b\} \cup \bigcup_{d \in \mathcal{D}} Desc(d)$ and $A = V(G) \setminus (B \setminus \{a, b\})$, $|\mathcal{D}| > 1$, some vertex $v \notin \{a, b\}$ is not a descendant of any $d \in \mathcal{D}$ and either $a'=b$ or $a=1$ or some vertex $u \in [a', \alpha-1] \setminus \{b\}$ has a neighbor in $G$ in $[1, a-1] \cup [\omega + 1, n]$.
    \end{itemize}
\end{proposition}


\begin{proof}
    Suppose first that $(A, B)$ is a half-connected type-1 separation of $G$ with separator $\{a, b\}$.
    By \cref{prop:charac-2-sep}, there exists a nonempty subset $\mathcal{C}$ of $\mathcal{D}$ such that, up to exchanging $A$ and $B$, we have ${B = \{a, b\} \cup \bigcup_{c \in \mathcal{C}} Desc(c)}$ and $A = V(G) \setminus (B \setminus \{a, b\})$, and some vertex $v \notin \{a, b\}$ is not a descendant of any $c \in \mathcal{C}$.
    By definition of $\mathcal{D}$, for every vertex $d \in \mathcal{D}$, we have that $G[Desc(d)]$ is a connected component of $G - \{a, b\}$.

    If $|\mathcal{C}| = 1$, there exists $d \in \mathcal{D}$ such that $B = \{a, b\} \cup Desc(d)$ and $A = V(G) \setminus (B \setminus \{a, b\})$ and we are done.
    If $|\mathcal{C}| \geq 2$ then $G[B \setminus A]$ is not connected, so $G[A \setminus B]$ is connected since $(A, B)$ is half-connected.
    If there exists a vertex $d \in \mathcal{D}$ such that $d \in A \setminus B$ then $G[A \setminus B]$ is the connected component of $d$ in $G - \{a, b\}$, so $A \setminus B = Desc(d)$. Thus, we have $A = \{a, b\} \cup Desc(d)$ and $B = V(G) \setminus (A \setminus \{a, b\})$, and since $B \neq \emptyset$ then some vertex $v' \notin \{a, b\}$ is not a descendant of $d$, so we are also done.
    Otherwise, we have $\mathcal{D} \subseteq B \setminus A$, so $\mathcal{D} = \mathcal{C}$ and $|\mathcal{D}| > 1$. Then, $B = \{a, b\} \cup \bigcup_{d \in \mathcal{D}} Desc(d)$ and $A = V(G) \setminus (B \setminus \{a, b\})$.
    If $a'=b$ or $a=1$ then we are done, so suppose that $a'$ is a proper ancestor of $b$ and $a \neq 1$.
    In that case, the root $r$ of $T$ lies in $A \setminus B$, and $a'$ as well.
    Since $G[A \setminus B]$ is connected then there is a path between $a'$ and $r$ in $G - \{a, b\}$.
    Let $T_{a'}$ be the connected component of $a'$ in $T - \{a, b\}$ and define $T_r$ analogously.
    Observe that $V(T_{a'}) = [a', b-1] \cup [b + ND(b), a' + ND(a')-1]$.
    If $[b + ND(b), a' + ND(a')-1] \neq \emptyset$ then $T_{a'}$ has a vertex $w \in [\omega + 1, n]$ so some vertex $u \in [a', \alpha-1] \setminus \{b\}$ has a neighbor in $G$ in $[1, a-1] \cup [\omega + 1, n]$. To see it, consider the vertex $u \in [a', b-1] \cap T[a', w]$ which is closest to $w$ in $T[a', w]$.
    Thus, we can assume that $V(T_{a'}) = [a', b-1]$.
    Since $T$ is a normal spanning tree and $a' \in T(a, b)$ then one of the following holds: \begin{itemize}
        \item Some vertex $w \in T_{a'}$ is adjacent in $G$ to some vertex $w' \in T_r$. 
        Then, $w \in [a', b-1] \subseteq [a', \alpha - 1]$ and $w'$ is a proper ancestor of $a$, so $w' \in [1, a-1]$.
        \item There is a child $c$ of $b$ in $T$ such that $c \in A$ and there is an edge between $Desc(c)$ and $T_{a'}$, and an edge between $Desc(c)$ and $T_r$.
        If $c \in [a', \alpha - 1]$ then $Desc(c) \subseteq [a', \alpha - 1]$ and the edge between $Desc(c)$ and $T_r$ is an edge between a vertex in $[a', \alpha - 1]$ and a vertex in $[1, a-1]$.
        Otherwise, since $c \in A$ then $c \notin \mathcal{D}$ so $c \geq \omega + 1$ since the elements of $\mathcal{D}$ are consecutive children of $b$.
        Then, the edge between $T_{a'}$ and $Desc(c)$ is an edge between $[a', b-1] \subseteq [a', \alpha - 1]$ and $[\omega + 1, n]$.
    \end{itemize}

    For the converse implication, we handle the two cases separately.
    Suppose first that there exists $d \in \mathcal{D}$ such that, up to exchanging $A$ and $B$, we have $B = \{a, b\} \cup Desc(d)$ and $A = V(G) \setminus (B \setminus \{a, b\})$, and some vertex $v \notin \{a, b\}$ is not a descendant of $d$.
    Then, $(A, B)$ is a type-1 separation by \cref{prop:charac-2-sep}, and is half-connected since ${G[B \setminus A] = G[Desc(d)]}$ is connected.

    Suppose now that up to exchanging $A$ and $B$, we have $B = \{a, b\} \cup \bigcup_{d \in \mathcal{D}} Desc(d)$ and $A = V(G) \setminus (B \setminus \{a, b\})$, $|\mathcal{D}| > 1$, some vertex $v \notin \{a, b\}$ is not a descendant of any $d \in \mathcal{D}$ and either $a'=b$ or $a=1$ or some vertex $u \in [a', \alpha-1] \setminus \{b\}$ has a neighbor in $G$ in $[1, a-1] \cup [\omega + 1, n]$.
    Again, we have that $(A, B)$ is a type-1 separation by \cref{prop:charac-2-sep}.
    We now argue that $G[A \setminus B]$ is connected.
    If $a' \neq b$, let $T_{a'}$ be the connected component of $a'$ in $T - \{a, b\}$, and define $T_r$ analogously if $r \neq a$.
    Note that every connected component of $T - \{a, b\}$ is of the form $T_{a'}, T_r$, $T[Desc(a'')]$ for some sibling $a''$ of $a$, or $T[Desc(c)]$ for some child $c$ of $b$.
    First, consider a child $c$ of $b$ such that $c \notin \mathcal{D}$.
    Since $G$ is 2-connected then $lwpt(c) < b$ by \cref{lem:charac-1-sep}.
    If $lwpt(c) \neq a$ then, since $T$ is a normal spanning tree, there is an edge between $Desc(c)$ and either $T_{a'}$ or $T_r$.
    If $lwpt(c) = a$ then $lwpt_2(c) \neq b$ since $c \notin \mathcal{D}$ so there is an edge between $Desc(c)$ and $T_{a'}$.
    Furthermore, if $c \geq \omega + 1$ then, since the numbering is compatible with $T$ we have $lwpt(c) < a$ so there is an edge between $Desc(c)$ and $T_r$.
    Next, consider a sibling $a''$ of $a'$.
    Since $a$ has at least two children then $a \neq r$ by \cref{lem:charac-1-sep}.
    Again by \cref{lem:charac-1-sep}, we have $lwpt(a'') < a$, so there is an edge in $G$ between $Desc(a'')$ and $T_r$.
    Thus, to prove that $G[A \setminus B]$ is connected, it suffices to prove that $T_{a'}$ and $T_r$ are connected in $G - \{a, b\}$, or that one of them is not defined. 

    If $a'=b$ then $T_{a'}$ is not defined, so we are done, and similarly if $a = 1$ then $T_{r}$ is not defined, so we are done.
    Thus, we can assume that some vertex $u \in [a', \alpha-1] \setminus \{b\}$ has a neighbor $w \in [1, a-1] \cup [\omega + 1, n]$ in $G$.
    We first argue that the connected component $T_u$ of $u$ in $T - \{a, b\}$ is connected to $T_{a'}$ in $G - \{a, b\}$.
    If $u \in [a', b-1]$ then $u \in T_{a'}$ so $T_u = T_{a'}$, so suppose that $u \in [b+1, \alpha - 1]$.
    Then, there exists a child $c$ of $b$ such that $u \in Desc(c)$, and $c < \alpha$.
    Thus, $c \notin \mathcal{D}$, and since the numbering is compatible with $T$ and $c < \alpha$ then one of $lwpt(c), lwpt_2(c)$ is in $T(a, b)$. This means that there is an edge in $G$ between $Desc(c)$ and $T(a, b)$, so $T_u$ and $T_{a'}$ are connected in $G - \{a, b\}$.
    We now argue that the connected component $T_w$ of $w$ in $T-\{a, b\}$ is connected to $T_r$ in $G - \{a, b\}$.
    We can assume that $w \notin T_r$, so $w \in [\omega + 1, n]$.
    If $T_w$ is of the form $T[Desc(a'')]$ for some sibling $a''$ of $a'$ then we already showed that there is an edge between $T_w$ and $T_r$.
    If $T_w$ is of the form $T[Desc(c)]$ for some child $c$ of $b$ then $c \geq \omega$ so we already showed that there is an edge between $T_w$ and $T_r$.
    Otherwise, we have $T_w = T_{a'}$, so $w$ is a descendant of $a'$ and not of $b$, and $w > \omega$. Thus, $LD(b) < LD(a')$ so $b$ is not a leftmost descendant of $a'$. However, by \cref{lem:follow-leftmost}, there exists a leftmost descendant $x$ of $a'$ which is adjacent to $lwpt(a')$. Then, we have $x \in T_{a'}$. However, since $a \neq 1$ and $G$ is 2-connected then $lwpt(a') < a$ by \cref{lem:charac-1-sep}. Thus, $lwpt(a') \in T_r$ and is adjacent to $x \in T_{a'}$.
    We proved that $T_u$ is connected to $T_{a'}$ in $G- \{a, b\}$ and that $T_w$ is connected to $T_r$ in $G - \{a, b\}$.
    Thus, $T_{a'}$ and $T_r$ are connected in $G - \{a, b\}$, so $G[A \setminus B]$ is connected.
\end{proof}

We now consider type-2 separations.

\begin{proposition} \label{prop:charac-t2-half-con}
    Let $G$ be a 2-connected graph and $T$ be a normal spanning tree of $G$ with root $r$. Suppose that the vertices of $G$ are numbered with a numbering compatible with $T$.
    Let $a < b \in V(G)$ and let $c_1 < \ldots < c_k$ denote the children of $b$.
    Let $\mathcal{D}$ be the set of children $d$ of $b$ such that $lwpt(d) = a$ and $lwpt_2(d) = b$.
    Then, $(A, B)$ is a half-connected type-2 separation of $G$ with separator $\{a, b\}$ if and only if all of the following holds:

    \begin{enumerate}
        \item there exists a child $a'$ of $a$ such that $b$ is a proper leftmost descendant of $a'$; \label{item:i}
         \item there exists $i \in \{0, 1, \ldots, k\}$ such that, if $T_2$ denotes the connected component of~$T - \{a, b\}$ containing $a'$, up to exchanging $A$ and $B$, we have $B = \{a, b\} \cup V(T_2) \cup \bigcup_{j \leq i} Desc(c_i)$ and $A = V(G) \setminus (B \setminus \{a, b\})$. \label{item:ii}
        \item $a \neq r$; \label{item:iii}
        \item $T[a, b]$ is stable; \label{item:iv}
        \item for every $j \leq i, lwpt(c_j) \geq a$ and for every $j > i, hgpt(c_j) \leq a$; \label{item:v}
        \item there do not exist $j,k$ such that $j \leq i < k$ and $c_j, c_k \in \mathcal{D}$; \label{item:vi}
       
    \end{enumerate}
\end{proposition}

\begin{proof}
    Suppose first that $(A, B)$ is a half-connected type-2 separation of $G$, with separator $\{a, b\}$.
    By \cref{prop:charac-2-sep}, there exists a child $a'$ of $a$ such that $b$ is a proper leftmost descendant of $a'$, and there exists a set $\mathcal{C}$ of children of $b$ such that, up to exchanging $A$ and $B$, we have ${B = \{a, b\} \cup V(T_2) \cup \bigcup_{c \in \mathcal{C}} Desc(c)}$ and $A = V(G) \setminus (B \setminus \{a, b\})$. Furthermore, $a \neq r$  ; $T[a, b]$ is stable ; and for every $c \in \mathcal{C}$, we have $lwpt(c) \geq a$ and for every child $d$ of $b$ which is not in $\mathcal{C}$, we have $hgpt(d) \leq a$. 
    Thus, it only remains to prove that there exists $i \in \{0, 1, \ldots, k\}$ such that $\mathcal{C} = \{c_1, \ldots, c_i\}$, and \cref{item:vi}.
    
    By \cref{lem:cc-type-2}, the connected components of $G - \{a, b\}$ are exactly the following: the component containing $a'$, the component containing $r$, and for every $d \in \mathcal{D}$, the component $G[Desc(d)]$.
    Furthermore, we have $a' \in B \setminus A$ and $r \notin B$ so $r \in A \setminus B$.
    Since $(A, B)$ is half-connected, the other connected components of $G - \{a, b\}$ all live on the same side of the separation. Thus, either $\mathcal{D} \subseteq A \setminus B$ or $\mathcal{D} \subseteq B \setminus A$.
    
    We now prove that there exists $i \in \{0, 1, \ldots, k\}$ such that $\mathcal{C} = \{c_1, \ldots, c_i\}$.
    If $\mathcal{C} = \emptyset$ it suffices to set $i = 0$. If $\mathcal{C} \neq \emptyset$, let $i$ be maximal such that $c_i \in \mathcal{C}$. By \cref{prop:charac-2-sep}, we have $lwpt(c_i) \geq a$. We consider two cases.
    \begin{itemize}
        \item If $c_i \in \mathcal{D}$ then $\mathcal{D} \subseteq B \setminus A$. 
        Let $j < i$ and consider $c_j$. By \cref{obs:numbering}, we have $lwpt(c_j) \geq lwpt(c_i) = a$. If $c_j \in \mathcal{D}$ then $c_j \in B$ and thus $c_j \in \mathcal{C}$. If $c_j \notin \mathcal{D}$, either $lwpt(c_j) > a$ or $lwpt_2(c_j) < b$. 
        In both cases, we have $hgpt(c_j) > a$ and therefore $c_j \in \mathcal{C}$ by \cref{prop:charac-2-sep}.
        \item If $c_i \notin \mathcal{D}$, let $j < i$ and consider $c_j$. Since the numbering of the vertices is compatible with $T$, we have that ${(lwpt(c_j), lwpt_2(c_j)) \leq' (lwpt(c_i), lwpt_2(c_i))}$. By definition of $\leq'$, either we have $lwpt(c_j) > lwpt(c_i)$, or we have $lwpt(c_j) = lwpt(c_i)$ and $lwpt_2(c_j) \leq lwpt_2(c_i)$. 
        Since $c_i \notin \mathcal{D}$ then $(lwpt(c_i), lwpt_2(c_i)) \neq (a, b)$, so one of them is in $T(a, b)$ since $lwpt(c_i) \geq a$. Thus, one of $lwpt(c_j), lwpt_2(c_j)$ is in $T(a, b) \subseteq V(T_2) \subseteq B \setminus A$. This implies $hgpt(c_j) > a$, so $c_j \in \mathcal{C}$ by \cref{prop:charac-2-sep}. 
    \end{itemize}
    In both cases, we indeed have $\mathcal{C} = \{c_1, \ldots , c_i\}$.
    Finally, since either $\mathcal{D} \subseteq A \setminus B$ or $\mathcal{D} \subseteq B \setminus A$ then there does not exist $j \leq i < k$ such that $c_j, c_k \in \mathcal{D}$.

    Suppose now that \cref{item:i,item:ii,item:iii,item:iv,item:v,item:vi} hold. 
    Setting $\mathcal{C} = \{c_1, \ldots, c_i\}$, \cref{prop:charac-2-sep} implies that $(A, B)$ is a type-2 separation of $G$ with separator $\{a, b\}$.
    Again, by \cref{lem:cc-type-2}, the connected components of $G - \{a, b\}$ are exactly the following: the component containing $a'$, the component containing $r$, and for every $d \in \mathcal{D}$, the component $G[Desc(d)]$.
    However, by \cref{item:vi}, either $\mathcal{D} \subseteq A \setminus B$ or $\mathcal{D} \subseteq B \setminus A$.
    In the first case, $G[B \setminus A]$ is the connected component of $G - \{a, b\}$ containing $a'$, and thus is connected.
    In the second case, $G[A \setminus B]$ is the connected component of $G - \{a, b\}$ containing $r$, and thus is connected.
    In both cases, $(A, B)$ is indeed half-connected.
\end{proof}

\begin{corollary} \label{cor:cyclic-intervals}
    Let $G$ be a 2-connected graph and $T$ be a normal spanning tree of $G$. Suppose that the vertices of $G$ are numbered with a numbering compatible with $T$.
    If $(A, B)$ is a half-connected 2-separation of $G$ then $A \setminus B$ and $B \setminus A$ form cyclic intervals of $[n] \setminus (A \cap B)$.
\end{corollary}

\begin{proof}
    Let $A \cap B = \{a < b\}$ and let $\mathcal{D}$ be the set of children $d$ of $b$ such that $lwpt(d) = a$ and $lwpt_2(d) = b$.
    Since the numbering is compatible with $T$, by \cref{dfn:compatible-numbering}~\ref{item:order-siblings}, the elements of $\mathcal{D}$ form an interval among the children of $b$.
    
    If $(A, B)$ is type-1, it then follows immediately from \cref{prop:charac-t1-half-con} and the fact that the numbering is compatible with $T$ that $B \setminus A$ is an interval of $[n] \setminus \{a, b\}$, and therefore $A \setminus B$ is also a cyclic interval of $[n] \setminus \{a, b\}$.
    
    If $(A, B)$ is type-2, we use \cref{prop:charac-t2-half-con}. With the same notations as in \cref{prop:charac-t2-half-con}, there exists $i \in \{0, 1, \ldots, k\}$ such that ${B = \{a, b\} \cup V(T_2) \cup \bigcup_{j \leq i} Desc(c_i)}$.
    Observe then that $V(T_2) = [a', b)$ in the numbering and that we have ${\bigcup_{j \leq i} Desc(c_i) = (b, c_i + ND(c_i) - 1]}$ if $i > 0$, and is empty otherwise.
    In both cases, $B \setminus A$ is an interval of $[n] \setminus \{a, b\}$, and therefore $A \setminus B$ is also a cyclic interval of $[n] \setminus \{a, b\}$.
\end{proof}

\begin{remark} \label{rem:compatible}
    It now follows immediately from \cref{cor:cyclic-intervals,,lem:tot-nested-half-con} that if $G$ is 2-connected and $T$ is a normal spanning tree of $G$ then every numbering of $V(G)$ that is compatible with $T$ is also compatible with $G$.
\end{remark}

\section{Computing half-connected 2-separations} \label{sec:4}

The goal of this section is to prove the following.

\begin{theorem} \label{thm:compute-pot-nested}
    There is an $O(n + m)$-time algorithm which, given a 2-connected graph $G$, computes a set $\mathcal{S}$ of half-connected 2-separations of $G$ that contains all totally-nested separations.
\end{theorem}

Note that we cannot hope to compute all half-connected 2-separations of $G$ in time $O(n + m)$ since there could be quadratically many: if $G$ is a cycle then every separation is a half-connected type-2 separation.

We start by mentioning some algorithmic results that we will need for various computations. We then describe in detail the precomputation that we perform on a given input graph. We next focus on stability and show that several queries related to stability can be answered in constant time after a linear-time precomputation. Finally, we explain how to compute all half-connected type-1 separations and a set of half-connected type-2 separations containing all the totally-nested ones.

\subsection{Algorithmic tools}

First, we consider the problem of sorting a list $L$ of $n$ integers. It is well-known that every comparison-based sorting algorithm runs in $\Omega(n\log(n))$. However, if we know that all entries in $L$ are small then we can sort $L$ in linear time using \emph{counting sort}.

\begin{thm}[\cite{Cormen}] \label{thm:linear-sort}
    There is an $O(n + k)$-time algorithm which, given an integer $k$ and a list $L$ of $n$ integers in the range between $-k$ and $k$, sorts $L$.
\end{thm}

Counting sort is a \emph{stable sort}. 
This means that if several entries of $L$ are equal then their relative order after the sort is the same as their relative order before the sort.
Thus, if $L$ is a list of length $n$ of $t$-tuples, which we want to sort lexicographically, we can do it using counting sort, as follows.
We start by sorting the elements of $L$ according to their last coordinate with counting sort, then according to their penultimate coordinate and so on, before finally sorting them according to their first coordinate.
After these $t$ steps of counting sort, $L$ will be sorted lexicographically.

\begin{corollary} \label{cor:sort-lex}
    There is an $O(t \cdot (n + k))$-time algorithm which, given an integer $k$ and a list $L$ of length $n$ of $t$-tuples such that every entry of every tuple is in the range between $-k$ and $k$, sorts $L$ lexicographically.
\end{corollary}

The following two results essentially say that computing all lowpoints and all highpoints can be done in linear time. 

\begin{lemma}[\cite{Kosinas23}] \label{lem:precomp-lwpts}
    There is an $O(k \cdot (n+m))$-time algorithm which, given an integer $k$ and a pair $(G, T)$ where $T$ is a normal spanning tree of a graph $G$, computes the first $k$ lowpoints of all vertices of $G$.
\end{lemma}

\begin{lemma}[\cite{Kosinas23}] \label{lem:precomp-hgpts}
    There is an $O(n+m)$-time algorithm which, given a pair $(G, T)$ where $T$ is a normal spanning tree of a graph $G$, computes the highpoint of all vertices of $G$.
\end{lemma}

The computation of all first $k$ lowpoints can be done easily by dynamic programming, starting from the leaves of $T$. However, the computation of all highpoints is more involved and the algorithm given by Kosinas in \cite{Kosinas23} relies on a variant of the Union-Find structure where the so-called Union Tree is known beforehand.

\begin{lemma} \label{lem:compute-numbering}
    There is an $O(n + m)$-time algorithm which, given a pair $(G, T)$ where $T$ is a normal spanning tree of a graph~$G$, computes a numbering of the vertices of $G$ which is compatible with $T$.
\end{lemma}

\begin{proof}
    First, perform a DFS of $T$ and order the vertices according to the first time they are encountered in the DFS.
    Let $\prec$ denote the corresponding linear order on $V(G)$. Assume that every $v \in V(G)$ stores $rank(v)$, its rank in the linear order $\prec$.
    Observe that if $u$ is a proper ancestor of $v$ then $u \prec v$.
    By \cref{lem:precomp-lwpts}, we can compute the first 2 lowpoints of all vertices in time $O(n + m)$.
    For every vertex $v$, let $seq(v) = (- rank(lwpt_1(v)), rank(lwpt_2(v)))$.
    By \cref{cor:sort-lex}, we can sort lexicographically the set of all $seq(v)$ in time $O(n)$.
    
    Then, in linear time overall, every vertex can sort its children by decreasing $seq(\cdot)$: traverse the sorted list of all $seq(v)$ and insert each vertex into the list of its parent.
    Finally, perform a DFS of $T$ starting from the same root, and satisfying the following: for every vertex $u$, the children of $u$ are visited by increasing values of $seq(\cdot)$.
    We number the vertices of $G$ from 1 to $n$ according to the order in which they are first visited during this DFS.

    This numbering trivially satisfies \cref{dfn:compatible-numbering}~\ref{item:numbered1-n}.
    Every vertex $j$ has exactly $ND(j)$ descendants and any DFS of $T$ visits all descendants of $j$ consecutively, starting from $j$, which proves \cref{dfn:compatible-numbering}~\ref{itm:compatible-2}.
    If $j < k$ are siblings then $j < k$ means that $j$ is visited before $k$ among the children of their parent so $seq(j) \leq seq(k)$ in the lexicographical order. 
    If $(lwpt_1(j), lwpt_2(j)) = (lwpt_1(k), lwpt_2(k))$ then \cref{dfn:compatible-numbering}~\ref{itm:compatible-3} holds.
    If $lwpt_1(j) \neq lwpt_1(k)$ then $lwpt_1(k) \prec lwpt_1(j)$. 
    Since $lwpt_1(j), lwpt_1(k)$ are ancestors of $p(j) = p(k)$, they are comparable in $T$ and therefore $lwpt_1(k)$ is a proper ancestor of $lwpt_1(j)$. 
    Thus, $lwpt_1(k) < lwpt_1(j)$ by the properties of a DFS and thus ${(lwpt_1(j), lwpt_2(j)) \leq' (lwpt_1(k), lwpt_2(k))}$, so \cref{dfn:compatible-numbering}~\ref{itm:compatible-3} holds.
    Otherwise, $lwpt_1(j) = lwpt_1(k)$ and thus $lwpt_2(j) \prec lwpt_2(k)$. If $lwpt_2(j)$ is a proper ancestor of $j$ then it is also a proper ancestor of $k$. 
    However $lwpt_2(k)$ is an ancestor of $k$ so $lwpt_2(k)$ and $lwpt_2(j)$ are comparable and $lwpt_2(j) \prec lwpt_2(k)$ implies that $lwpt_2(j)$ is a proper ancestor of $lwpt_2(k)$ which in turn implies $lwpt_2(j) < lwpt_2(k)$ by the properties of a DFS. 
    In that case, we therefore have $(lwpt_1(j), lwpt_2(j)) \leq' (lwpt_1(k), lwpt_2(k))$.
    If $lwpt_2(j)$ is not a proper ancestor of $j$ then $lwpt_2(j) = j$. Since $lwpt_2(j) \prec lwpt_2(k)$, $lwpt_2(k)$ cannot be an ancestor of $j$ so $lwpt_2(k) = k > j$. 
    In that case, we therefore also have $(lwpt_1(j), lwpt_2(j)) \leq' (lwpt_1(k), lwpt_2(k))$.
\end{proof}

\begin{lemma} \label{lem:precomp-maxhgpt}
    There is an $O(n)$-time algorithm which, given a pair $(G, T)$ where $T$ is a normal spanning tree of a graph $G$, a numbering of the vertices of $G$ and an array storing $hgpt(v)$ for every $v \in V(G)$, computes $maxhgpt(v)$ for every vertex $v \in V(G)$.
\end{lemma}

\begin{proof}
    Let $v \in V(G)$ and let $c_1 < \ldots < c_k$ be the children of $v$. Observe that $maxhgpt(c_k) = hgpt(c_k)$ and that for every $i < k$, we have $maxhgpt(c_i) = \max(hgpt(c_i), maxhgpt(c_{i+1}))$. Thus, we can compute successively $maxhgpt(c_k)$, $maxhgpt(c_{k-1}), \ldots$, up to $maxhgpt(c_1)$, in time $O(k)$. Summing over all vertices $v \in V(G)$, we indeed compute all $maxhgpt(\cdot)$ in time $O(n)$ overall.
\end{proof}

We conclude this section with two key tools for designing linear-time algorithms. 
Given a rooted tree $T$ and two vertices $u, v \in V(T)$, observe that all common ancestors of $u$ and $v$ are comparable in $T$. The \emph{least common ancestor (LCA)} of $u$ and $v$ is their unique common ancestor which is a descendant of all their common ancestors. Finding the least common ancestor of two vertices in an input tree $T$ requires $\Omega(n)$ time in general. However, if we are allowed to preprocess $T$ in \emph{linear time} beforehand then we can find the least common ancestor of any two vertices in \emph{constant time}.

Consider an array $A = A[0], A[1], \ldots, A[n-1]$ of size $n$. Then, an \emph{interval} of $A$ is any subarray of $A$ of the form ${A[x, y] := A[x], A[x+1], \ldots, A[y]}$.
The minimum of the interval $A[x, y]$ is the minimum value $v$ among all $A[i]$ for $i \in [x, y]$, and its corresponding \emph{position} is an index $i \in [x, y]$ such that $A[i] = v$.
Finding the minimum and its position in an arbitrary interval of an input array $A$ of size $n$ requires $\Omega(n)$ time in general.
However, once again if we are allowed to preprocess $A$ in \emph{linear time} beforehand then we can find the minimum and its position in any interval of $A$ in \emph{constant time}.
Interestingly, Farach-Colton and Bender \cite{bender2000LCA} proved the first result by reducing it to the second one, and the second one by reducing it to a subcase of the first one.

\begin{lemma}[\cite{bender2000LCA}] \label{lem:LCA}
    There is an algorithm which, given a rooted tree $T$ on $n$ vertices, after an $O(n)$-time precomputation, can answer in constant time LCA queries in $T$.
\end{lemma}

\begin{lemma}[\cite{bender2000LCA}] \label{lem:RMQ}
    There is an algorithm which, given an array $A$ of size $n$, after an $O(n)$-time precomputation, can return the minimum and maximum (and their positions) of any interval of $A$ in constant time.
\end{lemma}

\subsection{Precomputation} \label{subsec:precomputations}

We now describe the precomputation we perform when given a 2-connected graph $G$. 
Note that we can verify in linear time that $G$ is 2-connected using \cref{lem:algo-1-sep}.
For the remaining algorithmic sections, we will assume that all the precomputation presented in this section have already been performed.

We start by computing a normal spanning tree $T$ of $G$, rooted arbitrarily. This can be done in time $O(n + m)$ by a DFS.
Then, using \cref{lem:compute-numbering}, we compute a numbering of the vertices of $G$ that is compatible with $T$. 
Every vertex stores its number of children and an array containing the names of its children, sorted by increasing number. Every vertex $v$ also stores $sib\_rank(v)$, its rank amongst all its siblings.  
We preprocess $T$ to be able to find the least common ancestor of any two vertices in $T$ in constant time. This can be done in time $O(n)$ by \cref{lem:LCA}.
Then, we compute all first lowpoints, second lowpoints and highpoints using \cref{lem:precomp-lwpts,,lem:precomp-hgpts}. Once this is done, we can compute $maxhgpt(v)$ for every $v \in V(G)$ in time $O(n)$ overall using \cref{lem:precomp-maxhgpt}.
We also compute $N_{\min}(v)$ and $N_{\max}(v)$ for all vertices $v \in V(G)$. This can be done in time $O(n + m)$ by iterating over all edges of $G$.
Then, we apply \cref{lem:RMQ} to the arrays $N_{\min}$ and $N_{\max}$, to preprocess them in time $O(n)$ to be able to answer minimum and maximum queries on an interval in constant time.
Once this is done, we can compute $witn(v)$ for all $v \in V(G)$ in time $O(n)$ overall. 
Similarly, using dynamic programming from the leaves of $T$, we compute all values of $ND(\cdot)$.
Since the numbering is compatible with~$T$, we have $LD(v) = v + ND(v) - 1$ for every $v \in V(G)$ so we can also compute all values of $LD(\cdot)$.

Once this is done, we compute the partition of $E(T)$ into maximal leftmost paths. This can be done by a DFS in $T$, where we visit the children of every vertex by decreasing number. Let $P_1, \ldots, P_k$ denote these maximal leftmost paths. We store each $P_i = (v^i_0, \ldots, v^i_l)$ as an array $P_i$ so that $P_i[j] = v^i_j$. Since these paths form a partition of the edges of $T$, computing and storing all $P_i$ can be done in $O(n)$ time.
Observe that for every vertex $v \in [2, n]$, there is a unique path $P_{i_v}$ such that $v \in V(P_{i_v})$ and $v$ is not the first vertex of $P_{i_v}$ (namely, the path containing the edge $(v, p(v)) \in E(T)$). We denote by $r_v$ the unique integer $j$ such that $P_{i_v}[j] = v$. We compute and store all $i_v$ and $r_v$ on the fly while constructing the paths $P_i$.

We shall use the following lemma in various contexts. Informally, it says that, after a linear-time precomputation, we can answer in constant time maximum and minimum queries on any leftmost path, for any function.

\begin{lemma} \label{lem:query-leftmost-branch}
    There is an algorithm which, given a graph $G$ equipped with a spanning tree $T$ and an array storing the values of a function $f : V(G) \to [0, n]$, after an $O(n)$-time precomputation, can answer in constant time queries of the form: Given $a, b \in V(G)$, determine whether they belong to a common maximal leftmost path $P$, and if so return the minimum and the maximum of $f$ (and a corresponding vertex) over all vertices of $P[a, b]$ (or $P(a, b), P[a, b), P(a, b]$).
\end{lemma}

\begin{proof}
    We start by describing the precomputation. 
    For every path $P_i$, we create an array $f_i$ of the same size as $P_i$, and set $f_i[j] = f(P_i[j])$ for all relevant indices $j$.
    Since the $P_i$ partition the edges of $T$ then the collection of $f_i$ has linear size.
    We apply \cref{lem:RMQ} separately on each $f_i$ to be able find in constant time the minimum and the maximum (and their positions) of any interval of $f_i$. 
    This takes time $O(|f_i|)$ for every $f_i$, hence time $O(n)$ overall.
    This finishes the description of the precomputation.
    
    Suppose now that we are given a query $(a, b)$. We may assume without loss of generality that $a < b$. 
    Observe that $a$ and $b$ belong to a common maximal leftmost path if and only if $a \in P_{i_b}$.
    Furthermore, if $a \in P_{i_b}$, either $a = P_{i_b}[r_a]$ or $a = P_{i_b}[0]$.
    If this is not the case, we return that $a$ and $b$ do not belong to a common maximal leftmost path.
    Otherwise, let $r'_a \in \{0, r_a\}$ be the index such that $P_{i_b}[r'_a] = a$. 
    It now suffices to perform a minimum and maximum query over $f_i[r_a, r_b]$ to obtain the minimum and the maximum of $f$ (and their positions) over all vertices of $P[a, b]$. Given a position $p$, the corresponding vertex is simply $P_{i_b}[p]$. All these steps can be performed in constant time.
    For $P(a, b)$, simply query over $f_i[r_a + 1, r_b - 1]$, and similarly for $P(a, b]$ and $P[a, b)$.
\end{proof}

\subsection{Computing stability} \label{subsec:stability}

Stability will turn out to be a key feature in several steps of our algorithm. In this section, we show that we can answer queries related to stability in constant time, after just a linear-time precomputation.
Our main result is the following.

\begin{proposition} \label{prop:queries-stability}
    There is an algorithm which, given a graph $G$ equipped with a normal spanning tree $T$ and a numbering of its vertices which is compatible with $T$, after an $O(n+m)$-time precomputation, can compute every value $stab(a, b)$ in constant time.
\end{proposition}

Recall that if $a < b$ and $T[a, b]$ is a leftmost path then $stab(a, b)$ is the largest ancestor $\alpha$ of $a$ such that $T[\alpha, b]$ is stable, if there exists such an ancestor, and is 0 otherwise.
Given $a < b \in V(G)$, the following is a natural way to compute $stab(a, b)$.
If $T[a, b]$ is not a leftmost path of $\alpha$ then $stab(a, b) = 0$. 
Otherwise, we start by setting $\alpha = a$ and we look at all values of $witn(v)$ for $v \in T(\alpha, b)$. If $witn(v) \geq \alpha$ for all $v \in T(\alpha, b)$ then $T[\alpha, b]$ is stable by \cref{lem:witn-stab} and we return $\alpha$. Otherwise, there exists $v \in T(\alpha, b)$ such that $witn(v) = \beta < \alpha$. We then set $\alpha = \beta$ and repeat the process until it terminates.
This naive algorithm runs in quadratic time. 
There are two key observations to speed up this algorithm. First, instead of going over all values of $witn(v)$ for $v \in T(\alpha, b)$, we can use \cref{lem:query-leftmost-branch} to find the minimum in constant time. This will be formalized by \cref{lem:test-stab}. The second observation is that if we set the value of $\alpha$ to some value $\beta = witn(v)$ then the continuation of the algorithm (and thus the final value of $\alpha$) only depends on $v$, and not on $a$ and $b$. Therefore, it suffices to precompute the final value of $\alpha$ for every $v$. We will show in \cref{lem:precom-stab} that this can be done in linear time overall.

\begin{lemma} \label{lem:test-stab}
    There is an algorithm which, given a graph $G$ equipped with a normal spanning tree $T$ and a numbering of its vertices which is compatible with $T$, after an $O(n+m)$-time precomputation, can answer in constant time queries of the form: Given $a < b \in V(G)$, tell whether $T[a, b]$ is stable, and if not, either say that $T[a, b]$ is not a leftmost path or return a vertex $v \in T(a, b)$ with minimum $witn(v)$, witnessing that $T[a, b]$ is not stable.
\end{lemma}

\begin{proof}
    By \cref{lem:witn-stab}, $T[a, b]$ is stable if and only if it is a leftmost path and $witn(v) \geq a$ for every $v \in T(a, b)$.
    Since we already computed all values of $witn(\cdot)$, we can access each of them in constant time.
    Therefore, by \cref{lem:query-leftmost-branch}, after an $O(n)$-time precomputation, we can answer in constant time queries of the form: Given $a, b \in V(G)$, determine whether they belong to a common maximal leftmost path $P$, and if so return the minimum of $witn(\cdot)$ (and a corresponding vertex) over all vertices of $P(a, b)$.
    Note that $T[a, b]$ is a leftmost path if and only if $a$ and $b$ belong to a common maximal leftmost path.
    We query whether $a$ and $b$ belong to a common maximal leftmost path $P$, and if so return the minimum $\beta$ of $witn(\cdot)$ (and a corresponding vertex $v$) over all vertices of $P(a, b)$.
    If $a$ and $b$ are not in a common leftmost path, we return that $T[a, b]$ is not a leftmost path.
    If $\beta < a$, we return $v \in P(a, b)$ such that $witn(v) = \beta$.
    Otherwise, we return that $T[a, b]$ is stable.
\end{proof}

\begin{lemma} \label{lem:precom-stab}
    There is an $O(n+m)$-time algorithm which, given a graph $G$ equipped with a normal spanning tree $T$ and a numbering of its vertices which is compatible with $T$, computes all values of $stab(witn(v), v)$.
\end{lemma}

\begin{proof}
    We start by doing the $O(n+m)$-time precomputation of \ \cref{lem:test-stab}.
    Observe that $stab(witn(r), r) = 0$ since $witn(r) \geq r$.
    We then compute the $stab(witn(v), v)$ separately on all maximal leftmost paths.
    Let $P = v_0, \ldots, v_{|P|}$ be a maximal leftmost path.
    We have $witn(v_1) \leq v_0$ as $(v_0, v_1) \in E(G)$.
    If $witn(v_1) = v_0$ then $T[v_0, v_1]$ is stable by \cref{lem:witn-stab}, so $stab(witn(v_1), v_1) = v_0$.
    Otherwise, $stab(witn(v_1), v_1) = 0$ as there is no vertex $u \leq witn(v_1)$ such that $T[u, v_1]$ is a leftmost path.
    Now, suppose that we already computed $stab(witn(v_1), v_1), \ldots, stab(witn(v_{i-1}), v_{i-1})$ and we want to compute $stab(witn(v_i), v_i)$.
    If $witn(v_i) < v_0$ then $stab(witn(v_i), v_i) = 0$ for the same reason as above.
    Otherwise, $witn(v_i) = v_{i'}$ for some $i' < i$. Using \cref{lem:test-stab}, we query in constant time whether $T[v_{i'}, v_i]$ is stable.
    If it is, we store $stab(witn(v_i), v_i) = v_{i'}$.
    Otherwise, since $T[v_{i'}, v_i]$ is a leftmost path, we are given a witness $v_j$ of non-stability of $T[v_{i'}, v_i]$, with $witn(v_j)$ minimum, $i' < j < i$ and $witn(v_j) < v_{i'}$.
    
    \begin{claim*}
        $stab(witn(v_i), v_i) = stab(witn(v_j), v_j)$.
    \end{claim*}

    \begin{subproof}
        Suppose first that $stab(witn(v_j), v_j) \neq 0$ and let $v_k = stab(witn(v_j), v_j)$. Since $T[v_k, v_j]$ is stable, for every $k < l < j$, we have $witn(v_l) \geq v_k$ by \cref{lem:witn-stab}. Furthermore, $v_k = stab(witn(v_j), v_j) \leq witn(v_j)$ by definition. 
        Finally, for $j < l < i$, we have $witn(v_l) \geq witn(v_j) \geq v_k$ by minimality of $witn(v_j)$.
        Therefore, we have ${witn(v_l) \geq v_k}$ for every $k < l < i$ so $T[v_k, v_i]$ is stable by \cref{lem:witn-stab}, and $v_k = stab(witn(v_j), v_j) \leq witn(v_j) < witn(v_i)$. This proves that $stab(witn(v_i), v_i) \geq stab(witn(v_j), v_j)$.

        Conversely, suppose that $stab(witn(v_i), v_i) \neq 0$ and let $v_s = stab(witn(v_i), v_i)$.
        Then, $v_s \leq witn(v_i)$. Furthermore, by \cref{lem:witn-stab}, we have $witn(v_l) \geq v_s$ for every $s < l < i$. In particular, this implies $v_s \leq witn(v_j)$ since $s \leq i' < j < i$. 
        Thus, we have $witn(v_l) \geq v_s$ for every $s < l < j$ so $T[v_s, v_j]$ is stable and thus $stab(witn(v_i), v_i) \leq stab(witn(v_j), v_j)$.

        Therefore, if one of $stab(witn(v_i), v_i), stab(witn(v_j), v_j)$ is non-zero then the other as well and the two are equal. 
    \end{subproof}
    Hence, we can simply store $stab(witn(v_i), v_i) = stab(witn(v_j), v_j)$. 
    Computing each $stab(witn(v_i), v_i)$ takes constant time so the total running time is indeed linear.
\end{proof}


\begin{proof}[Proof of \cref{prop:queries-stability}]
    The precomputation consists in computing all $stab(witn(v), v)$ for $v \in V(G)$. This takes time $O(n+m)$ by \cref{lem:precom-stab}. We also do the $O(n+m)$-time precomputation described in \cref{lem:test-stab}.
    Now, suppose we are given a query $(a, b)$.
    If $a \geq b$, we return $stab(a, b) = 0$.
    Otherwise, using \cref{lem:test-stab}, we query whether $T[a, b]$ is stable. If not, we are told either that $T[a, b]$ is not a leftmost path or given a vertex $v \in T(a, b)$ with minimum $witn(v)$, witnessing that $T[a, b]$ is not stable, i.e. $witn(v) < a$.
    If $T[a, b]$ is stable, we return $stab(a, b) = a$.
    If $T[a, b]$ is not a leftmost path, we return $stab(a, b) = 0$.
    Otherwise, we return $stab(a, b) = stab(witn(v), v)$.

    \begin{claim*}
        If $T[a,b]$ is a leftmost path then $stab(a, b) = stab(witn(v), v)$.
    \end{claim*}

    \begin{subproof}
        Suppose first that $stab(witn(v), v) \neq 0$ and let $w = stab(witn(v), v)$.
        Every $u \in T(w, v)$ satisfies ${witn(u) \geq w}$ by \cref{lem:witn-stab}.
        Since $T[a, b]$ is a leftmost path and $v \in T(a, b)$ then $b$ is a leftmost descendant of $v$. Thus, $T[w, b]$ is a leftmost path. 
        Furthermore, if $u \in T(w, b)$, either $u \in T(w, v)$ or $u = v$ or $u \in T(v, b)$.
        If $u \in T(w, v)$ then $witn(u) \geq w$. 
        Since $w = stab(witn(v), v)$ then $w \geq witn(v)$ by definition. 
        If $u \in T(v, b)$ then $witn(u) \geq witn(v) \geq w$ by minimality of $witn(v)$.
        Therefore, $T[w, b]$ is stable by \cref{lem:witn-stab}, and $w \leq witn(v) < a$, so $stab(a, b) \geq w$.

        Conversely, suppose $stab(a, b) \neq 0$ and let $s = stab(a, b)$.
        Then, $s \leq a$ and $s \leq witn(u)$ for every $u \in T(s, b)$ by \cref{lem:witn-stab}.
        Thus, we have $s \leq witn(u)$ for every $u \in T(a, b)$ and thus $s \leq witn(v)$.
        Furthermore, if $u \in T(s, v)$ then $u \in T(s, b)$ so $witn(u) \geq s$. Thus, $T[s, v]$ is stable, and therefore $stab(witn(v), v) \geq s$.

        Therefore, if one of $stab(a, b), stab(witn(v), v)$ is non-zero then the other as well and the two are equal. 
    \end{subproof}
    This proves the correctness of the algorithm. Answering any query takes constant time.
\end{proof}

From now on, we assume that we performed the precomputation described in \cref{prop:queries-stability} and thus that we can compute any $stab(a, b)$ in constant time.

\subsection{Computing half-connected type-1 separations}

\begin{lemma}
\label{lem:comp-type-1-triples}
    There is an $O(n)$-time algorithm which, given a graph $G$, computes the set $\mathcal{S}$ of all triples of vertices $(a, b, d)$ such that $lwpt(d) = a, lwpt_2(d) = b$ and $b = p(d)$.
\end{lemma}

\begin{proof}
    We initially set $S = \emptyset$ and iterate over all vertices $d \in V(G)$. For each of them, if $lwpt_2(d) = p(d)$, we add the triple $(lwpt(d), p(d), d)$ to $S$.
    This takes time constant time for every $d \in V(G)$ so time $O(n)$ overall.

    Once we have iterated over all vertices, we have $S = \mathcal{S}$.
\end{proof}

\begin{lemma}
\label{lem:comp-type-1}
    There is an $O(n)$-time algorithm which, given a 2-connected graph $G$ equipped with a normal spanning tree $T$ and a numbering of its vertices which is compatible with $T$, and the set $\mathcal{S}$ all triples of vertices $(a, b, d)$ such that $lwpt(d) = a, lwpt_2(d) = b$ and $b = p(d)$, computes the set $\mathcal{T}_1$ of all half-connected type-1 separations.
\end{lemma}


\begin{proof}
    We start by sorting $\mathcal{S}$ lexicographically. Since $|\mathcal{S}| \leq n$ and all values in $\mathcal{S}$ are between 1 and $n$, this can be done in time $O(n)$ by \cref{cor:sort-lex}.
    By iterating over the tuples $(a, b, d) \in \mathcal{S}$ and over the children of all $a \in V(G)$, we can compute for every tuple $(a, b, d) \in \mathcal{S}$ the unique child $a'_{ab}$ of $a$ which is an ancestor of $b$.
    We initially set $\mathcal{T}_1 = \emptyset$. 
    We then iterate over all elements of $\mathcal{S}$.
    
    For every triple $(a, b, d) \in \mathcal{S}$, we add the tuple $(a, b, d, LD(d), d, d)$ to $\mathcal{T}_1$ unless $(a, b, d) = (1, 2, 3)$ and 3 is the only child of 2 in $T$.
    Furthermore, for every pair $(a, b)$ such that there is some triple $(a, b, \cdot)$ in $\mathcal{S}$, let $m_{a, b}$ be the smallest child $c$ of $b$ such that $(a, b, c) \in \mathcal{S}$, and $M_{a, b}$ be the largest child $c$ of $b$ such that $(a, b, c) \in \mathcal{S}$. Note that since $\mathcal{S}$ is sorted lexicographically, we can compute all $m_{a, b}$ and $M_{a, b}$ on the fly while traversing $\mathcal{S}$. 
    For every such $(a, b)$, if $m_{a, b} \neq M_{a, b}$, we add the tuple $(a, b, m_{a, b}, LD(M_{a, b}), m_{a, b}, M_{a, b})$ to $\mathcal{T}_1$ if $(a, b) \neq (1, 2)$ and either $a'_{ab}=b$ or $a=1$ or some vertex $u \in [a'_{ab}, m_{a,b}-1] \setminus \{b\}$ has a neighbor in $G$ in $[1, a-1] \cup [LD(M_{a,b}) + 1, n]$.

    \begin{claim*}
        If $(a, b, \textbf{f}, \textbf{l}, i, j) \in \mathcal{T}_1$ then $(a, b, \textbf{f}, \textbf{l}, i, j)$ encodes a half-connected type-1 separation.
    \end{claim*}

    \begin{subproof}
        Consider a tuple $(a, b, \textbf{f}, \textbf{l}, i, j) \in \mathcal{T}_1$.
        By design of the algorithm, either there exists $d \in V(G)$ such that we have ${(a, b, \textbf{f}, \textbf{l}, i, j) = (a, b, d, LD(d), d, d)}$ and $(a, b, d) \in \mathcal{S}$, or $(a, b, \textbf{f}, \textbf{l}, i, j) = (a, b, m_{a, b}, LD(M_{a, b}), m_{a, b}, M_{a, b})$ with $m_{a, b} \neq M_{a, b}$.
        Suppose first that there exists $d \in V(G)$ such that $(a, b, \textbf{f}, \textbf{l}, i, j) = (a, b, d, LD(d), d, d)$ and $(a, b, d) \in \mathcal{S}$. 
        By definition of $\mathcal{S}$, we have $a = lwpt(d), b = lwpt_2(d)$ and $b = p(d)$, thus $a < b < d$. If $(a, b, d) \neq (1, 2, 3)$ then $d > 3$ so there exists $v \notin \{a, b\}$ which is not a descendant of $d$ (take $v \in \{1, 2, 3\} \setminus \{a, b\}$). 
        If $(a, b, d) = (1, 2, 3)$, then 2 has a child $v \neq 3$, otherwise we would not have added $(a, b, d, LD(d), d, d)$ to $\mathcal{T}_1$. Then, $v \notin \{a, b\}$ is not a descendant of $d$.
        Since the ordering is compatible with $T$, the tuple $(a, b, d, LD(d), d, d)$ encodes the pair $(A, B)$ with $B = \{a, b\} \cup Desc(d)$ and $A = V(G) \setminus (B \setminus \{a, b\})$. By \cref{prop:charac-t1-half-con}, $(A, B)$ is a half-connected type-1 separation. Furthermore, we indeed have that $d$ is the minimum and maximum child of $b$ in $B$.
        
        Suppose now that $(a, b, \textbf{f}, \textbf{l}, i, j) = (a, b, m_{a, b}, LD(M_{a, b}), m_{a, b}, M_{a, b})$ with $m_{a, b} \neq M_{a, b}$.
        Let $\mathcal{C}$ be the set of children $c$ of $b$ such that $m_{a, b} \leq c \leq M_{a, b}$, and $\mathcal{D}$ be the set of children $d$ of $b$ such that $lwpt(d) = a$ and $lwpt_2(d) = b$.
        By definition of $\mathcal{S}$, we have $(a, b, d) \in \mathcal{S}$ if and only if $d \in \mathcal{D}$, and therefore $m_{a, b}$ is the minimum element of $\mathcal{D}$ and $M_{a, b}$ is the maximum element of $\mathcal{D}$.
        Thus, if $d \in \mathcal{D}$, then $d$ is a child of $b$ and $m_{a, b} \leq d \leq M_{a, b}$ so $d \in \mathcal{C}$.
        Conversely, if $c \in \mathcal{C}$ then $c$ is a child of $b$ such that $m_{a, b} \leq c \leq M_{a, b}$. 
        However, the children of $b$ are numbered according to their pair $(lwpt(\cdot), lwpt_2(\cdot))$. Since $(a, b, m_{a, b}) \in \mathcal{S}$ and $(a, b, M_{a, b}) \in \mathcal{S}$ then ${(lwpt(m_{a, b}), lwpt_2(m_{a, b})) = (a, b) = (lwpt(M_{a, b}), lwpt_2(M_{a, b}))}$. Since we have $m_{a, b} \leq c \leq M_{a, b}$, it follows that $(lwpt(c), lwpt_2(c)) = (a, b)$ and thus $c \in \mathcal{D}$. 
        This finishes to prove that ${\mathcal{C} = \mathcal{D}}$, so in particular we have $|\mathcal{D}| > 1$.
        Since the ordering is compatible with $T$, the tuple $(a, b, m_{a, b}, LD(M_{a, b}), \cdot, \cdot)$ encodes the pair $(A, B)$ with ${B = \{a, b\} \cup \bigcup_{d \in \mathcal{D}} Desc(d)}$ and $A = V(G) \setminus (B \setminus \{a, b\})$.
        Furthermore, $m_{a,b}$ is the smallest element of $\mathcal{D}$ and $LD(M_{a,b})$ is the maximum vertex which is the descendant of some $d \in \mathcal{D}$.
        Since ${(a, b) \neq (1, 2)}$, there exists ${v \in \{1, 2\} \setminus \{a, b\}}$. Then, $v \notin \{a, b\}$ is not a descendant of any $d \in \mathcal{D}$.
        Furthermore, since we added the tuple $(a, b, m_{a, b}, LD(M_{a, b}), m_{a, b}, M_{a, b})$ to $\mathcal{T}_1$ then either $a'_{ab}=b$ or $a=1$ or some vertex $u \in [a'_{ab}, m_{a,b}-1] \setminus \{b\}$ has a neighbor in $G$ in $[1, a-1] \cup [LD(M_{a,b}) + 1, n]$
        It now follows from \cref{prop:charac-t1-half-con} that $(A, B)$ is a half-connected type-1 separation. Furthermore, $m_{a, b}$ is the minimum child of $b$ in $B$ and $M_{a, b}$ the maximum child of $b$ in $B$.
    \end{subproof}

    \begin{claim*}
        If $(A, B)$ is a half-connected type-1 separation, some tuple $(a, b, \textbf{f}, \textbf{l}, c_m, c_M) \in \mathcal{T}_1$ correctly encodes it.
    \end{claim*}

    \begin{subproof}
        Let $(A, B)$ be such a separation with separator $\{a < b\}$. By \cref{lem:a-b-comp}, we have that $a$ is an ancestor of $b$. Let $a' \in V(G)$ be the only child of $a$ which is an ancestor of $b$. Let $\mathcal{D}$ be the set of children $d$ of $b$ such that $lwpt(d) = a$ and $lwpt_2(d) = b$. Let $\alpha$ be the smallest element of $\mathcal{D}$ and $\omega$ be the maximum vertex which is the descendant of some $d \in \mathcal{D}$.
        By \cref{prop:charac-t1-half-con}, one of the following holds.

        \begin{itemize}
            \item There exists $d \in \mathcal{D}$ such that, up to exchanging $A$ and $B$, we have $B = \{a, b\} \cup Desc(d)$ and $A = V(G) \setminus (B \setminus \{a, b\})$, and some vertex $v \notin \{a, b\}$ is not a descendant of $d$.
            \item Up to exchanging $A$ and $B$, we have $B = \{a, b\} \cup \bigcup_{d \in \mathcal{D}} Desc(d)$ and $A = V(G) \setminus (B \setminus \{a, b\})$, $|\mathcal{D}| > 1$, some vertex $v \notin \{a, b\}$ is not a descendant of any $d \in \mathcal{D}$ and either $a'=b$ or $a=1$ or some vertex $u \in [a', \alpha-1] \setminus \{b\}$ has a neighbor in $G$ in $[1, a-1] \cup [\omega + 1, n]$.
        \end{itemize}

        In the first case, the separation $(A, B)$ is encoded by $(a, b, d, LD(d), d, d)$. Since $d \in \mathcal{D}$ then $(a, b, d) \in \mathcal{S}$. If we have ${(a, b, d) \neq (1, 2, 3)}$ then $(a, b, d, LD(d), d, d)$ is added to $\mathcal{T}_1$ when considering $(a, b, d) \in \mathcal{S}$. If $(a, b, d) = (1, 2, 3)$ then $v$ witnesses that 3 is not the only child of 2 in $T$, thus we also added $(a, b, d, LD(d), d, d)$ to $\mathcal{T}_1$ when considering $(a, b, d) \in \mathcal{S}$.

        In the second case, the separation $(A, B)$ is encoded by $(a, b, m_{a,b},LD(M_{a,b}),m_{a,b},M_{a,b})$.
        Since $|\mathcal{D}| > 1$ then we have ${m_{a,b} \neq M_{a,b}}$ and there is some triple $(a, b, \cdot) \in \mathcal{S}$.
        Note that $m_{a,b} = \alpha$ and $\omega = LD(M_{a,b})$.
        Since $v$ is not a descendant of any $d \in \mathcal{D}$ then $(a, b) \neq (1, 2)$, and since either $a'_{ab} = b$ or $a=1$ or some vertex $u \in [a'_{ab}, m_{a,b}-1] \setminus \{b\}$ has a neighbor in $G$ in $[1, a-1] \cup [M_{a,b} + 1, n]$ then we added $(a, b, m_{a,b},LD(M_{a,b}),m_{a,b},M_{a,b})$ to $\mathcal{T}_1$ when considering the pair $(a, b)$.
    \end{subproof}

    To show that this algorithm runs in time $O(n)$, since $|\mathcal{S}| \leq n$, it suffices to show that we can check in constant time whether some vertex $u \in [a'_{ab}, m_{a,b}-1] \setminus \{b\}$ has a neighbor in $G$ in $[1, a-1] \cup [LD(M_{a,b}) + 1, n]$.
    This can simply be done by querying the minimum of $N_{\min}[a'_{ab}, b-1]$ and $N_{\min}[b+1, m_{a,b}-1]$, and the maximum of $N_{\max}[a'_{ab}, b-1]$ and $N_{\max}[b+1, m_{a,b}-1]$, which we can do in constant time with our precomputation.
\end{proof}

\subsection{Computing half-connected type-2 separations}

The purpose of this section is to prove the following proposition.
\begin{proposition}
\label{prop:comp-type-2}
    There is an $O(n+m)$-time algorithm which, given a 2-connected graph $G$ equipped with a normal spanning tree $T$ and a numbering of its vertices which is compatible with $T$, computes a set~$\mathcal S$ of half-connected type-2 separations such that every half-connected type-2 separation not in~$\mathcal S$ is crossed by a separation in~$\mathcal S$.
    In particular, $\mathcal S$ contains all totally-nested type-2 separations.
\end{proposition}

In the following definition, one should think of $\{a, b\}$ as the separator of a type-2 separation, and if $\mathcal{D}$ is the set of children $d$ of $b$ such that $lwpt(d) = a$ and $lwpt_2(d) = b$ then $d_1$ is the largest child of $b$ before the elements of $\mathcal{D}$, and $d_2$ is the smallest child of $b$ after the elements of $\mathcal{D}$.

\begin{definition}\label{dfn:type-2-p-good}
    Given a 2-connected graph $G$ equipped with a normal spanning tree $T$ and a numbering of its vertices which is compatible with $T$, and a maximal leftmost path $P$ of $G$, we say that a tuple $(a,b,d_1,d_2)$ is \emph{$P$-good} if it satisfies all of the following properties.
    \begin{enumerate}[(1)]
        \item\label{item:p-good-1} $a,b \in P$ and $1 < a \leq p(p(b))$.
        \item\label{item:p-good-2}  $p(d_1)=b$ or $d_1 = 0$, and $p(d_2)=b$ or $d_2=n+1$.
        \item\label{item:p-good-3}  All children $d \leq d_1$ of $b$ satisfy $lwpt(d) \geq a$ and $a < hgpt(d)$.
        \item\label{item:p-good-4}  All children $d$ of $b$ with $d_1 < d < d_2$ satisfy $lwpt(d) = a$ and $lwpt_2(d)=b$.
        \item\label{item:p-good-5}  All children $d \geq d_2$ of $b$ satisfy $hgpt(d) \leq a$ and $lwpt(d) < a$.
        \item\label{item:p-good-6}  $T[a,b]$ is stable.
    \end{enumerate}
\end{definition}

\begin{observation}\label{obs:p-good-unique}
    No two $P$-good tuples $(a,b,d_1,d_2) \neq (a',b',d_1',d_2')$ have $a=a'$ and $b=b'$.\qed
\end{observation}

\begin{observation}\label{obs:d-moves-to-left}
    If $(a,b,d_1,d_2)$ and $(a',b,d_1',d_2')$ with $a > a'$ are two $P$-good tuples then $d_1 \leq d_1'$ and $d_2 \leq d_2'$.\qed
\end{observation}

\begin{lemma}\label{lem:eq-P-good-separator}
    Let $G$ be a 2-connected graph equipped with a normal spanning tree $T$ and a numbering of its vertices which is compatible with $T$, and let $P$ be a maximal leftmost path of $G$.
    Let $a < b \in V(P)$.
    There exists a $P$-good tuple $(a,b,d_1,d_2)$ if and only if $\{a,b\}$ is the separator of a half-connected type-2 separation.
    If so, the tuple $(a,b,d_1,d_2)$ is unique for $\{a,b\}$.
\end{lemma}
\begin{proof}
    Observe that a child $c$ of $b$ satisfies $hgpt(c) = lwpt(c)$ if and only if $lwpt_2(c) = b$. 
    Thus, it follows from \cref{prop:charac-t2-half-con} that there exists a $P$-good tuple $(a,b,d_1,d_2)$ if and only if $\{a,b\}$ is a type-2 separator.
    The uniqueness follows from \cref{obs:p-good-unique}. 
    %
    %
    %
\end{proof}


%
We say that two $P$-good tuples $(a,b,d_1,d_2)$ and $(a',b',d_1',d_2')$ \emph{cross} if $a<a'<b<b'$ or $a'<a<b'<b$.
Otherwise, they are \emph{nested}.
A set $\mathcal S_P$ of $P$-good tuples is \emph{nested} if every two elements in $\mathcal S_P$ are nested.
If $\{a < b\}$ is the separator of a totally-nested type-2 separation then its corresponding $P$-good tuple is not crossed by another $P$-good tuple by \cref{lem:crossing-vertices}. 
With \cref{lem:eq-P-good-separator} it then suffices to compute a maximal nested set of $P$-good tuples for every maximal leftmost path~$P$, which we will do in the following lemma.
    
\begin{lemma}\label{lem:pot-nested-type-2}
    There is an $O(n+m)$-time algorithm which, given a 2-connected graph $G$ equipped with a normal spanning tree $T$ and a numbering of its vertices which is compatible with $T$, computes for every maximal leftmost path $P$ a maximal nested set $\mathcal S_P$ of $P$-good tuples.
\end{lemma}
\begin{proof}
    We perform \cref{alg:potentially-nested-type-2} on every maximal leftmost path $P$.
    Note that its requirements are satisfied by \cref{prop:queries-stability}.
    The idea is the following. For a maximal leftmost path $P=T[v_1,v_\ell]=v_1 v_2 \ldots v_\ell$, we start with $b=v_\ell$ and compute all $P$-good tuples $(a,b,d_1,d_2)$ (for this fixed $b$), which we add to $\mathcal S_P$. Moreover, we put the corresponding $a$'s on a stack $S$ (which is initially empty), which is ordered such that the top of the stack is the largest element. Then we move $b$ down.
    In step $i$, we consider $b = v_{\ell - i + 1}$. Let $a'$ be the current top of the stack $S$, which means that there exists a $P$-good tuple $(a',b',d_1',d_2') \in \mathcal S_P$ with $a'<b<b'$ (see invariant~\ref{inv:1}). So, we only have to compute the tuples $(a,b,d_1,d_2)$ with $a \geq a'$ (since otherwise $a<a'<b<b'$ and hence $(a,b,d_1,d_2)$ crosses $(a',b',d_1',d_2')$), which we will do in the following.
    
    Let $c_1 < c_2 < \cdots < c_k$ be the children of $b$.
    We introduce two indices $i_1$ and $i_2$ that must satisfy the invariants~\ref{inv:2}, \ref{inv:3} and~\ref{inv:4} below. Intuitively, $c_{i_1}$ and $c_{i_2}$ will correspond to the unique $d_1$ and $d_2$ that we want to compute for a given type-2 separator $\{a,b\}$.
    For the given $b=v_{\ell - i + 1}$, we want to compute all tuples $(a,b,d_1,d_2)$ with $a \geq a'$ where $a'$ is the top of the stack $S$.
    For this, we start with $a=stab(p(p(b)),b)$. Observe that $a \neq 0$ if and only if it lies on the path $P$. After that, we adjust in \cref{line:inv-a-b-start} to \cref{line:alg-invariant-a-b} the indices $i_1$ and $i_2$ such that invariants~\ref{inv:2}--\ref{inv:4} hold.
    In particular, we have $lwpt(c_{i_2}) < a$ and therefore \ref{inv:4} holds.
    Moreover, we have that $i_1$ is the highest index such that $(lwpt(c_{i_1}),hgpt(c_{i_1})) >_{lex} (a,a)$ (or $i_1=0$ if it does not exist). Therefore, \ref{inv:2} holds. 
    Then, all other children $c_{i_1}<d<c_{i_2}$ of $b$ must satisfy $lwpt(d)=a$ and $hgpt(d)=a$, which implies $lwpt_2(d)=b$. Therefore, \ref{inv:3} holds.

    Now, we need to check whether the tuple $(a,b,c_{i_1},c_{i_2})$ is $P$-good. Properties \ref{item:p-good-1} and \ref{item:p-good-6} of \cref{dfn:type-2-p-good} are satisfied if $a > 1$. Properties \ref{item:p-good-2}--\ref{item:p-good-4} are satisfied by invariants \ref{inv:2} and \ref{inv:3}. Hence, it remains to check property~\ref{item:p-good-5}. We have $lwpt(d) < a$ for all children $d \geq c_{i_2}$ by \ref{inv:4}, so we only have to check whether $maxhgpt(c_{i_2}) \leq a$ (if $1 \leq i_2 \leq k$, otherwise property~\ref{item:p-good-5} is satisfied). Therefore, if we have $a > 1$, and $i_2 = k+1$ or $maxhgpt(c_{i_2}) \leq a$ then we found a $P$-good tuple $(a,b,c_{i_1},c_{i_2})$ and we add it to the set $\mathcal S_P$, and we add $a$ to $S'$, a stack which stores all $a$'s such that $(a, b, \cdot, \cdot) \in \mathcal{S}_P$ for this value of $b$. Then, we compute the next possibility for $a$, which is $stab(p(a),b)$ and repeat the procedure (for the fixed $b$).
    Otherwise, if we found $maxhgpt(c_{i_2}) > a$ then we compute the next possibility for $a$, which is $stab(lwpt(c_{i_2}),b)$ (if $1 \leq i_2 \leq k$) and repeat the procedure  (for the fixed $b$).
    Otherwise, if we find $a \in \{0,1\}$ or $a<a'$ (where $a'$ is the top of the stack $S$), then we move $b$ down, and repeat the algorithm with the next $b$.

    We now describe the invariants that we maintain throughout the algorithm, and then we prove the correctness and runtime of the algorithm.

    Invariants when considering $b$ in \cref{line:alg-invariant-b}:
    \begin{enumerate}[({I}1)]
        \item\label{inv:1} The stack $S$ contains all vertices $a'$ such that there exists a tuple $(a',b',d_1',d_2') \in \mathcal S_P$ with $b' > b > a'$. Moreover, $S$ is ordered such that the top is the largest element.
        \item\label{inv:8} If there is a $P$-good tuple $(a',b',d_1',d_2')$ with $b' > b$ that is nested with all $P$-good tuples in $\mathcal S_P$ then we have ${(a',b',d_1',d_2') \in \mathcal S_P}$.
    \end{enumerate}

    Invariants when considering $(b,a)$ in \cref{line:alg-invariant-a-b}: ($k$ is the number of children of $b$. We define $c_0 = 0$ and ${c_{k+1}=n+1}$.)
    \begin{enumerate}[({I}1)]
        \setcounter{enumi}{2}
        \item\label{inv:2} All children $d \leq c_{i_1}$ of $b$ have $lwpt(d) \geq a$ and $hgpt(d) > a$.
        \item\label{inv:3} All children $d$ of $b$ with $c_{i_1} < d < c_{i_2}$ have $lwpt(d)=a$ and $lwpt_2(d)=b$.
        \item\label{inv:4} All children $d \geq c_{i_2}$ of $b$ have $lwpt(d) < a$.
        \item\label{inv:5} The stack $S'$ contains all $a' \geq a$ such that $(a',b,d_1',d_2') \in \mathcal S_P$, ordered such that the top is the smallest element.
        \item\label{inv:6} If there exists a $P$-good tuple $(a',b,d_1',d_2')$ with $a' > a$ then $(a',b,d_1',d_2') \in \mathcal S_P$.
        \item\label{inv:7} The $P$-good tuples in $\mathcal S_P$ form a nested set.
    \end{enumerate}

    \begin{algorithm}
        \caption{Algorithm for \cref{lem:pot-nested-type-2}}
        \label{alg:potentially-nested-type-2}
        \begin{algorithmic}[1]
        \Require $G$ graph with normal spanning tree $T$, consistent numbering and a maximal leftmost path $P=v_1 v_2 \ldots v_\ell$
        \Require $lwpt(v)$, $hgpt(v)$ and $maxhgpt(v)$ for every vertex $v$
        \Require Oracle to query $stab(x,y)$ for any $x,y \in V(G)$ in constant time.
    
        \State $\mathcal S_P \gets$ empty stack \Comment{nested set of $P$-good tuples}
        \State $S \gets$ empty stack \Comment{stack of already computed $a$'s}
        \For{$b=v_\ell,v_{\ell-1},\ldots,v_2$}\label{line:type-2-for-loop}
            \While{$b = S$.top()}
                \State $S$.pop()
            \EndWhile
            \State $c_1,c_2,\ldots,c_k \gets$ children of $b$ with $c_1 < c_2 < \cdots < c_k$ \label{line:alg-invariant-b}
            \State $c_0 \gets 0$ ; $c_{k+1} \gets n+1$
            \State $i_1 \gets 0$  ; $i_2 \gets 1$
            \State $a \gets stab(p(p(b)),b)$
            \State $S' \gets$ empty stack \Comment{stack of $a$'s for the current $b$}
            \While{$a > 1$ and [$S=\emptyset$ or $a \geq S$.top()]}\label{line:while-loop-a-b} 
                \While{$i_2 \leq k$ and $lwpt(c_{i_2}) \geq a$}\label{line:inv-a-b-start}
                    \State $i_2 \gets i_2 + 1$
                \EndWhile
                \While{$i_1 + 1 \leq k$ and $(lwpt(c_{i_1 + 1}),hgpt(c_{i_1 + 1})) >_{lex} (a,a)$}
                    \State $i_1 \gets i_1 + 1$
                \EndWhile
    
                \If{$i_2 = k+1$ or $maxhgpt(c_{i_2}) \leq a$}\label{line:alg-invariant-a-b} 
                    \State $S'$.push($a$)
                    \State $\mathcal S_P \gets \mathcal S_P \cup \{(a,b,c_{i_1},c_{i_2})\}$ \label{line:add-tuple-to-S_P}
                    \If{$a \neq v_1$}
                        \State\label{line:a-assign-1} $a \gets stab(p(a),b)$
                    \Else
                        \State\label{line:a-assign-2} $a \gets 0$
                    \EndIf
                \ElsIf{$i_2 \leq k$} 
                    \State\label{line:a-assign-3} $a \gets stab(lwpt(c_{i_2}),b)$
                \Else 
                    \State\label{line:a-assign-4} $a \gets 0$
                \EndIf
            \EndWhile
            \While{$S' \neq \emptyset$}
                \State $S$.push($S'$.pop())
            \EndWhile
        \EndFor
        \State\Return $\mathcal S_P$
        \end{algorithmic}
    \end{algorithm}
    
    \begin{claim*}
        The invariants \ref{inv:1},~\ref{inv:2}--\ref{inv:5} are satisfied throughout the algorithm.
    \end{claim*}
    \begin{subproof}
        \ref{inv:2}--\ref{inv:4}. \cref{line:inv-a-b-start} to \cref{line:alg-invariant-a-b} ensure that these invariants are satisfied. Observe that as $a$ moves down, the unique values for $i_1$ and $i_2$ that satisfy \ref{inv:2}--\ref{inv:4} cannot decrease. In particular, let $a^{(1)} > a^{(2)}$ be the values of $a$ in two consecutive runs of \cref{line:alg-invariant-a-b} and $i_1^{(1)},i_2^{(1)},i_1^{(2)},i_2^{(2)}$ be the corresponding values of $i_1$ and $i_2$ that satisfy \ref{inv:2}--\ref{inv:4}. If a child $d$ of $b$ has $lwpt(d) \geq a^{(1)}$ and $hgpt(d) > a^{(1)}$ then it has $lwpt(d) \geq a^{(2)}$ and $hgpt(d) > a^{(2)}$, therefore $i_1^{(1)} \leq i_1^{(2)}$. If a child $d$ of $b$ has $lwpt(d) < a^{(2)}$ then it also has $lwpt(d) < a^{(1)}$, therefore $i_2^{(1)} \leq i_2^{(2)}$.

        \ref{inv:5}. Observe that whenever we add a tuple $(a,b,d_1,d_2)$ to $\mathcal S_P$, we add $a$ to $S'$. Moreover, for a fixed $b$, the value of $a$ only decreases, so the stack $S'$ is ordered such that the top is the smallest element.
        
        \ref{inv:1}. Consider the $j$-th time the algorithm reaches \cref{line:alg-invariant-b}. If $j=1$ then $\mathcal S_P = \emptyset$, so the invariant is satisfied. Let $j>2$ and assume that the invariant was satisfied when the algorithm reached \cref{line:alg-invariant-b} for the $(j-1)$-th time. Observe that whenever we add a tuple $(a,b,d_1,d_2)$ to $\mathcal S_P$, we will add $a$ to the stack $S$. Moreover, we only add elements $a$ on top of $S$ that satisfy $a \geq a_t$ (where $a_t$ is the current top of the stack $S$). Therefore, the stack is ordered such that the top is the largest element. Moreover, we remove the top of the stack $a_t$ if $a_t = b$ and add all $a' \in S'$ (i.e.\ $a'$'s such that there exists $(a',b'=b,d_1',d_2') \in \mathcal S_P$). Hence, $S$ contains all $a'$ such that there exists $(a',b',d_1',d_2')\in \mathcal S_P$ with $b' > b > a'$.
    \end{subproof}
    
    \begin{claim}\label{claim:add-tuple-iff-p-good}
        If we add the tuple $(a,b,c_{i_1},c_{i_2})$ into $\mathcal S_P$, then $(a,b,c_{i_1},c_{i_2})$ is $P$-good.
    \end{claim}
    \begin{subproof}    
        Observe that we only add tuples $(a,b,d_1,d_2)$ to $\mathcal S_P$ in \cref{line:add-tuple-to-S_P}.
        Obviously $b \in P$. Since $a > 1$, we have $a \in P$, so \cref{dfn:type-2-p-good}~\ref{item:p-good-1} is satisfied.
        Moreover $i_1 \in \{0,1,\ldots,k\}$ and $i_2 \in \{1,2,\ldots,k+1\}$ (where $k$ is the number of children of $b$), and we defined $c_0:=0$ and $c_{k+1}:=n+1$. Hence, \cref{dfn:type-2-p-good}~\ref{item:p-good-2} is satisfied. Invariants~\ref{inv:2} and~\ref{inv:3} ensure \cref{dfn:type-2-p-good}~\ref{item:p-good-3} and~\ref{item:p-good-4}.
        Moreover, we have $lwpt(c_{i_2}) < a$ (by invariant~\ref{inv:4}) and $maxhgpt(c_{i_2}) \leq a$, which implies $lwpt(d) < a$ and $hgpt(d) \leq a$ for every child $d \geq c_{i_2}$ of $b$. Hence, \cref{dfn:type-2-p-good}~\ref{item:p-good-5} is satisfied.
        Moreover, note that \cref{dfn:type-2-p-good}~\ref{item:p-good-6} is satisfied since when assigning a value to $a$, we always ensure the stability of the interval $T[a,b]$.
    \end{subproof}
    
    \begin{claim*}
        Invariant~\ref{inv:6} is satisfied throughout the algorithm.
    \end{claim*}
    \begin{subproof}
        In the first run of \cref{line:alg-invariant-a-b}, we assign $a \gets stab(p(p(b)),b)$, i.e.\ $a$ is the largest vertex $a \leq p(p(b))$ such that $T[a,b]$ is stable. Hence, for every tuple $(a',b,\cdot,\cdot)$ with $a' > a$ either $a' \geq p(b)$ (violating property~\ref{item:p-good-1}) or $T[a',b]$ is not stable (violating property~\ref{item:p-good-6}). This satisfies invariant~\ref{inv:6} in the first run of \cref{line:alg-invariant-a-b}.
        
        Now, assume that the invariant holds in the $(j-1)$-th run of \cref{line:alg-invariant-a-b} and we show that it still holds in the $j$-th run. Let $a_{j-1}$ (respectively $a_j$) be the value of $a$ in the $(j-1)$-th (respectively $j$-th) run of \cref{line:alg-invariant-a-b}. Note that we assign the value for $a_j$ in one of \cref{line:a-assign-1,line:a-assign-2,line:a-assign-3,line:a-assign-4}, and since there is a $j$-th run of \cref{line:alg-invariant-a-b} then it must be either at \cref{line:a-assign-1} or at \cref{line:a-assign-3}.
        
        First, suppose that $a_j$ was assigned at \cref{line:a-assign-1}, and consider \cref{line:add-tuple-to-S_P} in the $(j-1)$-th run. By \cref{obs:p-good-unique}, there exists at most one $P$-good tuple $(a_{j-1},b,d_1,d_2)$. By \cref{claim:add-tuple-iff-p-good}, we add the tuple $(a_{j-1},b,d_1,d_2)$ to $\mathcal S_P$ if and only if it is $P$-good. 
        After that, we move $a$ down (i.e.\ $a_j$ is an ancestor of $a_{j-1}$).
        Since we set $a_j = stab(p(a_{j-1}),b)$ in \cref{line:a-assign-1}, then $a_j$ is the largest $a' \leq p(a_{j-1})$ such that $T[a',b]$ is stable. Since we added $(a_{j-1},b,d_1,d_2)$ to $\mathcal S_P$ if and only if it is $P$-good and since there exist no $P$-good tuples $(a',b,\cdot,\cdot)$ with $a_{j-1} > a' > a_j$, because of property~\ref{item:p-good-6}, invariant~\ref{inv:6} is satisfied.
        
        On the other hand, if we set $a_j = stab(lwpt(c_{i_2}),b)$ in \cref{line:a-assign-3} then there must exist a child $d \geq c_{i_2}$ of $b$ such that ${hgpt(d) > a_{j-1}}$. 
        However, there cannot be a $P$-good tuple $(a',b,\cdot,\cdot)$ with $hgpt(d) > a' > lwpt(d)$. Thus, there exists no $P$-good tuple $(a',b,\cdot,\cdot)$ with $a_{j-1} \geq a' > lwpt(c_{i_2})$. Since $T[a',b]$ needs to be stable, we can say that there exists no $P$-good tuple $(a',b,\cdot,\cdot)$ with $a_{j-1} \geq a' > stab(lwpt(c_{i_2}),b)$. Therefore, invariant \ref{inv:6} is satisfied.
    \end{subproof}
    
    \begin{claim*}
        Invariant~\ref{inv:7} is satisfied throughout the algorithm.
    \end{claim*}
    
    \begin{subproof}
        Note that invariant~\ref{inv:7} is satisfied in the first run of \cref{line:alg-invariant-a-b} since $\mathcal S_P$ is empty. Assume that invariant~\ref{inv:7} is satisfied in the $(j-1)$-th run of \cref{line:alg-invariant-a-b}. Note that we only add tuples to $\mathcal S_P$ in \cref{line:add-tuple-to-S_P}. Observe that when we add a tuple $(a,b,\cdot,\cdot)$ to $\mathcal S_P$ then $b \leq b'$ holds for all tuples $(a',b',\cdot,\cdot) \in \mathcal S_P$. Moreover, if $b' > b$ then we only add $(a,b,\cdot,\cdot)$ to $\mathcal S_P$ if $a \geq a'$ because of invariant~\ref{inv:1} and of the test at \cref{line:while-loop-a-b}. In this case, we have $a' \leq a < b < b'$ and therefore these tuples are nested. Otherwise, if $b=b'$, the tuples $(a,b,\cdot,\cdot)$ and $(a',b',\cdot,\cdot)$ are also nested.
    \end{subproof}
    
    \begin{claim*}
        Invariant~\ref{inv:8} is satisfied throughout the algorithm.
    \end{claim*}
    
    \begin{subproof}
        The invariant is clearly satisfied before we first consider \cref{line:alg-invariant-b}. Assume that the invariant holds before we consider \cref{line:alg-invariant-b} for the $(j-1)$-th time. 
        Before we consider \cref{line:alg-invariant-b} for the $j$-th time, invariant~\ref{inv:8} still holds for all tuples $(a',b',\cdot,\cdot)$ with $b'>b+1$ since we decreased $b$ by one compared to the previous passage. 
        Therefore we only need to show that if there exists a $P$-good tuple $(a',b',d'_1,d'_2)$ with $b'=b+1$ that is nested with all $P$-good tuples in $\mathcal S_P$ then it was added to $\mathcal{S}_P$ at the $(j-1)$-th run.

        For this, let $a_1,a_2,\ldots,a_\ell$ be the sequence of $a$'s that we consider for this fixed $b+1$ in \cref{line:alg-invariant-a-b}. By invariant~\ref{inv:6}, every $P$-good tuple $(a'',b+1,\cdot,\cdot)$ with $a'' \geq a_\ell$ is in $\mathcal S_P$. The next value of $a$ (below $a_\ell$) does no longer satisfy the condition of the while-loop in \cref{line:while-loop-a-b}. So either $a=0$ or $a<S$.top(), in both cases there does not exist a $P$-good tuple $(a'',b+1,\cdot,\cdot)$ with $a'' < a_\ell$ that is nested with all $P$-good tuples in $\mathcal S_P$.
    \end{subproof}
    \begin{claim}
        \label{claim:type-2-maximal-nested-set}
        At the end of \cref{alg:potentially-nested-type-2}, the tuples in $\mathcal S_P$ form a maximal nested set of $P$-good tuples.
    \end{claim}

    \begin{subproof}
        This follows directly from invariants~\ref{inv:7} (nestedness) and~\ref{inv:8} (maximality) since the algorithm terminates with $b=v_2$ and there is no $P$-good tuple $(\cdot, b', \cdot, \cdot)$ with $b' \leq b$ by property~\ref{item:p-good-1}.
    \end{subproof}
    
    \begin{claim}
        \label{claim:type-2-running-time}
        For a maximal leftmost path $P=v_1 v_2 \ldots v_\ell$, the runtime of \cref{alg:potentially-nested-type-2} is in $O\left(\ell +  
\sum_{i=2}^\ell \deg(v_i) \right)$.
    \end{claim}
    
    \begin{subproof}For a fixed $v_i$, let $\ell_i$ be the number of tuples $(\cdot, v_i, \cdot, \cdot)$ in $\mathcal S_P$.
        First, we show that the running time of one run of the for-loop in \cref{line:type-2-for-loop} for a fixed $b=v_i$ is in $O(\ell_i + \deg(v_i))$.
        For this, let $b=v_i$ be fixed. Observe that we repeat the while-loop in \cref{line:while-loop-a-b} at most $\ell_i+1$ times. Therefore, the total running time of \cref{line:alg-invariant-a-b} to \cref{line:a-assign-4} is in $O(\ell_i)$. To see that the total running time of \cref{line:inv-a-b-start} to \cref{line:alg-invariant-a-b} is in $O(\ell_i + \deg(v_i))$, observe that the indices $i_1$ and $i_2$ only ever increase and are at most $\deg(v_i)+1$.
        
        Now, note that $\sum_{i=1}^\ell \ell_i = \mathcal |\mathcal S_P|$. Moreover, note that the set of all intervals $(a,b)$ such that $(a,b,d_1,d_2)$ is in~$\mathcal S_P$, forms a laminar set family.
        Since $|P|=\ell$, by a standard result it follows that $|\mathcal S_P| \leq 2\ell-1$. Then the claim follows with $\sum_{i=1}^\ell \ell_i = \mathcal |\mathcal S_P| \leq 2 \ell - 3$.
    \end{subproof}

    We perform \cref{alg:potentially-nested-type-2} for every maximal leftmost path $P$. Then, the total running time is in $O(n+m)$, which follows from \cref{claim:type-2-running-time}. By \cref{claim:type-2-maximal-nested-set}, we get a maximal nested set of $P$-good tuples for every leftmost path $P$.
\end{proof}

\begin{lemma}
\label{lem:P-good-eq-type-2-sepn}
    Let $G$ be a 2-connected graph equipped with a normal spanning tree $T$ and a numbering of its vertices which is compatible with $T$.
    Let~$P$ be a maximal leftmost path in~$G$ and let~$(a,b,d_1,d_2)$ be a $P$-good tuple.
    Then there exist at most two half-connected type-2 separations with separator $\{a,b\}$.
    Moreover, there is an algorithm which, given the pair $(G, T)$ and the leftmost path $P$, after an $O(n+m)$-time precomputation, can compute for any given $P$-good tuple the encodings of these separations in constant time.
\end{lemma}

\begin{proof}
    Let $(A, B)$ be a half-connected type-2 separation with separator $\{a, b\}$, and let $i$ be as in the statement of \cref{prop:charac-t2-half-con}.
    We prove that $c_i \in \{d_1, d_2-1\}$.
    With the notations of \cref{prop:charac-t2-half-con}, since $(a,b,d_1,d_2)$ is a $P$-good tuple then $\mathcal{D} = \{c_j : d_1 < c_j < d_2\}$, so the last item of \cref{prop:charac-t2-half-con} implies that either $c_i \leq d_1$ or $c_i \geq d_2-1$.
    If $c_i < d_1$ then $hgpt(d_1) > a$, which contradicts the penultimate item of \cref{prop:charac-t2-half-con}.
    If $c_i > d_2-1$ then $lwpt(d_2) < a$, which again contradicts the penultimate item of \cref{prop:charac-t2-half-con}.
    Thus, $c_i \in \{d_1, d_2-1\}$.
    Since the separation $(A, B)$ is entirely determined by the index $i$ then there exist at most two half-connected type-2 separations with separator $\{a,b\}$.

    Furthermore, by \cref{prop:charac-t2-half-con}, if $c_1, \ldots, c_k$ are the children of $b$, either $B = \{a, b\} \cup V(T_2) \cup \bigcup_{c_i \leq d_1}Desc(c_i)$ or $B = \{a, b\} \cup V(T_2) \cup \bigcup_{c_i < d_2}Desc(c_i)$, and in both cases $A = V(G) \setminus (B \setminus \{a, b\})$. Denote the first separation by $(A, B)$ and the second by $(A', B')$. By \cref{prop:charac-t2-half-con} and by definition of $P$-good tuples, both are half-connected type-2 separations.
    

    Next, we show how to compute the encodings of~$(A,B)$ and~$(A',B')$.
    Let $\textbf{f}:=a'$.
    Let $\textbf l := LD(d_1)$ if $d_1 > 0$, and otherwise $\textbf l := b-1$.
    Let $m:=b+1$ if $d_1 > 0$, and otherwise $m:=0$.
    Let $M:=d_1$ if $d_1>0$, and otherwise $M:=0$.
    Then~$(A,B)$ is encoded by $(a,b,\textbf f, \textbf l, m, M)$.

    Let $\textbf f':=a'$.
    Let $\textbf l':=d_2-1$ if $d_2 < n+1$, and otherwise $\textbf l'=LD(b)$ if $b$ is not a leaf, and otherwise $\textbf l':=b-1$.
    Let $m':=b+1$ if $d_2>b+1$ and~$b$ is not a leaf, and otherwise $m':=0$.
    Let $M'$ be the right sibling of~$d_2$ if $b+1 < d_2 < n+1$, and otherwise let~$M'$ be the largest child of~$b$ if $d_2 = n+1$ and~$b$ is not a leaf, and otherwise let $M':=0$.
    Then~$(A',B')$ is encoded by $(a,b,\textbf f', \textbf l', m', M')$.

    In \cref{subsec:precomputations}, we showed that we can precompute $LD(v)$ for all~$v \in V(G)$ in time $O(n+m)$.
    Moreover, recall that we precomputed the rank of every vertex among its siblings as well as an array for each vertex that contains all children and is sorted by increasing number.
    Then it is straightforward to compute in $O(1)$-time the right sibling of any vertex if it exists, and the largest child of every vertex which is not a leaf.
    Additionally, for each $P$-good tuple $(a,b,d_1,d_2)$, we have access to the maximal leftmost path~$P$, on which we can compute the vertex~$a'$ in time~$O(1)$.
    It is straightforward to see that all other computations can be done in $O(1)$-time.
\end{proof}

Combining \cref{lem:eq-P-good-separator,lem:pot-nested-type-2,lem:P-good-eq-type-2-sepn}, we can finally prove \cref{prop:comp-type-2}.

\begin{proof}[Proof of \cref{prop:comp-type-2}]
    For each maximal leftmost path~$P$, we compute a maximal nested set $\mathcal S_P$ of $P$-good tuples.
    By \cref{lem:pot-nested-type-2}, this can be done in time $O(n+m)$.
    Then, for every tuple in some~$\mathcal S_P$, we compute the encoding of at most two half-connected type-2 separations as described in \cref{lem:P-good-eq-type-2-sepn}, and put them into the set~$\mathcal S$.
    Since the sum of the sizes of all sets $\mathcal S_P$ is in $O(n+m)$, this runs in linear time.
    Finally, we output~$\mathcal S$.

    Now, we prove the correctness of the algorithm.
    By \cref{lem:P-good-eq-type-2-sepn}, every encoding we add to $\mathcal{S}$ is the encoding of a half-connected type-2 separation.
    Let~$(A',B')$ be a half-connected type-2 separation that is not in~$\mathcal S$. Let $\{a'<b'\}$ be the separator of $(A',B')$.
    By \cref{prop:charac-2-sep}, there is a maximal leftmost path~$P$ that contains both~$a'$ and~$b'$.
    By \cref{lem:eq-P-good-separator}, there are unique values for~$d_1'$ and~$d_2'$ such that $(a',b',d_1',d_2')$ is $P$-good.
    First, assume that $(a',b',d_1',d_2')$ is nested with every tuple in~$\mathcal S_P$. 
    Then $(a',b',d_1',d_2') \in \mathcal S_P$ since $\mathcal S_P$ is a maximal nested set. 
    Then, however, the encoding of~$(A',B')$ is in $\mathcal S$ since we applied \cref{lem:P-good-eq-type-2-sepn} to~$(a',b',d_1',d_2')$. 
    Otherwise, there is a tuple $(a,b,d_1,d_2) \in \mathcal S_P$ that crosses $(a',b',d_1',d_2')$; that is, either $a<a'<b<b'$ or $a'<a<b'<b$.
    By \cref{lem:eq-P-good-separator} there is a half-connected type-2 separation $(A,B)$ with separator~$\{a,b\}$, whose encoding we added to~$\mathcal S$.
    Since the tuples $(a,b,d_1,d_2)$ and $(a',b',d_1',d_2')$ cross, the type-2 separations $(A,B)$ and~$(A',B')$ also cross by \cref{lem:crossing-vertices}.
    So, $(A',B')$ is crossed by a separation in~$\mathcal S$.
\end{proof}

The proof of \cref{thm:compute-pot-nested} is now straightforward.

\begin{proof}[Proof of \cref{thm:compute-pot-nested}]
    Combine \cref{lem:comp-type-1-triples} with \cref{lem:comp-type-1} to get the set of all half-connected type-1 separations in time $O(n+m)$, and use \cref{prop:comp-type-2} to get a set of half-connected type-2 separations that contains all totally-nested type-2 separations of~$G$.
\end{proof}
\section{The structure of totally-nested 2-separations} \label{sec:stru-tot-nest}

In this section, we give a structural characterization of which half-connected 2-separations of a 2-connected graph are totally-nested. 

We first prove that a half-connected type-1 separation is never crossed by another type-1 separation.

\begin{lemma} \label{lem:type1-nested}
    Let $G$ be a 2-connected graph equipped with a normal spanning tree $T$. 
    Suppose that the vertices of $G$ are numbered with a numbering compatible with $T$.
    Let $(A, B)$ be a half-connected type-1 separation and $(A', B')$ be any type-1 separation. Then, $(A, B)$ and $(A', B')$ are nested.
\end{lemma}

\begin{proof}
    By contradiction, suppose that $(A, B)$ and $(A', B')$ cross.
    Let $\{a < b\} = A \cap B$ and $\{a' < b'\} = A' \cap B'$.
    Let $\mathcal{D}$ be the set of children $d$ of $b$ such that $lwpt(d) = a$ and $lwpt_2(d) = b$.
    By \cref{prop:charac-2-sep}, there exists a subset $\mathcal{C}$ of $\mathcal{D}$ such that $B = \{a, b\} \cup \bigcup_{c \in \mathcal{C}} Desc(c)$ and $A = V(G) \setminus (B \setminus \{a, b\})$.
    Since $(A, B)$ is half-connected, by \cref{lem:crossing-vertices}, $(A, B)$ separates $a'$ and $b'$.
    Then, one of $a', b'$ is in $A \setminus B$ and the other in $B \setminus A$.
    Since $a' < b'$ then $a'$ is an ancestor of $b'$ by \cref{lem:a-b-comp}. Thus, we must have $a' \in A \setminus B$ and $b' \in B \setminus A$.
    There exists a child $c \in \mathcal{D}$ of $b$ such that $b' \in Desc(c)$. 
    Since $a' \notin B$ is an ancestor of $b'$ then $a'$ is a proper ancestor of $b$, different from $a$.
    Since $(A', B')$ is type-1, by \cref{prop:charac-2-sep} there exists a child $c'$ of $b'$ such that $lwpt(c') = a'$ and $lwpt_2(c') = b'$. 
    This means that some descendant of $c'$, hence of $c$, is adjacent to $a' < b$, which contradicts that $(lwpt(c), lwpt_2(c)) = (a, b)$.
\end{proof}

Thus, if a half-connected type-1 separation is crossed, it must be crossed by a type-2 separation. Furthermore, since it is half-connected it must separate the separator of this type-2 separation by \cref{lem:crossing-vertices}.
Consider a half-connected type-1 separation with separator $\{a < b\}$. To know whether this separation is crossed, we would then like to be able to answer the following question: Does there exist a type-2 separation with separator $\{a' < b'\}$ such that $a < a' < b < b'$ in the tree order?
Intuitively, if there exists such a separation, there should exist a minimal one, i.e. one with $a'$ maximal and $b'$ minimal, and furthermore such a minimal pair $(a', b')$ should depend only on $b$ and not on $a$. In the following definition, one should think of $\alpha(b)$ as the maximal such $a'$ and of $\beta(b)$ as the minimal such $b'$.

\begin{definition}
    Let $G$ be a 2-connected graph equipped with a normal spanning tree $T$ with root $r$.
    Suppose that the vertices of $G$ are numbered with a numbering compatible with $T$.
    Let $b \in V(G)$. If $b \neq r$ and $b$ is not a leaf of~$T$, let $\ell$ be the left child of $b$ and set $\alpha(b) = stab(p(b), \ell)$.
    There exists a leftmost descendant $\beta$ of $\ell$ such that either $\beta$ is a leaf of $T$ or the left child $\ell'$ of $\beta$ satisfies either $lwpt(\ell') \geq b$ or $hgpt(\ell') < b$. Let $\beta(b)$ be the minimum such vertex.
    If $b=r$ or $b$ is a leaf of $T$, we set $\alpha(b) = \beta(b) = 0$.
\end{definition}

With these definitions, we can now formalize the intuition given earlier. The next lemma and its proof are supported by \cref{fig:crossing-type-2-above}.

\begin{lemma} \label{lem:find-crossing-type-2-above}
    Let $G$ be a 2-connected graph equipped with a normal spanning tree $T$.
    Suppose that the vertices of $G$ are numbered with a numbering compatible with $T$.
    Let $a, b \in V(G)$ be such that $a$ is a proper ancestor of $b$ in $T$ and $b$ is not a leaf.
    Let $\ell$ be the left child of $b$. Suppose that there is no edge between $Desc(\ell)$ and $T(a, b)$.
    Let $\alpha = \alpha(b)$ and $\beta = \beta(b)$.
    There exists a type-2 separation with separator $\{a' < b'\}$ such that $a < a' < b < \ell \leq b'$ in the tree order if and only if $\alpha > a$, every child $c$ of $\beta$ satisfies either $lwpt(c) \geq \alpha$ or $hgpt(c) \leq \alpha$, and $T[\alpha, \beta]$ is stable.
\end{lemma}

\begin{figure}[h]
    \centering
    \includegraphics[width=0.25\linewidth]{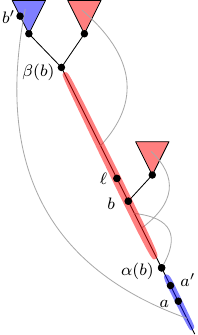}
    \caption{The separation~$(A,B)$ with separator $\{\alpha(b), \beta(b)\}$ constructed in the proof of \cref{lem:find-crossing-type-2-above}. In this figure, the proper sides of~$(A,B)$ are coloured in blue and red, respectively.}
    \label{fig:crossing-type-2-above}
\end{figure}

\begin{proof}
    Suppose first that there exists a type-2 separation with separator $\{a' < b'\}$ such that $a < a' < b < \ell \leq b'$ in the tree order.
    Then, $T[a', b']$ is a leftmost path and is stable by \cref{prop:charac-2-sep}. Thus, $b'$ is a proper leftmost descendant of $b$, hence a leftmost descendant of $\ell$. Thus, $T[a', \ell]$ is a leftmost path and is stable by restriction. Therefore we have $\alpha = stab(p(b), \ell) \geq a' > a$.

    We now show that if $b'$ has a left child $\ell'$ then either $lwpt(\ell') \geq b$ or $hgpt(\ell') < b$.
    Since $\{a', b'\}$ is the separator of a type-2 separation, by \cref{prop:charac-2-sep} either $lwpt(\ell') \geq a'$ or $hgpt(\ell') \leq a' < b$. Furthermore, there is no edge between $Desc(\ell)$ and $T(a, b)$ by assumption so if $lwpt(\ell') \geq a'$ then actually $lwpt(\ell') \geq b$. Therefore, either $lwpt(\ell') \geq b$ or $hgpt(\ell') < b$.
    This proves that $\beta$ is an ancestor of $b'$. 
    If $\beta = b'$, by \cref{prop:charac-2-sep} every child $c$ of $\beta$ satisfies either $hgpt(c) \leq a' \leq \alpha$ or $lwpt(c) \geq a'$, which implies $lwpt(c) \geq b \geq \alpha$ since there is no edge between $Desc(\ell)$ and $T(a, b)$.
    Suppose now that $\beta \neq b'$. Then, $\beta \in T(a', b')$ so $witn(\beta) \geq a'$ since $T[a', b']$ is stable. There is no edge between $Desc(\ell)$ and $T(a, b)$ so $witn(\beta) \geq b$. Therefore, for every child $c$ of $\beta$ apart from its left child $\ell'$, we have $lwpt(c) \geq witn(\beta) \geq b \geq \alpha$. Finally, by definition of $\beta$, either $lwpt(\ell') \geq b \geq \alpha$ or $hgpt(\ell') < b$, which implies $hgpt(\ell') \leq a \leq \alpha$ since there is no edge between $Desc(\ell)$ and $T(a, b)$.
    
    We finally show that $T[\alpha, \beta]$ is stable, using \cref{lem:witn-stab}.
    $T[\alpha, \beta]$ is a leftmost path since it is a restriction of the leftmost path $T[a', b']$, and since $\beta$ is a proper descendant of $b$, which is a proper descendant of $\alpha$. Let $\gamma \in T(\alpha, \beta)$. If $\gamma \in T(\alpha, \ell)$ then $witn(\gamma) \geq \alpha$ since $T[\alpha, \ell]$ is stable. If not, then $\gamma \in T(b, \beta)$ and therefore $\gamma \in Desc(\ell) \cap T(a', b')$ so $witn(\gamma) \geq a'$ and since there is no edge between $Desc(\ell)$ and $T(a, b)$, we eventually get $witn(\gamma) \geq b > \alpha$. Therefore, $T[\alpha, \beta]$ is stable.

    Conversely, suppose that $\alpha > a$, every child $c$ of $\beta$ satisfies either $lwpt(c) \geq \alpha$ or $hgpt(c) \leq \alpha$, and $T[\alpha, \beta]$ is stable. Then, $a < \alpha < b < \ell \leq \beta$ in the tree order by assumption and by definition of $\alpha$ and $\beta$. Furthermore, it follows from \cref{prop:charac-2-sep} and from our assumptions that there exists a type-2 separation of $G$ with separator $\{\alpha < \beta\}$.
\end{proof}

We can now characterize the half-connected type-1 separations which are totally-nested.

\begin{proposition} \label{lem:charac-totally-nested-type-1}
    Let $G$ be a 2-connected graph equipped with a normal spanning tree $T$.
    Suppose that the vertices of $G$ are numbered with a numbering compatible with $T$.
    Let $(A, B)$ be a half-connected type-1 separation of $G$ with separator $\{a < b\}$.
    Let $\mathcal{D}$ be the set of children $d$ of $b$ such that $lwpt(d) = a$ and $lwpt_2(d) = b$.
    Let $\ell$ be the left child of $b$.
    Let $\alpha = \alpha(b)$ and $\beta = \beta(b)$.
    Then, $(A, B)$ is crossed if and only if all of the following holds:
    \begin{itemize}
        \item Up to exchanging $A$ and $B$, either $B = \{a, b\} \cup Desc(\ell)$ or $B = \{a, b\} \cup \bigcup_{d \in \mathcal{D}} Desc(d)$;
        \item $\ell \in \mathcal{D}$;
        \item $\alpha > a$, every child $c$ of $\beta$ satisfies either $lwpt(c) \geq \alpha$ or $hgpt(c) \leq \alpha$, and $T[\alpha, \beta]$ is stable.
    \end{itemize}
\end{proposition}

\begin{proof}
    Suppose first that $(A, B)$ is crossed by a separation $(A', B')$ with separator $\{a' < b'\}$. 
    By \cref{lem:type1-nested}, $(A', B')$ is a type-2 separation. 
    Since $(A, B)$ is a half-connected type-1 separation, by \cref{prop:charac-t1-half-con}, up to exchanging $A$ and $B$, either $B = \{a, b\} \cup \bigcup_{d \in \mathcal{D}} Desc(d)$ or there exists $d \in \mathcal{D}$ such that $B = \{a, b\} \cup Desc(d)$.
    Since $(A, B)$ is half-connected then $(A', B')$ separates $a$ and $b$, and $(A, B)$ separates $a'$ and $b'$ by \cref{lem:crossing-vertices}.
    In particular, we have $\{a, b\} \cap \{a', b'\} = \emptyset$.
    Since $a'$ is an ancestor of $b'$ by \cref{lem:a-b-comp}, we must have $a' \in A \setminus B$ and $b' \in B \setminus A$.
    Thus, there exists $c \in \mathcal{D}$ such that $b' \in Desc(c)$, and $a' < b < c \leq b'$ in the tree order.
    Since $(A', B')$ is a separation which separates $a$ from $b \in T(a', b')$, we must have $a \notin T(a', b')$, so $a < a' < b < c \leq b'$ in the tree order. 
    Furthermore, since $(A', B')$ is a type-2 separation then $T[a', b']$ is a leftmost path by \cref{prop:charac-2-sep}, so $c$ is the left child of $b$, i.e. $c = \ell$. Thus, $\ell = c \in \mathcal{D}$, so there is no edge between a descendant of $\ell$ and $T(a, b)$. Therefore, \cref{lem:find-crossing-type-2-above} implies that $\alpha > a$, every child $c$ of $\beta$ satisfies either $lwpt(c) \geq \alpha$ or $hgpt(c) \leq \alpha$, and $T[\alpha, \beta]$ is stable.
    
    We now prove the converse. Suppose that the following holds:
    Up to exchanging $A$ and $B$, either $B = \{a, b\} \cup Desc(\ell)$ or ${B = \{a, b\} \cup \bigcup_{d \in \mathcal{D}} Desc(d)}$. Furthermore, $\ell \in \mathcal{D}$, $\alpha > a$, every child $c$ of $\beta$ satisfies either $lwpt(c) \geq \alpha$ or $hgpt(c) \leq \alpha$, and $T[\alpha, \beta]$ is stable.
    Since $\ell \in \mathcal{D}$, there is no edge between a descendant of $\ell$ and $T(a, b)$. Thus, \cref{lem:find-crossing-type-2-above} implies that there is a type-2 separation $(A', B')$ with separator $\{a' < b'\}$ such that $a < a' < b < \ell \leq b'$ in the tree order. 
    Let $r$ be the root of $T$.
    By \cref{prop:charac-2-sep}, up to exchanging $A'$ and $B'$, we have $T(a', b') \subseteq B' \setminus A'$ and $T[r, a') \subseteq A' \setminus B'$.
    Then, we have that $(A, B)$ separates $a'$ and $b'$, and $(A', B')$ separates $b \in T(a', b')$ and $a \in T[r, a')$. Thus, the separations $(A, B)$ and $(A', B')$ cross by \cref{lem:crossing-vertices}.
\end{proof}

We will use the next result to characterize when a half-connected type-2 separation is crossed by a type-1 separation.
The next lemma and its proof are supported by \cref{fig:t2-crossed-by-t1}.

\begin{lemma} \label{lem:t2-crossed-by-t1}
    Let $G$ be a 2-connected graph equipped with a normal spanning tree $T$.
    Suppose that the vertices of $G$ are numbered with a numbering compatible with $T$.
    Let $(A, B)$ be a half-connected type-2 separation with separator $\{a < b\}$. The following are equivalent.
    \begin{itemize}
        \item There exists a type-1 separation $(A', B')$ with separator $\{a'' < b''\}$ such that $a'' < a < b'' < b$ in the tree order, and $(A', B')$ separates $a$ and $b$.
        \item There exists $\gamma \in T(a, b)$ such that the left child $\ell'$ of $\gamma$ satisfies $lwpt_2(\ell') = \gamma$.
    \end{itemize}
\end{lemma}

\begin{figure}[h]
    \centering
    \includegraphics[width=0.25\linewidth]{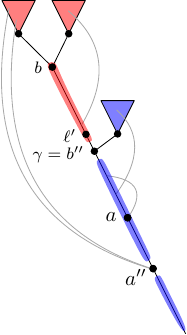}
    \caption{The type-1 separation~$(A',B')$ constructed in \cref{lem:t2-crossed-by-t1} that crosses the given type-2 separation~$(A,B)$. In this figure, the proper sides of~$(A',B')$ are coloured in red and blue, respectively.}
    \label{fig:t2-crossed-by-t1}
\end{figure}

\begin{proof}
    Suppose first that there exists a type-1 separation $(A', B')$ with separator $\{a'' < b''\}$ such that $a'' < a < b'' < b$ in the tree order, and $(A', B')$ separates $a$ and $b$.
    Then, $b'' \in T(a, b)$ and by \cref{prop:charac-2-sep} there exists a child $d$ of $b''$ such that $lwpt(d) = a''$ and $lwpt_2(d) = b''$.
    Since $b'' \in T(a, b)$ and $T[a, b]$ is stable by \cref{prop:charac-2-sep}, then $witn(b'') \geq a$. Thus, for every child $c$ of $b''$ except its left child we have $lwpt(c) \geq witn(b'') \geq a > a''$. Thus, $d$ is the left child of $b''$ and $lwpt_2(d) = b''$.
    
    Conversely, suppose that there exists $\gamma \in T(a, b)$ such that the left child $\ell'$ of $\gamma$ satisfies $lwpt_2(\ell') = \gamma$.
    Since $\{a, b\}$ is the separator of a type-2 separation then $a$ is not the root of $T$. Therefore, since $G$ is 2-connected, the unique child $a'$ of $a$ that is an ancestor of $b$ satisfies $lwpt(a') < a$ by \cref{lem:charac-1-sep}. 
    Let $a'' = lwpt(a')$. By \cref{lem:follow-leftmost}, there is a leftmost descendant $d$ of $a'$ that is adjacent to $a''$. Since $T[a, b]$ is stable, $d$ must be a descendant of $b$, hence of $\ell'$. Thus, $lwpt(\ell') \leq a''$. However, $\ell'$ is a descendant of $a'$ so $lwpt(\ell') \geq lwpt(a') = a''$ and thus $lwpt(\ell') = a''$. 
    Therefore, $\ell'$ satisfies $lwpt(\ell') = a'', lwpt_2(\ell') = \gamma$ and $a \notin \{a'', \gamma\}$ is not a descendant of $\ell'$. Set $B' = \{a'', \gamma\} \cup Desc(\ell')$ and $A' = V(G) \setminus (B' \setminus \{a'', \gamma\})$. By \cref{prop:charac-2-sep}, $(A', B')$ is a type-1 separation with separator $\{a'' < \gamma\}$. Furthermore, $a'' < a < \gamma < b$ in the tree order, and $(A', B')$ separates $a \in A'$ and $b \in B'$.
\end{proof}

We can now characterize the half-connected type-2 separations which are totally-nested.

\begin{proposition} \label{lem:charac-tot-nested-t2}
    Let $G$ be a 2-connected graph equipped with a normal spanning tree $T$ with root $r$.
    Suppose that the vertices of $G$ are numbered with a numbering compatible with $T$.
    Let $(A, B)$ be a half-connected type-2 separation with separator $\{a < b\}$. Let $a'$ be the unique child of $a$ which is an ancestor of $b$. Then, $(A, B)$ is crossed if and only if one of the following holds.
    \begin{itemize}
        \item There exists $\gamma \in T(a, b)$ such that the left child $\ell'$ of $\gamma$ satisfies $lwpt_2(\ell') = \gamma$.
        \item There exists $\gamma \in T(a, b)$ such that the left child $\ell'$ of $\gamma$ satisfies $hgpt(\ell') \leq stab(p(a), a')$.
        \item The vertex $b$ is not a leaf of $T$ and the following holds. Let $\ell$ be the left child of $b$, $\alpha = \alpha(b)$ and $\beta = \beta(b)$. Then, the separation $(A, B)$ does not separate $\ell$ and $r$, $\alpha > a$, every child $c$ of $\beta$ satisfies either $lwpt(c) \geq \alpha$ or $hgpt(c) \leq \alpha$, and $T[\alpha, \beta]$ is stable. 
    \end{itemize}
\end{proposition}

\begin{proof}
    Suppose first that $(A, B)$ is crossed, and let $(A', B')$ be a 2-separation of $G$ with separator $\{a'' < b''\}$ that crosses $(A, B)$.
    Since $(A, B)$ is half-connected, by \cref{lem:crossing-vertices}, we have that $(A, B)$ separates $a''$ and $b''$, and $(A', B')$ separates $a$ and $b$.
    In particular, this implies that $\{a, b\} \cap \{a'', b''\} = \emptyset$.
    Since $(A, B)$ is a type-2 separation, by \cref{prop:charac-2-sep}, up to exchanging $A$ and $B$, we can assume that $r \in A \setminus B$ and $T(a, b) \subseteq B \setminus A$.

    Suppose that $(A', B')$ is a type-1 separation. By \cref{prop:charac-2-sep}, there exists a non-empty subset $\mathcal{C}$ of children of $b''$ such that, up to exchanging $A'$ and $B'$, we have $B' = \{a'', b''\} \cup \bigcup_{c \in \mathcal{C}} Desc(c)$. 
    One of $a, b$ is in $A' \setminus B'$ and the other in $B' \setminus A'$. Since $a$ is a proper ancestor of $b$ then $a \in A' \setminus B'$ and $b \in B' \setminus A'$. 
    Thus, $b$ is a proper descendant of $b''$ and $a$ is a proper ancestor of $b''$, i.e. $a < b'' < b$ in the tree order.
    However, $(A, B)$ separates $b'' \in T(a, b) \subseteq B \setminus A$ and $a''$, which is an ancestor of $b''$, so $a'' \in T[r, a)$ and therefore $a'' < a < b'' < b$ in the tree order. \cref{lem:t2-crossed-by-t1} then implies that there exists $\gamma \in T(a, b)$ such that the left child $\ell'$ of $\gamma$ satisfies $lwpt_2(\ell') = \gamma$.

    Suppose now that $(A', B')$ is a type-2 separation. 
    Suppose first that $a'' \in A \setminus B$ and $b'' \in B \setminus A$. Since $a''$ is an ancestor of $b''$, by \cref{prop:charac-2-sep}, we must have that $a'' \in T[r, a)$ and $b''$ is a descendant of $a'$, the unique child of $a$ which is an ancestor of $b$.
    By \cref{prop:charac-2-sep}, $T[a'', b'']$ is stable, so $T[a'', a']$ is stable by restriction (it is indeed a leftmost path since $a'$ is a proper descendant of $a$, which is itself a proper descendant of $a''$). 
    Thus, $a'' \leq stab(p(a), a')$.
    Since $(A', B')$ is a type-2 separation then $a'' \neq r$. Let $a_0$ be the unique child of $a''$ that is an ancestor of $b''$. Since $G$ is 2-connected, $lwpt(a_0) < a''$ by \cref{lem:charac-1-sep}. 
    By \cref{lem:follow-leftmost}, there exists a leftmost descendant $d$ of $a_0$ that is adjacent to a proper ancestor of $a''$, hence of $a$.
    Since $T[a'', b'']$ is stable, $d$ must be a (leftmost) descendant of $b''$.
    We now prove that $b'' \in T(a, b)$.
    By contradiction, suppose that $b'' \notin T(a, b)$. 
    Since $a'' < a' < b''$ in the tree order, then $a' \in T(a'', b'')$, which is a leftmost path, so $b''$ is a leftmost descendant of $a'$. Since $b'' \notin T(a, b)$ then $b''$ is a proper descendant of $b$ and thus $b \in T(a'', b'')$. Let $b'$ be the unique child of $b$ that is an ancestor of $b''$. Since $b'' \in B \setminus A$ then $b' \in B \setminus A$ and thus $lwpt(b') \geq a$ by \cref{prop:charac-2-sep}. 
    However, $d$ is a descendant of $b''$, hence of $b'$, and is adjacent to a proper ancestor of $a$, a contradiction.
    This proves that $b'' \in T(a, b)$.
    Let $b_0$ be the left child of $b''$. By \cref{prop:charac-2-sep}, either $lwpt(b_0) \geq a''$ or $hgpt(b_0) \leq a''$. Recall that there exists a leftmost descendant $d$ of $b''$ that is adjacent to a proper ancestor of $a''$. Since $b'' \in T(a, b)$ and $T[a, b]$ is stable, then $d$ must be a proper descendant of $b''$. Thus, $d$ is a descendant of $b_0$, therefore $lwpt(b_0) < a''$, and thus $hgpt(b_0) \leq a''$.
    Thus, $b''$ is a vertex of $T(a, b)$ and its left child $b_0$ satisfies $hgpt(b_0) \leq a'' \leq stab(p(a), a')$.

    Otherwise, we have that $a'' \in B \setminus A$ and $b'' \in A \setminus B$. Since $a''$ is an ancestor of $b''$, by \cref{prop:charac-2-sep} we must have that $a'' \in T(a, b)$ and $b''$ is a proper descendant of $b$, so $a < a'' < b < b''$ in the tree order. Thus, $b$ is not a leaf. Let $\ell$ be the left child of $b$, $\alpha = \alpha(b)$ and $\beta = \beta(b)$. $T[a'', b'']$ is stable, hence a leftmost path and thus $b''$ is a descendant of $\ell$. Since $b'' \in A \setminus B$ then $\ell \in A \setminus B$ and thus by \cref{prop:charac-2-sep}, $hgpt(\ell) \leq a$ so there is no edge between a descendant of $\ell$ and $T(a, b)$. Therefore, \cref{lem:find-crossing-type-2-above} implies that $\alpha > a$, every child $c$ of $\beta$ satisfies either $lwpt(c) \geq \alpha$ or $hgpt(c) \leq \alpha$, and $T[\alpha, \beta]$ is stable.
    
    We now prove the converse implication, and consider the three cases separately.
    \begin{itemize}
        \item Suppose that there exists $\gamma \in T(a, b)$ such that the left child $\ell'$ of $\gamma$ satisfies $lwpt_2(\ell') = \gamma$. 
        Then, \cref{lem:t2-crossed-by-t1} implies that there exists a type-1 separation $(A', B')$ with separator $\{a'' < b''\}$ such that $a'' < a < b'' < b$ in the tree order, and $(A', B')$ separates $a$ and $b$. 
        By \cref{prop:charac-2-sep}, we also have that $(A, B)$ separates $a'' \in T[r, a)$ and $b'' \in T(a, b)$, so $(A, B)$ and $(A', B')$ cross by \cref{lem:crossing-vertices}.
        \item Suppose that there exists $\gamma \in T(a, b)$ such that the left child $\ell'$ of $\gamma$ satisfies $hgpt(\ell') \leq stab(p(a), a') =: a''$.
        Therefore, $a'' > 0$ so $T[a'', a']$ is stable and $a''$ is a proper ancestor of $a$. We have that $T[a'', a']$ is stable, $T[a, \gamma]$ is stable (by restriction of $T[a, b]$, which is stable too) and $a'' < a < a' \leq \gamma$ in the tree order, therefore $T[a'', \gamma]$ is stable. 
        Since $T[a, b]$ is stable and $\gamma \in T(a, b)$ then $witn(\gamma) \geq a$. 
        Thus, for every child $c$ of $\gamma$ apart from its left child $\ell'$, we have $lwpt(c) \geq witn(\gamma) \geq a \geq a''$. 
        By assumption, we have $hgpt(\ell') \leq a''$.
        If $a''$ is the root of $T$, then $lwpt(\ell') = a''$ and $lwpt_2(\ell') = \gamma$, and $a \notin \{a'', \gamma\}$ is not a descendant of $\ell'$. Thus, it follows from \cref{prop:charac-2-sep} that there is a type-1 separation $(A', B')$ of $G$ with separator $\{a'' < \gamma\}$ that separates $a \in A'$ and $b \in B'$. 
        Otherwise, it follows from \cref{prop:charac-2-sep} that there is a type-2 separation $(A', B')$ of $G$ with separator $\{a'' < \gamma\}$ that separates $a \in B'$ and $b \in A'$. 
        Since $(A, B)$ separates $a'' \in T[r, a)$ and $b'' \in T(a, b)$, it follows from \cref{lem:crossing-vertices} that $(A, B)$ and $(A', B')$ cross.
        \item Suppose that $b$ is not a leaf of $T$. Let $\ell$ be the left child of $b$, $\alpha = \alpha(b)$ and $\beta = \beta(b)$. Assume that $\ell \in A \setminus B$, $\alpha > a$, every child $c$ of $\beta$ satisfies either $lwpt(c) \geq \alpha$ or $hgpt(c) \leq \alpha$, and $T[\alpha, \beta]$ is stable.
        Since $\ell \in A \setminus B$ then $hgpt(\ell) \leq a$ by \cref{prop:charac-2-sep}, so there is no edge between a descendant of $\ell$ and $T(a, b)$. Thus, \cref{lem:find-crossing-type-2-above} implies that there exists a type-2 separation $(A', B')$ with separator $\{a'' < b''\}$ such that $a < a'' < b < \ell \leq b''$ in the tree order. Since $\ell \in A \setminus B$ then $b'' \in A \setminus B$ so and $(A, B)$ separates $a''$ and $b''$ by \cref{prop:charac-2-sep}.
        We also have that $(A', B')$ separates $a \in T[r, a'')$ from $b \in T(a'', b'')$ by \cref{prop:charac-2-sep}. 
        Therefore, the separations $(A, B)$ and $(A', B')$ cross by \cref{lem:crossing-vertices}. \qedhere
    \end{itemize}
\end{proof}
\section{Building the Tutte-decomposition} \label{sec:build-tutte}

In this section, we first show how to turn the structural characterizations from \cref{sec:stru-tot-nest} into an algorithm to efficiently test whether a half-connected 2-separation is totally-nested. Then, we explain how to construct the tree-decomposition induced by a set of nested separations, and apply it to build the Tutte-decomposition from the set of all totally-nested 2-separations.

\subsection{Checking for totally-nestedness}

Throughout this section, we assume the precomputation described in \cref{subsec:precomputations,,subsec:stability} has already been performed.
The main result of this section is the following.

\begin{theorem} \label{thm:query-tot-nested}
    There is an algorithm which, given a 2-connected graph $G$ equipped a spanning tree $T$ and a numbering of $V(G)$ compatible with $T$, after an $O(n+m)$-time precomputation, can answer in constant time queries of the form: Given a half-connected 2-separation, is it totally-nested?
\end{theorem}

This result is an algorithmic consequence of \cref{lem:charac-totally-nested-type-1,,lem:charac-tot-nested-t2}.
We now describe some precomputation we need to perform to be able to use these characterizations efficiently.

\begin{lemma} \label{lem:compute-alpha-beta}
    There is an $O(n+m)$-time algorithm which, given a 2-connected graph $G$ equipped a spanning tree $T$ and a numbering of $V(G)$ compatible with $T$, computes $\alpha(v), \beta(v)$ for every $v \in V(G)$.
\end{lemma}

\begin{proof}
    If $v$ is a leaf or the root of $T$ then $\alpha(v) = 0$ by definition.
    Otherwise, $\alpha(v) = stab(p(v), \ell(v))$, where $\ell(v)$ denotes the left child of $v$.
    For every $v \in V(G)$ which is not the root or a leaf of $T$, we can access $p(v)$ and $\ell(v)$ in constant time. Then, computing $stab(p(v), \ell(v))$ also takes constant time by \cref{prop:queries-stability}.
    The computation of all $\alpha(v)$ thus takes time $O(n)$ overall.

    We compute the $\beta(v)$ separately on each maximal leftmost path. 
    By definition, $\beta(r) = 0$. 
    Observe that if $P$ is a leftmost path then for every $v \in V(P)$, $\ell(v)$ (the left child of $v$) is the successor of $v$ in $P$ (except if $v$ is the first vertex of $P$).
    For each leftmost path $P$, we compute for every $v \in V(P)$ (except the first vertex) the minimum vertex $\beta_1(v) \geq \ell(v)$ in $V(P)$ such that either $\beta_1(v)$ is a leaf of $T$ (i.e. the last vertex of $P$) or $lwpt(\ell(\beta_1(v))) \geq v$.
    Then, we compute for every $v \in V(P)$ (except the first vertex) the minimum vertex $\beta_2(v) \geq \ell(v)$ in $V(P)$ such that either $\beta_2(v)$ is a leaf of $T$ or $hgpt(\ell(\beta_2(v))) < v$.
    Then, for every $v \in V(P)$, we will have $\beta(v) = \min(\beta_1(v), \beta_2(v))$.

    We now explain how to compute $\beta_1(\cdot)$ on a maximal leftmost path $P = v_0v_1\ldots v_k$. The details are given in \cref{alg:compute-alpha-beta}. We use a queue $Q$ to store all vertices below our current vertex for which $\beta_1$ has not yet been computed. $Q$ will always be sorted with the smallest elements first. 
    When we process a new vertex $v_i \in V(P)$, if $v_i = v_k$ then $\beta_1(v_j) = v_k$ for all vertices $v_j \in Q$. If not, we keep removing the first element $v_j$ of $Q$ and setting $\beta_1(v_j) = v_i$ as long as $lwpt(v_{i+1}) \geq v_j$. Once this is done, we add $v_i$ as the last element of $Q$.

    \begin{algorithm}
    \caption{Algorithm to compute $\beta_1$}
    \label{alg:compute-alpha-beta}
    \begin{algorithmic}[1]
    \Require Maximal leftmost path $P = v_0v_1\ldots v_k)$.
    \Require $lwpt(v)$ for every $v \in V(P)$.

    \State $Q \gets$ empty queue
    \For{$i \in [1, k-1]$}
        \While{$Q \neq \emptyset$ and $lwpt(v_{i+1}) \geq Q.first()$} \label{line:while-beta1}
            \State $\beta_1(Q.first()) \gets v_i$ \label{line:set-beta1}
            \State $Q.pop()$
        \EndWhile
        \State $Q.push(v_i)$
    \EndFor

    \While{$Q \neq \emptyset$} \label{line:while-Q-not-empty}
        \State $\beta_1(Q.first()) \gets v_k$ \label{line:set-beta1'}
        \State $Q.pop()$
    \EndWhile
    \State $\beta_1(v_k) \gets 0$
    
    \end{algorithmic}
    \end{algorithm}

    Since the $v_i$ are inserted into $Q$ by increasing order and since we always ever remove the first element of $Q$ then $Q$ remains sorted (with the smallest elements first) throughout the algorithm.

    \begin{claim*}
        At any point of the algorithm, the following holds.
        \begin{itemize}
            \item If $v_j \in Q$ when we start processing $v_i$ then $v_j < v_i$ and $\beta_1(v_j) \geq v_i$.
            \item Whenever we remove some $v_j$ from $Q$, we correctly computed $\beta_1(v_j)$.
        \end{itemize}
    \end{claim*}

    \begin{subproof}
        The two invariants trivially hold when we start processing $v_1$ since $Q$ is initialized empty.
        Let $i \in [1, k - 1]$ and suppose that the two invariants hold when we start processing $v_i$. We show that they still hold once we are done processing $v_i$ and start processing $v_{i+1}$.
        The only element we add to $Q$ while processing $v_i$ is $v_i$, so when we start processing $v_{i+1}$, every $v_j$ in $Q$ will satisfy $v_j \leq v_i < v_{i+1}$. This proves the first part of the first invariant.
        Observe that $\beta_1(v_i) \geq v_{i+1}$ by definition.
        Consider a $v_j$ such that $\beta_1(v_j) < v_{i+1}$ and $v_j \in Q$ when we start processing $v_i$. By the first invariant, we have $v_i \leq \beta_1(v_j) < v_{i+1}$, hence $\beta_1(v_j) = v_i$ since $\beta_1(v_j)$ is a leftmost descendant of $v_j$ hence a vertex of $P$. Therefore, since $v_i$ is not a leaf of $T$ (since $i < k$), we have $lwpt(v_{i+1}) \geq v_j$.
        When we stop entering the While loop at \cref{line:while-beta1}, either $Q$ is empty or the first vertex $v_l$ of $Q$ satisfies $lwpt(v_{i+1}) < v_l$, hence $v_j < v_l$. 
        Since $Q$ remains sorted with the smallest elements first, in both cases we must have removed $v_j$ from $Q$ when processing $v_i$.
        This proves that when we are done processing $v_i$, every $v_l \in Q$ satisfies $\beta_1(v_l) > v_i$, hence $\beta_1(v_l) \geq v_{i+1}$.
        Finally, consider any $v_j$ which was removed from $Q$ when processing $v_i$. 
        $v_j$ was in $Q$ when we started processing $v_i$ so $\beta_1(v_j) \geq v_i$ by the first invariant. Furthermore, since we removed $v_j$ from $Q$ it must be that $lwpt(v_{i+1}) \geq v_j$ so $\beta_1(v_j) \leq v_i$ and thus $\beta_1(v_j) = v_i$, which is exactly what we set at \cref{line:set-beta1}.
        Finally, when we process $v_k$, every $v_j \in Q$ at that point satisfies $\beta_1(v_j) \geq v_k$ by the first invariant, and therefore $\beta_1(v_j) = v_k$ since $v_k$ is a leaf of $T$ (otherwise the maximal leftmost path $P$ would not end at $v_k$), which is exactly what we set at \cref{line:set-beta1'}.
    \end{subproof}

    For every $i \in [1, k-1]$, $v_i$ is added to $Q$ when processing $v_i$. Since $Q$ is empty at the end of the algorithm because of \cref{line:while-Q-not-empty}, every $v_i$ is removed from $Q$ at some point. By the second invariant, we thus correctly computed every $\beta_1(v_i)$. Finally, since $v_k$ is a leaf of $T$, we also correctly computed $\beta_1(v_k)$.

    As for the running time of this algorithm, every $v_i$ gets inserted exactly once in $Q$, and the time spent while processing $v_i$ is $O(1) + O(r_i)$ where $r_i$ is the number of elements we removed from $Q$ while processing $v_i$. Summing over all $v_i$, the running time is therefore $O(|V(P)|)$.

    We use the same idea to compute $\beta_2(\cdot)$ on $P$. The algorithm is given in \cref{alg:compute-alpha-beta2}. The main difference is that here we use a stack $S$ instead of a queue. In this stack, we again store all vertices below our current vertex for which $\beta_2$ has not yet been computed (but this time with the largest on top). When we process a new vertex $v_i \in V(P)$, if $v_i = v_k$ then $\beta_2(v_j) = v_k$ for all vertices $v_j \in S$. If not, we keep removing the first element $v_j$ of $S$ and setting $\beta_2(v_j) = v_i$ while $hgpt(v_{i+1}) < v_j$.

    \begin{algorithm}
    \caption{Algorithm to compute $\beta_2$}
    \label{alg:compute-alpha-beta2}
    \begin{algorithmic}[1]
    \Require Maximal leftmost path $P = (v_0 < v_1 < \ldots < v_k)$.
    \Require $hgpt(v)$ for every $v \in V(P)$.
    \State $S \gets$ empty stack
    
    \For{$i \in [1, k-1]$}
        \While{$S \neq \emptyset$ and $hgpt(v_{i+1}) < S.top()$}
            \State $\beta_2(S.top()) \gets v_i$
            \State $S.pop()$
        \EndWhile
        \State $S.push(v_i)$
    \EndFor
    
    \While{$S \neq \emptyset$}
        \State $\beta_2(S.top()) \gets v_k$
        \State $S.pop()$
    \EndWhile
    \State $\beta_2(v_k) \gets 0$
    
    \end{algorithmic}
    \end{algorithm}

    The complexity and correctness analysis of \cref{alg:compute-alpha-beta2} are similar to that of \cref{alg:compute-alpha-beta}.
    
    Finally, we set $\beta(v) = \min(\beta_1(v), \beta_2(v))$ for every $v \in V(P)$ except the first vertex. 
    This correctly computes all values of $\beta(v)$ for $v \in V(P) \setminus \{v_0\}$ in time $O(|V(P)|)$.

    Since for every vertex $v$ (except the root $r$) there exists a maximal leftmost path for which $v$ is not the first vertex, we compute this way all values of $\beta(v)$ for $v \in V(G) \setminus \{r\}$.
    Since the maximal leftmost paths partition the edges of $T$ and since we already computed this partition, the total running time to compute all $\beta(v)$ is $O(n)$.
\end{proof}

We now assume that we computed $\alpha(v)$ and $\beta(v)$ for every $v \in V(G)$.

\begin{lemma} \label{lem:compute-Z}
    There is an $O(n+m)$-time algorithm which, given a 2-connected graph $G$ equipped a spanning tree $T$ and a numbering of $V(G)$ compatible with $T$, computes a table $Z$ such that for every $v \in V(G)$, we have $Z[v] = 1$ if and only if the following holds: $\alpha(v) \neq 0$, $\beta(v) \neq 0$, every child $c$ of $\beta(v)$ satisfies either $lwpt(c) \geq \alpha(v)$ or $hgpt(c) \leq \alpha(v)$ and $T[\alpha(v), \beta(v)]$ is stable.
\end{lemma}

\begin{proof}
    We initialize our table $Z$ by setting $Z[v] = 1$ for every $v \in V(G)$.
    For every $v \in V(G)$ such that $\alpha(v) = 0$ or $\beta(v) = 0$, we set $Z[v] = 0$.
    For all others $v \in V(G)$, we query $stab(\alpha(v), \beta(v))$ using \cref{prop:queries-stability}. Note that $T[\alpha(v), \beta(v)]$ is stable if and only if $stab(\alpha(v), \beta(v)) = \alpha(v)$. Thus, whenever $stab(\alpha(v), \beta(v)) \neq \alpha(v)$, we set $Z[v] = 0$.
    
    Let $X$ be the set of all triples $(\beta(v), \alpha(v), v)$ for which $Z[v] = 1$ at this point. 
    By \cref{cor:sort-lex}, we can sort $X$ lexicographically in time $O(n)$.
    We now iterate over all triples $(\beta(v), \alpha(v), v) \in X$.
    For every fixed $\beta$, we consider all triples of the form $(\beta, \cdot, \cdot)$ together. Write them as $(\beta, \alpha_1, v_1), \ldots, (\beta, \alpha_k, v_k)$ by increasing lexicographical order.
    We go over all these triples while simultaneously going over all children of $\beta$ by increasing order.
    If $\beta$ has no child, we do nothing.
    Otherwise, we consider the first triple $(\beta, \alpha_1, v_1)$. We now go over the children $c$ of $\beta$ until we find the smallest child $c_1$ such that $lwpt(c_1) < \alpha_1$. Then, if $maxhgpt(c_1) > \alpha_1$, we set $Z[v_1] = 0$. Indeed, in that case there is a child $d \geq c_1$ of $\beta$ such that $hgpt(d) > \alpha_1$. Furthermore, since $d \geq c_1$ then $lwpt(d) \leq lwpt(c_1) < \alpha_1$. On the other hand, if $maxhgpt(c_1) \leq \alpha_1$ then every child $d$ of $\beta$ indeed satisfies either $lwpt(d) \geq \alpha_1$ or $hgpt(d) \leq \alpha_1$ (by minimality of $c_1$).
    We now want to do the same for $(\beta, \alpha_2, v_2)$. The key observation here is that since $\alpha_2 \geq \alpha_1$, the first child $c_2$ such that $lwpt(c_2) < \alpha_2$ satisfies $c_2 \geq c_1$. Thus, to find $c_2$ we do not need to start over from the smallest child of $\beta$, we can just start from $c_1$.
    This way, the time needed to find all $c_i$ will be proportional to the number of children of $\beta$ in $T$, and to $k$. Therefore, summing over all values of $\beta$, this takes time $O(n)$ overall.
    The correctness of this algorithm is straightforward to check.
\end{proof}

We now assume that we computed this table $Z$.

\begin{proposition} \label{prop:query-totally-nested-1}
    There is an algorithm which, given a 2-connected graph $G$ equipped a spanning tree $T$ and a numbering of $V(G)$ compatible with $T$, after an $O(n+m)$-time precomputation, can answer in constant time queries of the form: Given a half-connected type-1 separation, is it totally-nested?
\end{proposition}

\begin{proof}
    All the precomputation has already been described and takes linear time. 
    
    Suppose we are given a query $(a, b, \textbf{f}, \textbf{l}, m, M)$ encoding a half-connected type-1 separation. Observe that if $\mathcal{C}$ is the set of all children $c$ of $b$ such that $m \leq c \leq M$ then $(a, b, \textbf{f}, \textbf{l}, m, M)$ encodes the half-connected type-1 separation $(A, B)$ with $B = \{a, b\} \cup \bigcup_{c \in \mathcal{C}} Desc(c)$ and $A = V(G) \setminus (B \setminus \{a, b\})$.

    Suppose first that $m = M =: d$. 
    Then, $lwpt(d) = a$ and $lwpt_2(d) = b$ since $(A, B)$ is type-1.
    If $d$ is not the left child of $b$, we answer that $(A, B)$ is totally-nested.
    If $d$ is the left child of $b$, then $b$ is not a leaf of $T$, and $b \neq r$ so $\beta(b) \neq 0$. We check whether $\alpha(b) > a$, every child $c$ of $\beta(b)$ satisfies either $lwpt(c) \geq \alpha(b)$ or $hgpt(c) \leq \alpha(b)$ and $T[\alpha(b), \beta(b)]$ is stable. This is done using the table $Z$ we already computed.
    If all these conditions hold, we answer that $(A, B)$ is crossed. Otherwise, we answer that $(A, B)$ is totally-nested.

    Suppose now that $m < M$. 
    Since $(A, B)$ is a half-connected type-1 separation then $\mathcal{C}$ is the set of children $d$ of $b$ such that $lwpt(d) = a, lwpt_2(d) = b$ and $M$ is the maximum such vertex, by \cref{prop:charac-t1-half-con}.
    Since $M$ is a child of $b$ then $b$ is not a leaf of $T$, and $b \neq r$ so $\beta(b) \neq 0$.
    We check if $M$ is the left child of $b$, $\alpha(b) > a$, every child $c$ of $\beta(b)$ satisfies either $lwpt(c) \geq \alpha(b)$ or $hgpt(c) \leq \alpha(b)$ and $T[\alpha(b), \beta(b)]$ is stable. This is done using the table $Z$ we already computed.
    If all these conditions hold, we answer that $(A, B)$ is crossed. Otherwise, we answer that $(A, B)$ is totally-nested.

    The correctness follows immediately from \cref{lem:charac-totally-nested-type-1}, and everything can be done in constant time.
\end{proof}

\begin{proposition} \label{prop:query-totally-nested-2}
    There is an algorithm which, given a 2-connected graph $G$ equipped a spanning tree $T$ and a numbering of $V(G)$ compatible with $T$, after an $O(n+m)$-time precomputation, can answer in constant time queries of the form: Given a half-connected type-2 separation, is it totally-nested?
\end{proposition}

\begin{proof}
    We create an array $Y$, such that $Y[v] = 1$ if $v$ has a left child $\ell$ and $lwpt_2(\ell) = v$, and $Y[v] = 0$ otherwise.
    Since we already computed all second lowpoints, $Y$ can be computed in time $O(n)$.
    We also create an array $Y'$ such that $Y'[v] = hgpt(\ell(v))$ if $v$ is not a leaf of $T$, where $\ell(v)$ is the left child of $v$, and $Y'[v] = 0$ otherwise.
    Since the highpoints are already computed, this takes time $O(n)$.
    We do the precomputation described in \cref{lem:query-leftmost-branch} to be able to find in constant time the minimum and maximum of $Y$ and of $Y'$ on every leftmost path.

    This finishes the description of the precomputation, which takes time $O(n+m)$ and we now describe how we process a query. 
    Suppose that we are given a half-connected type-2 separation encoded as $(a, b, \textbf{f}, \textbf{l}, m, M)$.
    Let $\mathcal{C}$ be the set of children $c$ of $b$ such that $c \leq M$ (note that $\mathcal{C}$ can be empty).
    Let $T_2$ be the connected component of $T - \{a, b\}$ that contains $V(T(a, b)) \neq \emptyset$.
    Recall that $(a, b, \textbf{f}, \textbf{l}, m, M)$ encodes the separation $(A, B)$ with $B = \{a, b\} \cup V(T_2) \cup \bigcup_{c \in \mathcal{C}} Desc(c)$ and $A = V(G) \setminus (B \setminus \{a, b\})$.
    Let $a'$ be the unique child of $a$ that is an ancestor of $b$ in $T$. Since $(A, B)$ is a type-2 separation then $a' \neq b$ by \cref{prop:charac-2-sep}. Observe that $a'$ is the smallest vertex of $B \setminus \{a, b\}$ in the numbering since it is an ancestor of every vertex in $B \setminus \{a, b\}$, hence $a' = \textbf{f}$.  

    We check whether any of the three following items hold. \begin{itemize}
        \item There exists $\gamma \in V(T(a, b))$ such that the left child $\ell'$ of $\gamma$ satisfies $lwpt_2(\ell') = \gamma$.
        \item There exists $\gamma \in V(T(a, b))$ such that the left child $\ell'$ of $\gamma$ satisfies $hgpt(\ell') \leq stab(p(a), a')$.
        \item $b$ is not a leaf of $T$. Let $\ell$ be the left child of $b$, $\alpha = \alpha(b)$ and $\beta = \beta(b)$. Then, the separation $(A, B)$ does not separate $\ell$ and $r$, $\alpha > a$, every child $c$ of $\beta$ satisfies either $lwpt(c) \geq \alpha$ or $hgpt(c) \leq \alpha$, and $T[\alpha, \beta]$ is stable. 
    \end{itemize}

    The first can be checked by querying the maximum of $Y$ on the leftmost path $T(a, b)$.
    The second one can be checked by querying the minimum of $Y'$ on the leftmost path $T(a, b)$ and comparing it to $stab(p(a), a')$, which we can compute in constant time using \cref{prop:queries-stability}.
    Checking that $b$ is not a leaf takes constant time. If so, then $\alpha(b) \neq 0$ and $\beta(b) \neq 0$.
    Observe that $\ell$ is on the side of $(A, B)$ that contains $r$ if and only if $M < \ell$, and checking that every child $c$ of $\beta$ satisfies either $lwpt(c) \geq \alpha$ or $hgpt(c) \leq \alpha$, and $T[\alpha, \beta]$ is stable can be done using $Z$.

    If one of these items holds, we return that $(A, B)$ is crossed. Otherwise, we return that $(A, B)$ is totally-nested.

    The correctness follows from \cref{lem:charac-tot-nested-t2}, and answering a query takes constant time.
\end{proof}

Finally, \cref{thm:query-tot-nested} follows immediately from \cref{prop:query-totally-nested-1,,prop:query-totally-nested-2}.

\subsection{Building the tree-decomposition from a set of nested separations}

Recall that \cref{lem:nested-iff-tree-decomp} states that, given a set $\mathcal{S}$ of nested separations of a graph $G$, there exists a unique (up to isomorphism) tree-decomposition of $G$ induced by $\mathcal{S}$.
The proof of \cref{lem:nested-iff-tree-decomp} can be turned into an efficient algorithm computing the corresponding tree-decomposition from a set of nested separations.
We believe that such a result already appears in the literature but we could not find it explicitly, so we prove it in \cref{sec:proof-build-tree} for completeness.

First, let us describe how we store separations of a graph.
Let $G$ be a graph on vertex set $[n]$. 
Any set $A \subseteq V(G)$ can be uniquely written as a disjoint union of maximal intervals $A=[a_1,a'_1] \cup [a_2,a'_2] \cup \ldots \cup [a_\ell,a'_\ell]$. The corresponding tuple ${\underline{a} = (a_1,a'_1,a_2,a'_2,\ldots,a_\ell,a'_\ell)}$ is the \emph{representation} of $A$, and has \emph{length} $\ell$. 
For a separation $(A,B)$, if $\underline a$ is the representation of $A$ and $\underline b$ is the representation of $B$ then $(\underline a, \underline b)$ is the \emph{representation} of $(A,B)$. If $\underline{a}$ has length $\ell_1$ and $\underline{b}$ has length $\ell_2$ then $(\underline a, \underline b)$ has length $\ell_1 + \ell_2$.

Note that, given the representation of $(A, B)$, say of length $\ell$, computing the representations of $A \cap B$, $A \cup B$, $A \setminus B$ and $B \setminus A$ can be done in time $O(\ell)$. This implies that these representations have length $O(\ell)$.
In particular, computing explicitly the list of all vertices in $A \cap B$ can be done in time $O(\ell + |A \cap B|)$.

Observe that if $(a, b, \textbf{f}, \textbf{l}, m, M)$ is the encoding of a 2-separation $(A, B)$ then we can compute the representation of $(A, B)$ in constant time.

\begin{restatable}{theorem}{buildtree} \label{thm:build-tree+torsos}
    Let $G$ be a connected graph on vertex set $[n]$ and $\mathcal{S} = \{(A_1, B_1), \ldots, (A_p, B_p)\}$ be a set of nested minimal separations of $G$, each of order at most $k$.
    Suppose that for every separation $(A_i, B_i)$, we are given the representation $(\underline{a_i}, \underline{b_i})$ of $(A_i, B_i)$, and that all representations have length at most $\ell$.
    Let $(T, \mathcal{V})$ be the tree-decomposition of $G$ induced by $\mathcal{S}$. 
    There is an $O(n + m + p\ell + pk^2)$-time algorithm which, given the graph $G$ and all representations $(\underline{a_i}, \underline{b_i})$, computes the tree-decomposition $(T, \mathcal{V})$ and all torsos.
\end{restatable}

The idea behind this algorithm is simple. Since $\mathcal{S}$ is a nested set of separations, there is a way to ``orient'' all the separations so that they form a laminar set family (see \cref{lem:sep-nested}). 
The Hasse diagram of this laminar set family is a tree, which is isomorphic to the tree $T$ (see \cref{prop:isomorphism-T}), and which can be built efficiently from the laminar set family (see \cref{prop:compute-tree}).
Once the tree is built, it suffices to add the vertices to the correct bags in the tree-decomposition (see \cref{lem:algo-add-not-x}), and to construct the torsos (see \cref{lem:algo-build-torsos}).

With \cref{thm:build-tree+torsos}, we can now easily derive our main result, which we restate for convenience.

\main*

\begin{proof}
    By \cref{thm:compute-pot-nested}, we can compute a set $\mathcal{S}$ of half-connected 2-separations of $G$, which contains all totally-nested separations, in time $O(n+m)$.
    Then, by \cref{thm:query-tot-nested}, after a precomputation done in time $O(n+m)$, we can decide in constant time whether a half-connected 2-separation of $G$ is totally-nested.
    By iterating over all $O(n+m)$ separations in $\mathcal{S}$, we can compute all totally-nested 2-separations of $G$ in time $O(n+m)$.
    Since $G$ is 2-connected, these separations are minimal.
    All the half-connected 2-separations of $G$ were computed with their encoding, from which we can get a representation of order $O(1)$ in constant time.
    Therefore, by \cref{thm:build-tree+torsos}, there is an algorithm which computes the tree-decomposition $(T, \mathcal{V})$ induced by these separations, and all torsos, in time $O(n+m)$.
    Then, the tree-decomposition induced by these separations is the Tutte-decomposition of $G$ by definition, and its torsos are the triconnected components of $G$.
\end{proof}

\section*{Acknowledgments}
We thank Jan Kurkofka and Johannes Carmesin for precious discussions, and Marthe Bonamy for her valuable comments.

\bibliographystyle{plain}
\bibliography{biblio.bib}

\appendix
\section{Proofs of \texorpdfstring{\cref{subsec:typo-2-sep}}{Section 3.1}} \label{sec:proof-charac-2-sec}

We first prove \cref{prop:charac-2-sep}, which we restate for convenience.

\charactwosep*

\begin{proof}
    We first show that if one of the cases holds then $(A, B)$ is a 2-separation. 
    First, suppose that there exists a nonempty set $\mathcal{C}$ of children of $b$ such that, up to exchanging $A$ and $B$, we have $B = \{a, b\} \cup \bigcup_{c \in \mathcal{C}} Desc(c)$ and ${A = V(G) \setminus (B \setminus \{a, b\})}$. Suppose also that for every $c \in \mathcal{C}$, we have $lwpt(c) = a$, $lwpt_2(c) = b$ ; and some vertex $v \notin \{a,b\}$ is not a descendant of any $c \in \mathcal{C}$. Observe that $B \setminus A$ is nonempty as $\mathcal{C}$ is nonempty, and $A \setminus B$ is nonempty since it contains $v$. 
    If $u \in B \setminus A$, there exists $c \in \mathcal{C}$ such that $u \in Desc(c)$. Since $lwpt(c) = a$ and $lwpt_2(c) = b$, and since $T$ is a normal spanning tree, then $N(u) \subseteq Desc(c) \cup \{a, b\} \subseteq B$ . Thus, $(A, B)$ is indeed a 2-separation of $G$ with separator $\{a < b\}$. 
    
    Second, suppose that there exists a child $a'$ of $a$ such that $b$ is a proper leftmost descendant of $a'$, and a set $\mathcal{C}$ of children of $b$ such that, if $T_2$ denotes the connected component of $T - \{a, b\}$ containing $a'$, up to exchanging $A$ and $B$, we have ${B = \{a, b\} \cup V(T_2) \cup \bigcup_{c \in \mathcal{C}} Desc(c)}$ and $A = V(G) \setminus (B \setminus \{a, b\})$. Suppose also that $a \neq r$ ; $T[a, b]$ is stable ; and for every $c \in \mathcal{C}$, we have $lwpt(c) \geq a$ and for every child $d$ of $b$ which is not in $\mathcal{C}$, we have $hgpt(d) \leq a$. 
    Since $a \neq r$, let $T_1$ be the component of $T-\{a,b\}$ containing $r$. Let $T_2$ be the component of $T-\{a,b\}$ containing $a'$ (hence all vertices of $T(a, b)$).
    Denote by $\overline{\mathcal{C}}$ the set of children of $b$ that are not in $\mathcal{C}$ and by $\mathcal{A}$ the set of siblings of $a'$.
    Since ${B = \{a, b\} \cup V(T_2) \cup \bigcup_{c \in \mathcal{C}} Desc(c)}$ then $A = \{a, b\} \cup V(T_1) \cup \bigcup_{a'' \in A} Desc(a'') \cup \bigcup_{d \in \overline{\mathcal{C}}} Desc(d)$. 
    Since $V(T_1), V(T_2) \neq \emptyset$ then $A \setminus B, B \setminus A \neq \emptyset$.
    We now prove that there is no edge between $A \setminus B$ and $B \setminus A$.
    First, every $c \in \mathcal{C}$ has $lwpt(c) \geq a$ and therefore since $T$ is a normal spanning tree we have $N(Desc(c)) \subseteq T[a, b] \cup Desc(c) \subseteq B$. 
    Moreover, every $d \in \overline{\mathcal{C}}$ has $hgpt(d) \leq a$ and therefore no vertex in $Desc(d)$ is adjacent to a vertex in $B \setminus A$. 
    Since $T$ is a normal spanning tree, every vertex in $Desc(a'')$ for $a'' \in \mathcal{A}$ is only adjacent to $V(T_1) \cup \{a\}$ and therefore to no vertex in $B \setminus A$. 
    Since $T[a, b]$ is stable then $b$ is a leftmost descendant of $a'$ and every back-edge $(x,y)$ with $a' \leq x < b$ satisfies $y \geq a$.
    Additionally, it follows from \cref{obs:numbering} and from the fact that $b$ is a leftmost descendant of $a'$ that $V(T_2) = [a', b)$. Therefore no vertex in $V(T_2)$ is adjacent to a vertex in $V(T_1)$. This finishes to prove that $(A, B)$ is a 2-separation of $G$, with separator $\{a < b\}$.

    It follows immediately from the definition of $A$ and $B$ in the two cases that $(A, B)$ is type-1 in the first case and type-2 in the second case. 

    Conversely, let $(A, B)$ be a 2-separation of $G$ with separator $\{a < b\}$.
    Suppose first that all vertices in $B \setminus A$ are descendants of $b$. Let $\mathcal{C}$ be the set of children of $b$ that are in $B \setminus A$.
    Since $(A, B)$ is a separation, there exists $u \in B \setminus A$. Then $u$ is a proper descendant of $b$. Let $c$ be the unique child of $b$ that is an ancestor of $u$. $u$ and $c$ are connected in $G - \{a, b\}$ so $c \in B \setminus A$, hence $c \in \mathcal{C}$ and $\mathcal{C} \neq \emptyset$. This also proves that $u \in \bigcup_{c \in \mathcal{C}} Desc(c)$, so $B \subseteq \{a, b\} \cup \bigcup_{c \in \mathcal{C}} Desc(c)$.
    If $c \in \mathcal{C}$ then $c \in B \setminus A$ so $Desc(c) \subseteq B \setminus A$ since $Desc(c)$ is connected in $G - \{a, b\}$. This finishes to prove that ${B = \{a, b\} \cup \bigcup_{c \in \mathcal{C}} Desc(c)}$.
    If $c \in \mathcal{C}$, since $G$ is 2-connected and $b \neq 1$ then $lwpt(c) < b$ by \cref{lem:charac-1-sep}. 
    Thus, some descendant of $c$ is adjacent to $lwpt(c) < b$. 
    Since $Desc(c) \subseteq B \setminus A$, this implies that $lwpt(c) \in B$ and therefore ${lwpt(c) = a}$.
    Since $(b, c) \in E(G)$ then $lwpt_2(c) \leq b$. The same reasoning as above implies that $lwpt_2(c) \in B$ and therefore $lwpt_2(c) = b$.
    Finally, let $v \in A \setminus B$. Then, $v \notin \{a, b\}$ is not a descendant of any $c \in \mathcal{C}$.
    This finishes to show that we are in the first case.
    Similarly, if all vertices in $A \setminus B$ are descendants of $b$ then we are also in the first case. 

    We now assume that not all vertices in $A \setminus B$ are descendants of $b$ and not all vertices in $B \setminus A$ are descendants of $b$.

    Let $a'$ be the unique child of $a$ such that $b$ is a descendant of $a'$ (which exists by \cref{lem:a-b-comp}).
    Then, $A \setminus B$ and $B \setminus A$ induce a partition of the connected components of $T - \{a, b\}$. Every connected component of $T - \{a, b\}$ is of one of the following forms: either it contains all vertices of $T$ that are not descendants of $a$, or it is $T(c)$ for some child $c$ of $b$, or it is $T(a'')$ for some child $a'' \neq a'$ of $a$, or it is the component that contains $T(a, b)$. 
    
    If $a = r$ there is no child $a'' \neq a'$ by \cref{lem:charac-1-sep}.
    Thus, there is at most one connected component of $T - \{a, b\}$ that doesn't contain only descendants of $b$ and thus one of $A \setminus B, B \setminus A$ must only contain descendants of $b$, which contradicts our assumption. 
    Therefore, $a \neq r$. 
    Let $T_1$ be the connected component of $T - \{a, b\}$ containing all vertices that are not descendants of $a$. Up to exchanging $A$ and $B$ we can assume that $V(T_1) \subseteq A$.
    Every child $a''$ of $a$ satisfies $lwpt(a'') < a$ by \cref{lem:charac-1-sep}.
    If $a''$ is a sibling of $a'$, there is an edge between a descendant of $a''$ and a proper ancestor of $a$, i.e. between a vertex in $Desc(a'')$ and a vertex of $T_1$. Thus, $Desc(a'') \subseteq A$.

    Since $B \setminus A$ contains a vertex which is not a descendant of $b$ then $T(a, b)$ is not empty and the connected component $T_2$ of $T - \{a, b\}$ that contains $T(a, b)$ is in $B$.
    Since $T(a, b) \neq \emptyset$ then $a' \neq b$ so $b$ is a proper descendant of $a'$. 

    Let $\mathcal{C}$ be the set of children $c$ of $b$ such that $Desc(c) \subseteq B \setminus A$ and $\overline{\mathcal{C}}$ be the set of all other children of $b$ (for which $Desc(c) \subseteq A \setminus B$).
    Then, $B = \{a, b\} \cup V(T_2) \cup \bigcup_{c \in \mathcal{C}}Desc(c)$ and $A = V(G) \setminus (B \setminus \{a, b\})$.
    
    Since $V(T_1) \subseteq A \setminus B$ and $V(T_2) \subseteq B \setminus A$, there is no child $c$ of $b$ such that $Desc(c)$ is connected to both $T_1$ and $T_2$ in $G-\{a,b\}$. 
    If $a' \leq x < b$ then $x$ is a descendant of $a'$ by \cref{obs:numbering} but $x < b$ so $b \notin T[a', x]$ so $x$ is in the component of $a'$ in $T - \{a, b\}$, i.e. $x \in V(T_2)$. Thus, every back-edge $(x,y)$ with $a' \leq x < b$ satisfies $a \leq y$, otherwise we would have an edge between $V(T_2)$ and $V(T_1)$.
    Furthermore, all $d \in \overline{\mathcal{C}}$ satisfy $hgpt(d) \leq a$ and all $c \in \mathcal{C}$ satisfy $lwpt(c) \geq a$, otherwise there would be an edge between $A \setminus B$ and $B \setminus A$. 

    Finally we show that $T[a,b]$ is a leftmost path. By contradiction, suppose not. Thus, $b$ is not a leftmost descendant of $a'$. By \cref{lem:follow-leftmost}, there is a leftmost descendant $x$ of $a'$ that is adjacent to $lwpt(a')$. Then, we have $x \in V(T_2)$ since $b$ is not a leftmost descendant of $a'$. Furthermore, $lwpt(a') < a$ by \cref{lem:charac-1-sep} so $x$ is adjacent to a vertex in $V(T_1)$. Now, $T_1$ and $T_2$ are connected in $G-\{a,b\}$, a contradiction. Hence, $T[a,b]$ is a leftmost path so $T[a, b]$ is stable.

    This finishes to prove that we are in the second case.
\end{proof}

We now prove \cref{lem:cc-type-2}, which we also restate for convenience.

\cctypetwo*

\begin{proof}
    Each connected component of $G - \{a, b\}$ is the union of connected components of $T - \{a, b\}$.
    Since $a \neq r$ and there exists a child of $a$ which is a proper ancestor of $b$ then the connected components of $T - \{a, b\}$ are exactly the following (and they are pairwise distinct). \begin{itemize}
        \item The component that contains $r$.
        \item The component that contains $a'$ (and thus all vertices of $T(a, b)$).
        \item $T[Desc(c)]$ for every child $c$ of $b$.
        \item $T[Desc(a'')]$ for every sibling $a''$ of $a'$.
    \end{itemize}

    Consider $d \in \mathcal{D}$. The connected component of $d$ in $T - \{a, b\}$ is $T[Desc(d)]$. By definition of $\mathcal{D}$ and since $T$ is a normal spanning tree, we have $N(Desc(c)) = \{a, b\}$, so $G[Desc(d)]$ is a connected component of $G - \{a, b\}$.
    
    Let $C$ be a connected component of $G - \{a, b\}$ and suppose that $C$ contains neither $r$ nor $a'$.
    Thus, either there exists a child $c$ of $b$ such that $Desc(c) \subseteq V(C)$, or there exists a sibling $a''$ of $a'$ such that $Desc(a'') \subseteq V(C)$.
    Since $C$ does not contain $a'$ then $C$ does not contain any vertex of $T(a, b)$.
    Suppose first that there exists a child $c$ of $b$ such that $Desc(c) \subseteq V(C)$. Since $G$ is 2-connected, it follows from \cref{lem:charac-1-sep} that $lwpt(c) < b$. 
    Thus, we have ${lwpt(c), lwpt_2(c) \in \{a, b\} \cup V(T[r, a)) \cup V(T(a, b))}$. 
    Since $r \notin V(C)$ and $C$ does not contain any vertex of $T(a, b)$, then we must have $lwpt(c) = a$ and $lwpt_2(c) = b$, so $c \in \mathcal{D}$. Since the connected component that contains $c$ is $G[Desc(c)]$ then $C = G[Desc(c)]$.
    Suppose now that there exists a sibling $a''$ of $a'$ such that $Desc(a'') \subseteq V(C)$. Since $G$ is 2-connected and $a \neq r$, by \cref{lem:charac-1-sep} we have $lwpt(a'') < a$ and thus $lwpt(a'') \in V(T[r, a))$. Hence, some vertex in $Desc(a'')$ is adjacent to some vertex in $T[r, a)$ and thus $C$ contains $r$, a contradiction. 
    This finishes the characterization of the connected components of $G - \{a, b\}$.

    We now prove that these connected components are pairwise distinct.
    It is clear that those of the form $G[Desc(d)]$ for $d \in \mathcal{D}$ are pairwise distinct, and distinct from those that contain $r$ and $a'$, so we only have to prove that $r$ and $a'$ are not in the same connected component in $G - \{a, b\}$.
    Let $T_{a'}$ be the connected component of $a'$ in $T - \{a, b\}$ and define $T_r$ analogously.
    Since $T$ is a normal spanning tree, there is no edge in $G$ between a vertex in $T_{a'}$ and a vertex in $Desc(a'')$ where $a''$ is a sibling of $a$.
    Similarly, there is no edge in $G$ between $Desc(c)$ and $Desc(c')$ if $c \neq c'$ are two children of $b$, no edge between $Desc(a_1)$ and $Desc(a_2)$ if $a_1 \neq a_2$ are two siblings of $a'$, and no edge between between $Desc(c)$ and $Desc(a'')$ if $c$ is a child of $b$ and $a''$ is a sibling of $a'$.
    Thus, if $a'$ and $r$ are in the same connected component in $G - \{a, b\}$, it is either because there exists a child $c$ of $b$ such that there is an edge in $G$ between $Desc(c)$ and $T_r$ and between $Desc(c)$ and $T_{a'}$, or because there is an edge in $G$ between $T_r$ and $T_{a'}$
    Since every child $c$ of $b$ satisfies either $lwpt(c) \geq a$ or $hgpt(c) \leq a$ then for every child $c$ of $b$, either there is no edge in $G$ between $Desc(c)$ and $T_{r}$ or there is no edge in $G$ between $Desc(c)$ and $T_{a'}$.
    Finally, since $b$ is a leftmost descendant of $a'$ and since the numbering of the vertices is consistent with $T$ then $V(T_{a'}) = [a', b-1]$.
    Since $T[a, b]$ is stable, every back-edge $(x, y)$ with $x \in [a', b-1]$ satisfies $y \geq a$.
    Thus, since $T$ is a normal spanning tree, there is no edge in $G$ between a vertex in $T_{a'}$ and a vertex in $T_r$.
    This proves that $a'$ and $r$ are not in the same connected component in $G - \{a, b\}$.
\end{proof}

\section{Proof of \texorpdfstring{\cref{thm:build-tree+torsos}}{Theorem 6.6}} \label{sec:proof-build-tree}


In this section, we show how to build this tree-decomposition efficiently if we are given a representation of small length of each separation in $\mathcal{S}$. We restate the result for convenience.

\buildtree*

We start by proving a lemma characterizing the vertices which only appear in a single bag of the tree-decomposition.

\begin{lemma} \label{lem:leaf-X-nonempty}
    Let $G$ be a graph and $\mathcal{S} = \{(A_1, B_1), \ldots, (A_p, B_p)\}$ be a set of nested separations of $G$.
    Let $(T, \mathcal{V})$ be the tree-decomposition of $G$ induced by $\mathcal{S}$. 
    Let $X \subseteq V(G)$ be the set of vertices which only appear in a single bag of $(T, \mathcal{V})$.
    For every $x \in V(G)$, we have $x \in X$ if and only if for every separation $(A, B) \in \mathcal{S}$ it holds that $x \notin A \cap B$.
    Furthermore, if ${t \in V(T)}$ is a leaf of $T$ then $V_t \cap X \neq \emptyset$.
    Therefore, if $(A, B) \in \mathcal{S}$ then $A \cap X \neq \emptyset$ and $B \cap X \neq \emptyset$.
\end{lemma}

\begin{proof}
    Let $x \in V(G)$ and suppose first that $x \in X$.
    Let $(A, B) \in \mathcal{S}$, and let $e$ be the edge of $T$ which induces the separation $(A, B)$. 
    Let $T_A, T_B$ be the two connected components of $T - e$. 
    Up to renaming $A$ and $B$, we have $A = \bigcup_{t \in V(T_A)} V_t$ and $B = \bigcup_{t \in V(T_B)} V_t$. 
    Thus, if $x \in A \cap B$, then there exist $t_A \in V(T_A)$ and $t_B \in V(T_B)$ such that $x \in V_{t_A} \cap V_{t_B}$. This would contradict that $x \in X$. Therefore, $x \notin A \cap B$.
    
    Conversely, suppose that $x \notin X$. Then, there are two distinct nodes $t, u \in V(T)$ such that $x \in V_t \cap V_{u}$. Let $e = (t, t')$ be the first edge of the path from $t$ to $u$ in $T$. 
    Since $(T, \mathcal{V})$ is a tree-decomposition then $x \in V_{t'}$.
    Let $(A, B) \in \mathcal{S}$ be the separation induced by the edge $e$. 
    Since $t$ and $t'$ belong to the two connected components of $T - e$ and since $x \in V_t \cap V_{t'}$ then $x \in A \cap B$.

    Consider now a leaf $t$ of $T$. If $t$ is the unique node of $T$ then $X = V(G)$ and thus $X \cap V_t = V(G) \neq \emptyset$.
    Otherwise, let $t'$ be the unique neighbour of $t$ in $T$, and let $(A, B) \in \mathcal{S}$ be the separation of $G$ induced by the edge $(t, t')$ of $T$. Up to renaming $A$ and $B$, we can assume that $A = V_t$.
    Consider any vertex $x \in A \setminus B$.
    Then, $x \notin A \cap B = V_t \cap V_{t'}$ so $x \notin V_{t'}$.
    Since $(T, \mathcal{V})$ is a tree-decomposition then $t$ is the only node of $T$ such that $x \in V_t$, and therefore $x \in V_t \cap X$.

    Finally, let $(A, B) \in \mathcal{S}$ and let $e$ be the edge of $T$ which induces this separation. Let $T_A$ be the component of $T - e$ such that $A = \bigcup_{t \in V(T_A)}V_t$ and let $t \in V(T_A)$ be a leaf of $T$. By the last paragraph, we have $V_t \cap X \neq \emptyset$ so $A \cap X \neq \emptyset$. A similar argument shows that $B \cap X \neq \emptyset$.
\end{proof}

We now show that we can use the vertices which belong to a single bag of the tree-decomposition to orient the separations of $\mathcal{S}$ consistently, so that they form a laminar set family.

\begin{lemma} \label{lem:sep-nested}
    Let $G$ be a graph and $\mathcal{S} = \{(A_1, B_1), \ldots, (A_p, B_p)\}$ be a set of nested minimal separations of~$G$.
    Let ${x \in V(G)}$ be a vertex such that $x \in A \setminus B$ for every separation $(A, B) \in \mathcal{S}$.
    Then, for every $(A_i, B_i) \neq (A_j, B_j) \in \mathcal{S}$, the sets $B_i \setminus A_i$ and $B_j \setminus A_j$ are either disjoint, or one of them is a proper subset of the other.
\end{lemma}

\begin{proof}
    Since $\mathcal{S}$ is a set of nested separations, one of the following holds: \begin{itemize}
        \item $A_i \subseteq A_j$ and $B_j \subseteq B_i$,
        \item $A_i \subseteq B_j$ and $A_j \subseteq B_i$,
        \item $B_i \subseteq A_j$ and $B_j \subseteq A_i$,
        \item $B_i \subseteq B_j$ and $A_j \subseteq A_i$.
    \end{itemize}
    In the first case, we get $B_j \setminus A_j \subseteq B_i \setminus A_i$.
    However, for $k \in \{i, j\}$, we have $A_k = V(G) \setminus (B_k \setminus A_k)$, and since the separation $(A_k, B_k)$ is minimal then $A_k \cap B_k = N(B_k \setminus A_k)$.
    Therefore, $B_j \setminus A_j = B_i \setminus A_i$ would imply $(A_i, B_i) = (A_j, B_j)$, which does not hold, so $B_j \setminus A_j \subsetneq B_i \setminus A_i$.
    The second case is impossible since $x \in A_i \setminus B_j$.
    In the third case, we get that $B_i \setminus A_i$ and $B_j \setminus A_j$ are disjoint.
    In the last case, we get $B_i \setminus A_i \subseteq B_j \setminus A_j$.
    We prove that the inclusion is strict exactly like in the first case.
\end{proof}

\begin{remark}\label{rmk:nb-nested-seps}
    Let $G$ be a graph and $\mathcal{S} = \{(A_1, B_1), \ldots, (A_p, B_p)\}$ be a set of nested minimal separations of~$G$.
    By \cref{lem:leaf-X-nonempty}, up to renaming $A_i$ and $B_i$ for every $i \in [p]$, there exists a vertex $x \in V(G)$ such that $x \in A_i \setminus B_i$ for every $i \in [p]$.
    Then, by \cref{lem:sep-nested}, the set family $\mathcal{F} = \{B_i \setminus A_i : i \in [p]\}$ is a laminar set family with $p$ sets.
    Furthermore, we can add the sets $V(G)$ and $\emptyset$ to this laminar set family and obtain a set family $\mathcal{F}'$ which is still a laminar set family, with $p+2$ sets.
    It is a standard result that every laminar set family on a ground set of size $n$ contains at most $2n-1$ sets.
    Thus, $p+2 \leq 2n-1$, so $p \leq 2n-3$.

    If $(T, \mathcal{V})$ is the Tutte-decomposition of $G$, the number of bags which contain a vertex is $1$ plus the number of totally-nested separations for which this vertex is in the separator.
    Therefore, $\sum_{t \in V(T)}|V_t| \leq n + 2|E(T)| \leq 5n-6$.
\end{remark}

Consider the tree $T'$ obtained by adding a root to the Hasse diagram of this laminar set family, and labeling the edge from $B_i \setminus A_i$ to its ancestor with label $i$.
The next result states that $T'$ is isomorphic to the tree $T$ indexing the tree-decomposition induced by $\mathcal{S}$, and that the edge indexed with label $i$ induces the separation $(A_i, B_i)$.

\begin{proposition} \label{prop:isomorphism-T}
    Let $G$ be a graph and $\mathcal{S} = \{(A_1, B_1), \ldots, (A_p, B_p)\}$ be a set of nested minimal separations of $G$.
    Let $(T, \mathcal{V})$ be the tree-decomposition of $G$ induced by $\mathcal{S}$. 
    Let $X \subseteq V(G)$ be the set of vertices which only appear in a single bag of $(T, \mathcal{V})$.
    Let $x \in X$ and suppose that $x \in A \setminus B$ for every separation $(A, B) \in \mathcal{S}$.
    Root $T$ at $t_x$, the only node of $T$ such that $x \in V_{t_x}$.
    Let $T'$ be the rooted tree whose edges are labelled from 1 to $p$ such that the edge labelled $i$ is a proper ancestor of the edge labelled $j$ if and only if $B_j \setminus A_j \subsetneq B_i \setminus A_i$.
    Then, there is an isomorphism of rooted trees between $T'$ and $T$ that maps the edge of $T$ which induces the separation $(A_i, B_i)$ to the edge labelled $i$ in $T'$.
\end{proposition}

We will use the following technical result in the course of the proof of \cref{prop:isomorphism-T}.

\begin{lemma} \label{lem:desc-compr-prec}
    Let $G$ be a graph and $\mathcal{S} = \{(A_1, B_1), \ldots, (A_p, B_p)\}$ be a set of nested minimal separations of $G$.
    Let $(T, \mathcal{V})$ be the tree-decomposition of $G$ induced by $\mathcal{S}$. 
    Let $X \subseteq V(G)$ be the set of vertices which only appear in a single bag of $(T, \mathcal{V})$.
    Let $x \in X$ and suppose that $x \in A \setminus B$ for every separation $(A, B) \in \mathcal{S}$.
    Root $T$ at $t_x$, the only node of $T$ such that $x \in V_{t_x}$.
    Consider two distinct separations $(A_i, B_i), (A_j, B_j) \in \mathcal{S}$.
    Let $e_i$ be the edge of $T$ that induces the separation $(A_i, B_i)$ and define $e_j$ analogously.
    Then, $e_i$ is a proper ancestor of $e_j$ in $T$ if and only if $B_j \setminus A_j \subsetneq B_i \setminus A_i$.
\end{lemma}

\begin{proof}
    Suppose that $e_i$ is a proper ancestor of $e_j$ in $T$. Let $t_i$ be the top-endpoint of $e_i$ and $t_j$ the top-endpoint of $e_j$. 
    Then, $t_j$ is a descendant of $t_i$ so $Desc(t_j) \subseteq Desc(t_i)$. 
    However, since $t_x$ is the only node of $T$ such that ${x \in V_{t_x}}$ and since $x \in A_i \setminus B_i$ then $B_i = \bigcup_{t \in Desc(t_i)}V_t$, and $A_i = \bigcup_{t \notin Desc(t_i)}V_t$.
    Similarly, $B_j = \bigcup_{t \in Desc(t_j)}V_t$, and ${A_j = \bigcup_{t \notin Desc(t_j)}V_t}$.
    Therefore, $B_j \subseteq B_i$ and $A_i \subseteq A_j$, so $B_j \setminus A_j \subseteq B_i \setminus A_i$.
    However, for $k \in \{i, j\}$, we have $A_k = V(G) \setminus (B_k \setminus A_k)$, and since the separation $(A_k, B_k)$ is minimal then $A_k \cap B_k = N(B_k \setminus A_k)$.
    Therefore, since $(A_i, B_i) \neq (A_j, B_j)$ then $B_j \setminus A_j \subsetneq B_i \setminus A_i$.

    Conversely, suppose that $B_j \setminus A_j \subsetneq B_i \setminus A_i$.
    Let $t_i$ be the top-endpoint of $e_i$ and $t_j$ the top-endpoint of $e_j$.
    Let $t$ be a descendant of $t_j$ which is a leaf of $T$. By \cref{lem:leaf-X-nonempty}, there exists $y \in V_{t} \cap X$. 
    Then, $y \in B_j \setminus A_j$ by \cref{lem:leaf-X-nonempty} so $y \in B_i \setminus A_i$. Since $V_t$ is the only bag that contains $y$ then $t$ is a descendant of $t_i$.
    Since $t_i$ and $t_j$ have a common descendant in $T$ then $t_i$ and $t_j$ are comparable in $T$, so $e_i$ and $e_j$ are as well.
    Since $(A_i, B_i) \neq (A_j, B_j)$ then $e_i \neq e_j$.
    If $e_j$ were a proper ancestor of $e_i$ then by the first implication we would have $B_i \setminus A_i \subsetneq B_j \setminus A_j$, which would be a contradiction.
    Therefore, $e_i$ is a proper ancestor of $e_j$ in $T$.
\end{proof}

We are now ready to prove \cref{prop:isomorphism-T}.

\begin{proof}[\textit{Proof of \cref{prop:isomorphism-T}}.]
    Denote by $t_0$ the root of $T$ and by $t'_0$ the root of $T'$.
    For $i \in [p]$, let $e_i$ be the edge of $T$ that induces the separation $(A_i, B_i)$, and $e'_i$ the edge of $T'$ labelled $i$.
    For $i \in [p]$, let $t_i$ be the top endpoint of $e_i$ and $t'_i$ be the top endpoint of $e'_i$.
    We claim that $\phi: t_i \mapsto t'_i$ is an isomorphism of rooted trees between $T$ and $T'$.

    Clearly, $\phi$ is a bijection from $V(T)$ to $V(T')$ which maps the root of $T$ to the root of $T'$.
    By contradiction, suppose that there is an edge $e_i = (s_i, t_i)$ of $T$ such that $(\phi(s_i), \phi(t_i))$ is not an edge of $T'$.
    We first argue that $s_i \neq t_0$. Indeed, suppose $s_i = t_0$. 
    Then, $e_i$ has no proper ancestor in $T$. 
    By \cref{lem:desc-compr-prec}, there is no $j \in [p]$ such that $B_i \setminus A_i \subsetneq B_j \setminus A_j$.
    Therefore, the edge $e'_i$ has no proper ancestor in $T'$. 
    Thus, the bottom endpoint of $e'_i$ is $t'_0 = \phi(t_0) = \phi(s_i)$. 
    Furthermore, $\phi(t_i) = t'_i$ is the top endpoint of $e'_i$ thus $e'_i = (t'_0, t'_i) = (\phi(s_i), \phi(t_i))$ is an edge of $T'$.
    This proves that $s_i \neq t_0$.
    
    Next, let $e_j$ be the unique edge of $T$ whose top endpoint is $s_i$. Then, $s_i = t_j$. 
    Since $e_j$ is a proper ancestor of $e_i$ in $T$ then $B_i \setminus A_i \subsetneq B_j \setminus A_j$ by \cref{lem:desc-compr-prec}, and thus $e'_i$ is a proper descendant of $e'_j$ in $T'$.
    Thus, $t'_i$ is a proper descendant of $t'_j$. 
    By contradiction, assume that $t'_i$ is not a child of $t'_j$ and let $t'_k$ be the unique child of $t'_j$ that is an ancestor of $t'_i$.
    Then, the edge $(t'_j, t'_k)$ is the edge $e'_k$ hence is labelled $k$. 
    Then, $B_i \setminus A_i \subsetneq B_k \setminus A_k \subsetneq B_j \setminus A_j$. Therefore, by \cref{lem:desc-compr-prec}, $e_j$ is a proper ancestor of $e_k$, which itself is a proper ancestor of $e_i$.
    Then, $t_i$ is a proper descendant of $t_k$, which is itself a proper descendant of $t_j = s_i$. This contradicts the existence of the edge $(s_i, t_i)$ in $T$.
    Thus, $t'_i$ is a child of $t'_j = \phi(t_j) = \phi(s_i)$ and thus there is an edge $(\phi(s_i), \phi(t_i))$ in $T'$, a contradiction.

    Since $\phi$ is injective, no two edges of $T$ are mapped to the same edge of $T'$ by $\phi$. Therefore, at least $p$ edges of $T'$ are the image by $\phi$ of an edge of $T$. However, $T'$ has exactly $p$ edges so every edge of $T'$ is the image by $\phi$ of an edge of $T$.
    This concludes the proof that $\phi$ is an isomorphism of rooted trees.

    For the last part of the statement, observe that $t_i$ is mapped to $t'_i$, so the top endpoint of the edge $e_i$ is mapped to the top endpoint of the edge $e'_i$. Therefore $e_i$, the edge of $T$ which induces the separation $(A_i, B_i)$, is mapped to $e'_i$, the edge labelled $i$ in $T'$.
\end{proof}

The next result will help us find, for every vertex $y \in V(G)$ which belongs to a single bag of the tree-decomposition, the corresponding node in the tree $T'$ obtained from Hasse diagram. 

\begin{lemma} \label{lem:which-node-x}
    Let $G$ be a graph and $\mathcal{S} = \{(A_1, B_1), \ldots, (A_p, B_p)\}$ be a set of nested minimal separations of $G$.
    Let $(T, \mathcal{V})$ be the tree-decomposition of $G$ induced by $\mathcal{S}$. 
    Let $X \subseteq V(G)$ be the set of vertices which only appear in a single bag of $(T, \mathcal{V})$.
    Let $x \in X$ and root $T$ at $t_x$, the only node of $T$ such that $x \in V_{t_x}$.
    Suppose that $x \in A \setminus B$ for every separation $(A,B) \in \mathcal{S}$.
    For every $y \in X$, the set $\{B_i \setminus A_i : y \in B_i\}$ is totally ordered by inclusion.
    Furthermore, $y \notin V_{t_x}$ if and only if there exists a separation $(A_i, B_i) \in \mathcal{S}$ such that $y \in B_i$.
    If $y \notin V_{t_x}$, let $(A_i, B_i) \in \mathcal{S}$ be the separation with $y \in B_i$ such that $B_i \setminus A_i$ is minimal under inclusion.
    Let $e_i$ be the edge of $T$ that induces the separation $(A_i, B_i)$ and let $t_i$ be the top endpoint of the edge $e_i$.
    Then, $y \in V_{t_i}$.
\end{lemma}

\begin{proof}
    Let $y \in X$ and let $t_y$ be the unique node of $T$ such that $y \in V_{t_y}$.
    Let $(A_i, B_i) \neq (A_j, B_j) \in \mathcal{S}$ be such that $y \in B_i \cap B_j$.
    Let $e_i$ (resp. $e_j$) be the edge of $T$ that induces the separation $(A_i, B_i)$ (resp. $(A_j, B_j)$) and $t_i$ (resp. $t_j$) be the top endpoint of $e_i$ (resp. $e_j$).
    Since we have that $y \in B_i$, $x \in A_i \setminus B_i$ and $t_x$ is the root of $T$ then $t_y$ is a descendant of $t_i$. Similarly, $t_y$ is a descendant of $t_j$.
    Therefore, $t_i$ and $t_j$ are comparable in $T$ so $e_i$ and $e_j$ too.
    Thus, by \cref{lem:desc-compr-prec}, one of $B_i \setminus A_i$ and $B_j \setminus A_j$ is a proper subset of the other.

    If $y \notin V_{t_x}$, let $e_y$ be the edge of $T$ whose top endpoint is $t_y$ and let $(A_i, B_i) \in \mathcal{S}$ be the separation induced by $e_y$. Since $T$ is rooted at $t_x$ and $x \in A_i \setminus B_i$ then $y \in B_i$.
    Conversely, if there exists a separation $(A_i, B_i) \in \mathcal{S}$ such that $y \in B_i$, let $e_i$ be the edge of $T$ which induces the separation $(A_i, B_i)$ and let $t_i \neq t_x$ be the top endpoint of $e_i$. Then, $t_y$ is a descendant of $t_i$ so $t_y \neq t_x$ and thus $y \notin V_{t_x}$.
    
    Suppose that $y \notin V_{t_x}$. 
    Then, $\{B_i \setminus A_i : y \in B_i\}$ is non-empty and totally ordered by inclusion. 
    Let $(A_i, B_i)$ be the separation with $y \in B_i$ such that $B_i \setminus A_i$ is minimal. Let $e_i$ be the edge of $T$ which induces the separation $(A_i, B_i)$ and let $t_i$ be the top endpoint of $e_i$.
    Let $t_y$ be the unique node of $t$ such that $y \in V_{t_y}$.
    Since $y \in B_i$ then $t_y$ is a descendant of $t_i$.
    By contradiction, suppose that $t_y$ is a proper descendant of $t_i$ and let $e_j$ be the first edge on the path from $t_i$ to $t_y$. 
    Let $t_j$ be the top endpoint of $e_j$ and $(A_j, B_j)$ be the separation induced by $e_j$. 
    Since $t_y$ is a descendant of $t_j$ then $y \in B_j$.
    Since $e_j$ is a proper descendant of $e_i$ then $B_j \setminus A_j \subsetneq B_i \setminus A_i$ by \cref{lem:desc-compr-prec}, which contradicts the minimality of $B_i \setminus A_i$.
    Thus, $t_y = t_i$, so $y \in V_{t_i}$.
\end{proof}

Combining all previous results, we can prove the first step towards \cref{thm:build-tree+torsos}, namely that we can efficiently build the tree $T$ underlying the tree-decomposition induced by $\mathcal{S}$, label its edges with the corresponding separations, and add the vertices which belong to a single bag to the correct bag.

The idea behind this algorithm is the following: we first orient all separations $(A_i, B_i) \in \mathcal{S}$ so that the sets $B_i \setminus A_i$ form a laminar set family. Then, we iterate through all vertices of $G$ to construct the Hasse diagram of this laminar set family, from which we can build $T$ by \cref{prop:isomorphism-T}, and while doing so we add the vertices of $G$ which belong to a single bag of the tree-decomposition to their corresponding bag, using \cref{lem:which-node-x}.

\begin{proposition} \label{prop:compute-tree}
    Let $G$ be a graph on vertex set $[n]$ and $\mathcal{S} = \{(A_1, B_1), \ldots, (A_p, B_p)\}$ be a set of nested minimal separations of $G$ of order at most $k$.
    Suppose that for every $(A_i, B_i) \in \mathcal{S}$ we are given a representation of $(A_i, B_i)$ of length at most $\ell$.
    Let $(T, \mathcal{V})$ be the tree-decomposition of $G$ induced by $\mathcal{S}$. 
    Let $X \subseteq V(G)$ be the set of vertices which only appear in a single bag of $(T, \mathcal{V})$.
    There is an $O(n + p(\ell + k))$-time algorithm which, given the graph $G$ and all representations $(\underline{a_i}, \underline{b_i})$, constructs the tree $T$ and labels its edges from $1$ to $p$ so that the edge labelled $i$ induces the separation $(A_i, B_i)$ in $(T, \mathcal{V})$.
    Furthermore, the algorithm constructs bags $(W_t)_{t \in V(T)}$ such that for every $t \in V(T)$, we have $W_t = V_t \cap X$.
\end{proposition}

\begin{proof}
    In time $O(p(\ell + k))$, we compute a representation of all sets $A_i \setminus B_i$ and $B_i \setminus A_i$, and all sets $A_i \cap B_i$ explicitly.
    In time $O(p\ell)$, we also compute $|B_i \setminus A_i|$ for every $i \in [p]$.
    We initialize an array $InX$ of size $n$, initially filled with 1's.
    We iterate over all sets $A_i \cap B_i$ and set $InX[v] = 0$ for every $v \in A_i \cap B_i$.
    Once this is done, we have $v \in X$ if and only if $InX[v] = 1$ by \cref{lem:leaf-X-nonempty}.
    For every $v \in V(G)$, we store the next and the previous element of $X$.
    This takes time $O(n + p(\ell + k))$ overall.

    Let $x$ be the smallest element of $X$.
    For every $i \in [p]$, we check in time $O(\ell)$ whether $x \in A_i \setminus B_i$ or $x \in B_i \setminus A_i$. 
    We rename $A_i$ and $B_i$ so that $x \in A_i \setminus B_i$ for every $i \in [p]$.
    For every $i \in [p]$, we compute the representation of $B_i \cap X$ relative to $X$ (meaning the representation of $B_i \cap X$ as a disjoint union of maximal intervals of $X$) in time $O(\ell)$ using that we stored for every vertex of $G$ the next and previous element of $X$.
    Note that this new representation has length at most $\ell$. 

    For every $i \in [p]$ and every interval $[a, b]$ of the representation of $B_i \cap X$, we create a vector $w \coloneqq (a, -b, -|B_i \setminus A_i|, i)$.
    Observe that every coordinate of such a vector is between $-n-p$ and $n+p$, and there are $O(p\ell)$ such vectors. Thus, by \cref{cor:sort-lex}, we can sort the list of all such vectors lexicographically in increasing order in time $O(n + p\ell)$.

    We denote the resulting sorted list by $w_1 \leq_{lex} \ldots \leq_{lex} w_q$, and we store it using a stack $L$ with $w_1$ on top and $w_q$ at the bottom.
    Note that given a vector $w_j$, we can find the index of the separation $(A_i, B_i)$ which gave rise to the vector $w_j$ using its last coordinate. We say that $(A_i, B_i)$ \emph{generates} $w_j$. Every vector $w = (\alpha, -\beta, -|B_i \setminus A_i|, i)$ comes from an interval of the representation of $B_i \cap X$, and this interval is $[\alpha, \beta]$. We will often assimilate the vector $w$ and its corresponding interval.

    The remainder of the algorithm is described in \cref{alg:compute-forest}. 
    The intuition is the following: we iterate over all vertices ${y \in X}$ while maintaining a stack $S$ of vectors $w$, whose elements correspond exactly to the intervals $[a, b]$ of the representations of $B_i \cap X$ such that $y \in [a, b]$. 
    At any point of the algorithm, we ensure that if $w$ is below $w'$ in the stack and if $w$ is generated by $(A, B) \in \mathcal{S}$ and $w'$ by $(A', B') \in \mathcal{S}$ then $B' \setminus A' \subsetneq B \setminus A$.
    Said differently, the vectors $w$ in the stack $S$ are ordered according to the inclusion of the corresponding sets $B_i \setminus A_i$, with the vector $w$ on top corresponding to the set $B_i \setminus A_i$ furthest from the root in the Hasse diagram. 
    While doing so, we also construct a rooted forest $F$ on vertex set $[0, k]$, with its edges labelled by $[k]$. Initially, $F$ is edgeless. We will gradually add edges to $F$ and define a parent function on the nodes, so that $F$ eventually becomes a tree rooted at node 0.
    To do so, whenever we add to $S$ the vector $w$ corresponding to the first interval of the representation of $B_i \cap X$, if $w'$ is the vector just below $w$ and $w'$ is generated by a separation $(A_j, B_j) \in \mathcal{S}$, we add an edge $ij$ to $F$, label it with $i$ and set $j$ to be the parent of $i$.
    It will then follow from \cref{prop:isomorphism-T} that at the end of the algorithm, we have $F = T$.
    Furthermore, if $w$ is the vector on top of $S$ while considering $y \in X$ and if $w$ is generated by $(A_i, B_i)$, we add $y$ to the bag $W_i$. The last property of the statement will then follow immediately from \cref{lem:which-node-x}.

    \begin{algorithm}
    \caption{Algorithm for \cref{prop:compute-tree}}
    \label{alg:compute-forest}
    \begin{algorithmic}[1]
    \Require Stack $L$ containing $w_1, \ldots, w_q$ from top to bottom.

    \State $S \gets$ empty stack
    \State $F \gets$ edgeless forest on vertex set $[0, k]$
    \State $W_i \gets \emptyset$ for every $i \in [0, k]$
    \For{$y \in X$} \label{line:enter-for-X}
        \If{$L \neq \emptyset$} \label{line:enter-loop}
            \State $w \coloneqq (a, -b, -|B_i \setminus A_i|, i) \gets L.top()$ \label{line:def-w-1}
        \Else
            \State $a \gets \infty$
        \EndIf
        \While{$a = y$}  \label{line:begin-while}
            \If{$a$ is the first element of $B_i \cap X$} \label{line:test-first}
                \If{$S \neq \emptyset$}
                    \State $(\cdot, \cdot, \cdot, j) \gets S.top()$
                    \State Add an edge $ij$ in $F$, with label $i$ and set $parent(i) = j$ \label{line:add-edge}
                \Else
                    \State Add an edge $i0$ in $F$, with label $i$ and set $parent(i) = 0$
                \EndIf
            \EndIf
            \State $S.push(w)$
            \State $L.pop()$
            \If{$L \neq \emptyset$} \label{line:16}
                \State $w := (a, -b, -|B_i \setminus A_i|, i) \gets L.top()$ \label{line:def-w-2}
            \Else \label{line:18}
                \State $a \gets \infty$ \label{line:19}
            \EndIf
        \EndWhile
        \\ \label{line:end-while-insert-stack}
        \If{$S \neq \emptyset$} \label{line:enter-first-if}
            \State $(\cdot, \cdot, \cdot, i) \gets S.top()$
            \State Add $y$ to $W_{i}$
        \Else
            \State Add $y$ to $W_0$
        \EndIf
        \\
        \If{$S \neq \emptyset$} \label{line:enter-if}
            \State $w \coloneqq (a, -b, -|B_i \setminus A_i|, i) \gets S.top()$ \label{line:def-l-1}
            \While{$b = y$} \label{line:while-remove}
                \State $S.pop()$
                \If{$S \neq \emptyset$} \label{line:test-removal}
                    \State $w \coloneqq (a, -b, -|B_i \setminus A_i|, i) \gets S.top()$ \label{line:def-l-2}
                \Else
                    \State $b \gets \infty$
                \EndIf
            \EndWhile
        \EndIf
        
    \EndFor

    \State Return $F, (W_i)_{i \in [0, k]}$

    \end{algorithmic}
\end{algorithm}

    \begin{claim}
        The following invariants hold throughout the algorithm.
        \begin{itemize}
            \item At any point of time, if $w$ is below $w'$ in $S$, and if $w$ is generated by $(A_i, B_i)$ and $w'$ by $(A_j, B_j)$, then $B_j \setminus A_j \subsetneq B_i \setminus A_i$.
            \item Suppose that $w$ is generated by $(A_i, B_i) \in \mathcal{S}$. When we add $w$ to $S$, there is a vector $w'$ in $S$ generated by $(A_j, B_j)$ if and only if $B_i \setminus A_i \subsetneq B_j \setminus A_j$, and if so there is exactly one such vector. 
            Furthermore, if $w$ is the vector corresponding to the first interval of $B_i \cap X$ then we add an edge to $F$ labelled $i$, which defines the parent for the vertex $i$. 
            This edge creates a path in $F$ from the node $i$ to the node 0, and $l \in [k]$ is the label of an edge of this path if and only if either $l=i$ or $B_i \setminus A_i \subsetneq B_l \setminus A_l$.
            \item When we consider $y \in X$, at \cref{line:enter-loop} the vectors in $S$ correspond exactly to the intervals $[a, b]$ of the $B_i \cap X$ such that $y \in (a, b]$, at \cref{line:end-while-insert-stack} the elements of $S$ correspond to the intervals $[a, b]$ of the $B_i \cap X$ such that $y \in [a, b]$, and when we finish examining $y$ the elements of $S$ correspond exactly to the intervals $[a, b]$ of the $B_i \cap X$ such that $y \in [a, b)$.
        \end{itemize}
    \end{claim}
    
    \begin{subproof}
    We prove the three invariants simultaneously. 
    Initially $S$ is empty so the first invariant holds.
    The first $y \in X$ which we consider is $x$, the smallest element of $X$, and we have that $x \in A_i \setminus B_i$ for every separation $(A_i, B_i) \in \mathcal{S}$.
    Thus, at \cref{line:enter-loop}, $S$ is empty and therefore the elements of $S$ are exactly the intervals $[a, b]$ of the $B_i \cap X$ such that $y \in (a, b]$.
    Initially, the $w$ we consider is $w_1$. 
    Write $w_1 = (a, -b, -|B_i \setminus A_i|, i)$. Then, $a \in B_i \cap X$ so $a \neq x$. Thus, we do not enter the While loop at \cref{line:begin-while} and directly go to \cref{line:end-while-insert-stack}.
    Thus, at \cref{line:end-while-insert-stack} the elements of $S$ correspond to the intervals $[a, b]$ of the $B_i \cap X$ such that $x \in [a, b]$.
    Since $S = \emptyset$, we do not enter the Ifs at \cref{line:enter-first-if,line:enter-if} so when we finish examining $x$, the elements of $S$ correspond exactly to the intervals $[a, b]$ of the $B_i \cap X$ such that $x \in [a, b)$.
    Note that we didn't modify $S$ while considering $x$ so the first two invariants are also maintained.

    Suppose now that we enter the For loop at \cref{line:enter-for-X} considering some $y \in X \setminus \{x\}$. Let $y^-$ be the previous $y \in X$ which we considered, and assume that all three invariants held for all previous $y' \in X$.
    By induction hypothesis, at \cref{line:enter-loop}, the vectors in $S$ correspond exactly to the intervals $[a, b]$ of the $B_i \cap X$ such that $y^- \in [a, b)$. 
    Since $y$ is the successor of $y^-$ in $X$, these intervals are exactly the intervals $[a, b]$ of the $B_i \cap X$ such that $y \in (a, b]$. This proves the first part of the third invariant.

    Suppose that we enter the While loop at \cref{line:begin-while}. 
    Then, it cannot be that $a = \infty$ so $L$ is non-empty. Denote by ${w = (a, -b, - |B_i \setminus A_i|, i)}$ the element on top of $L$. 
    Observe that because of \cref{line:16,line:def-w-2,line:18,line:19}, this is true regardless of whether or not it is the first time we enter the While loop when considering $y$. 
    We add $w$ on top of $S$ and remove it from $L$. 
    Observe that since we enter the While loop, we have $a = y$, and thus $y$ is the first vertex of the interval corresponding to $w$.

    We now prove that the first invariant is maintained. Suppose that $w \coloneqq (\alpha, -\beta, -|B_i \setminus A_i|, i)$ is generated by $(A_i, B_i) \in \mathcal{S}$ and is added to $S$ while considering $y$. 
    We now show that if $w'$ is below $w$ in $S$ and is generated by $(A_j, B_j) \in \mathcal{S}$ then $B_i \setminus A_i \subsetneq A_j \setminus B_j$.
    If $S$ was empty then this is trivially true. 
    Otherwise, let $w'$ be the element that was on top of $S$ before adding $w$, say that $w'$ is generated by $(A_j, B_j) \in \mathcal{S}$ and write $w' = (\alpha', -\beta', -|B_j \setminus A_j|, j)$. 
    Since the invariant held before adding $w$ to $S$, it suffices to show that $B_i \setminus A_i \subsetneq B_j \setminus A_j$. 
    If $w'$ was added to $S$ before considering $y$ then $y$ is in the interval corresponding to $w'$ by the first part of the third invariant so $y \in B_j \cap X$. 
    Similarly, if $w'$ was added to $S$ while considering $y$ then $y \in B_j \cap X$ (recall that we only add a vector to $S$ when considering $y$ if $y$ is the first element of its corresponding interval). 
    In both cases, $y \in B_j \cap X$. 
    Since $w$ was added to $S$ while considering $y$ then $y \in B_i \cap X$. 
    Thus, $B_i \cap X$ and $B_j \cap X$ are not disjoint. However, by \cref{lem:leaf-X-nonempty}, we have $B_i \cap X \subseteq B_i \setminus A_i$ and $B_j \cap X \subseteq B_j \setminus A_j$ so $B_i \setminus A_i$ and $B_j \setminus A_j$ are not disjoint. 
    Note that since the intervals $[\alpha', \beta']$ and $[\alpha, \beta]$ overlap (they both contain $y$) then $i \neq j$ so $(A_i, B_i) \neq (A_j, B_j)$.
    Therefore, one of $B_i \setminus A_i$ and $B_j \setminus A_j$ is a proper subset of the other by \cref{lem:sep-nested}. 
    Since $w'$ was added to $S$ before $w$ then $w'$ was above $w$ in $L$ (since we only ever consider the element on top of $L$). 
    Thus, $(\alpha', -\beta', - |B_j \setminus A_j|, j) \leq_{lex} (\alpha, -\beta, -|B_i \setminus A_i|, i)$.
    
    If $\alpha' < \alpha$ then the predecessor of $\alpha$ in $X$ is not in $B_i \cap X$ (otherwise $[\alpha, \beta]$ would not be a maximal interval of $B_i \cap X$ in $X$) but is in $B_j \cap X$. 
    Therefore, $B_j \cap X \not\subseteq B_i \cap X$, so it cannot be that $B_j \setminus A_j \subseteq B_i \setminus A_i$, and thus $B_i \setminus A_i \subsetneq B_j \setminus A_j$. 
    If $\alpha' = \alpha$ and $-\beta' < -\beta$ then $\beta < \beta'$ so the successor of $\beta$ in $X$ is not in $B_i \cap X$ but is in $B_j \cap X$. 
    As before, this implies $B_i \setminus A_i \subsetneq B_j \setminus A_j$. 
    Otherwise, $(\alpha', \beta') = (\alpha, \beta)$, so $-|B_j \setminus A_j| \leq -|B_i \setminus A_i|$ so $|B_i \setminus A_i| \leq |B_j \setminus A_j|$. Since one of them is a proper subset of the other, this inequality must be strict, which implies $B_i \setminus A_i \subsetneq B_j \setminus A_j$. 
    We considered all possible cases, and in each of them we indeed have $B_i \setminus A_i \subsetneq B_j \setminus A_j$, which proves that the first invariant is maintained.

    We now move to the second invariant. Suppose that we add $w = (\alpha, -\beta, -|B_i \setminus A_i|, i)$ to $S$. The fact that all vectors $w'$ below $w$ in $S$ are generated by separations $(A_j, B_j)$ such that $B_i \setminus A_i \subsetneq B_j \setminus A_j$ follows from the first invariant. Furthermore, since for every such $w'$ the corresponding interval contains $y$ then there can be at most one vector $w'$ for each separation $(A_j, B_j) \in \mathcal{S}$.
    Conversely, consider a separation $(A_j, B_j) \in \mathcal{S}$ such that $B_i \setminus A_i \subsetneq B_j \setminus A_j$. 
    Then, $[\alpha, \beta] \cap X \subseteq B_i \cap X \subseteq B_j \cap X$.
    Therefore, there exists a maximal interval $[\alpha', \beta']$ of $B_j \cap X$ in $X$ such that $[\alpha, \beta] \subseteq [\alpha', \beta']$. 
    Let $w'$ be the vector $(\alpha', -\beta', -|B_j \setminus A_j|, j)$ which corresponds to this interval.
    It then follows from the fact that $B_i \setminus A_i \subsetneq B_j \setminus A_j$ that $w' \leq_{lex} w$, so $w'$ was initially above $w$ in $L$.
    Since are considering $w$, $w'$ was removed from $L$ at some point and thus added to $S$ (we only remove an element from $L$ when adding it to $S$). If $y$ is not the first vertex of $[\alpha', \beta']$ then $w'$ was in $S$ when we started considering $y$ by the first part of the third invariant. 
    Since we didn't yet remove any element from $S$ while considering $y$, it still holds that $w'$ is in $S$. 
    Otherwise, $y$ is the first vertex of $[\alpha', \beta']$ and thus $w'$ was added to $S$ when considering $y$, therefore $w'$ cannot have been removed from $S$ yet.
    In both cases, we still have $w' \in S$.
    This proves the first part of the second invariant.
    
    Suppose that $w$ is the vector corresponding to the first interval $[\alpha, \beta]$ of $B_i \cap X$ and that we add $w$ to $S$. Consider the iteration of the While loop at \cref{line:begin-while} where $w$ was added to $S$. Then, $w$ is the tuple considered by the algorithm (defined either at \cref{line:def-w-1} or at \cref{line:def-w-2}). Thus, $a = \alpha$ is the first element of $B_i \cap X$ so we enter the If at \cref{line:test-first}.
    If $S$ was empty before adding $w$ then the last part of the invariant holds since we add the edge $i0$ to $F$, label it $i$ and set $parent(i) = 0$. 
    Indeed, in that case, by the first part of the second invariant, there is no separation $(A_l, B_l) \in \mathcal{S}$ such that $B_i \setminus A_i \subsetneq B_l \setminus A_l$.
    Otherwise, let $w'$ be the element that was on top of $S$ before adding $w$, and say that $w'$ is generated by $(A_j, B_j) \in \mathcal{S}$. 
    The first invariant and the first part of the second invariant together imply that $(A_j, B_j)$ is the separation with $B_j \setminus A_j$ inclusion-wise minimal among all separations $(A_l, B_l)$ in $\mathcal{S}$ such that $B_i \setminus A_i \subsetneq B_l \setminus A_l$. At \cref{line:add-edge}, we add an edge $ij$, label this edge $i$ and set $parent(i) = j$. 
    When we added the vector corresponding to the first interval of $B_j \cap X$ to $S$, we created a path from the node $j$ to the node 0 in $F$. Therefore, adding the edge $ij$ connects the node $i$ to the node 0 in $F$.
    The labels of the edges on the path from $i$ to 0 are $i$ and the labels of the edges of the path from $j$ to 0.
    By the last part of the second invariant applied to $j$, an index $l \in [k]$ labels one of these edges if and only if either $l=j$ or $B_j \setminus A_j \subsetneq B_l \setminus A_l$. 
    Therefore, an index $l \in [k]$ labels an edge of the path from $i$ to 0 if and only if either $l = i$ or $B_i \setminus A_i \subsetneq B_l \setminus A_l$.
    This concludes the proof for the second invariant. 

    We now prove the second part of the third invariant. 
    We know that at \cref{line:enter-loop} the vectors in $S$ correspond exactly to the intervals $[a, b]$ of the $B_i \cap X$ such that $y \in (a, b]$.
    Since we only add $w$ to $S$ if $y$ is the first vertex of the interval corresponding to $w$ then every vector $w$ in $S$ at \cref{line:end-while-insert-stack} corresponds to an interval $[a, b]$ such that $y \in [a, b]$.
    Conversely, consider an interval $[\alpha, \beta]$ of a $B_i \cap X$ such that $y \in [\alpha, \beta]$. 
    Let $w$ be the vector corresponding to that interval.
    Since we did not yet remove any element from $S$ while considering $y$, if $w$ was in $S$ at \cref{line:enter-loop} then $w$ is still in $S$ at \cref{line:end-while-insert-stack}. 
    Thus, by the first part of the second invariant, if $y \in (\alpha, \beta]$ then $w$ is still in $S$ at \cref{line:end-while-insert-stack}. 
    Assume now that $y$ is the first vertex of the interval corresponding to $w$, i.e. $y = \alpha$. 
    If $L$ is empty then $w$ was added to $S$ at some point since we only remove an element from $L$ when adding it to $S$. 
    If $L$ is non-empty, let $w' = (\alpha', -\beta', -|B_j \setminus A_j|, j)$ be the top element of $L$. 
    Since $L \neq \emptyset$, the last test at \cref{line:begin-while} was performed with $w'$ (either at \cref{line:def-w-1} or at \cref{line:def-w-2}). 
    Since we are at \cref{line:end-while-insert-stack}, we didn't enter the While loop and therefore $\alpha' \neq y$. 
    If $\alpha' < y$ then by the third invariant applied to $\alpha'$, we get that $w'$ was in $S$ when we were considering $\alpha'$ and therefore was removed from $L$, which is not the case. Thus, $\alpha' > y$ so $\alpha = y < \alpha'$, thus $w$ was above $w'$ in $L$. Since $w'$ is the top element of $L$ then $w$ was removed from $L$ and thus added to $S$.
    In both cases, $w$ was added to $S$ at some point.
    Since $\alpha = y$ and because of the test at \cref{line:begin-while}, $w$ must have been added to $S$ when considering $y$ and thus cannot have been removed from $S$ yet. Thus, $w$ is still in $S$.
    This proves that at \cref{line:end-while-insert-stack} the elements of $S$ correspond to the intervals $[a, b]$ of the $B_i \cap X$ such that $y \in [a, b]$.

    To prove the last part of the third invariant, it suffices to show that we remove from $S$ exactly all the vectors $w$ such that $y$ is the last vertex of the interval corresponding to $w$, i.e. such that $y = \beta$.
    First, consider a vector $w$ which we remove from $S$. Write $w = (\alpha, -\beta, -|B_i \setminus A_i|, i)$.
    Consider the time when we enter the While loop at \cref{line:while-remove} to remove $w$. 
    At that point, $w$ is on top of $S$.
    The current value of $b$ was defined either at \cref{line:def-l-1} or at \cref{line:def-l-2} when $w$ was already on top of $S$ and thus $b = \beta$. 
    Since we enter the While loop at \cref{line:while-remove}, it holds that $\beta = y$.
    Conversely, consider a vector $w = (\alpha, -\beta, -|B_i \setminus A_i|, i)$ such that $\beta = y$. 
    By the second part of the third invariant, $w$ was in $S$ at \cref{line:end-while-insert-stack}. 
    By contradiction, suppose that we didn't remove $w$ from $S$. 
    Thus, $S \neq \emptyset$. 
    Let $w' = (\alpha', -\beta', -|B_j \setminus A_j|, j)$ be the element on top of $S$.
    By the first invariant, we have that $B_j \setminus A_j \subseteq B_i \setminus A_i$, and thus $B_j \cap X \subseteq B_i \cap X$. 
    Since $y = \beta$ is the last vertex of $w$ then the successor of $y$ in $X$ is not in $B_i \cap X$, hence not in $B_j \cap X$. 
    However, the interval corresponding to $w'$ contains $y$ (because $w'$ was in $S$ at \cref{line:end-while-insert-stack}), so $y$ is the last vertex of $w'$, i.e. $y = \beta'$.
    Thus, when we last defined $b$ (either at \cref{line:def-l-1} or at \cref{line:def-l-2}), $w'$ was on top of $S$ so $b$ was set to $\beta' = y$. 
    Therefore, we should have entered the While loop at \cref{line:while-remove} and removed $w'$ from $S$, a contradiction. 
    Thus, every vector $w = (\alpha, -\beta, -|B_i \setminus A_i|, i)$ such that $\beta = y$ was indeed removed from $S$.
    Therefore when we finish examining $y$, the elements of $S$ correspond exactly to the intervals $[a, b]$ of the $B_i \cap X$ such that $y \in [a, b)$.
    
    This concludes the proof of the invariants.
    \end{subproof}

    We now use these invariants to prove the correctness of the algorithm.
    By the third invariant, every vector $w$ is added to $S$ at some point.
    However, every $B_i \cap X$ is non-empty by \cref{lem:leaf-X-nonempty}.
    Therefore, by the second part of the second invariant, every node in $F$ is connected to the node 0. 
    Furthermore, we add an edge to $F$ only when considering the vector corresponding to the first interval of $B_i \cap X$ (because of the test at \cref{line:test-first}), so we have at most $p$ edges in $F$. 
    Since $F$ is connected and has $p+1$ nodes then $F$ is a tree.
    Since we define a parent when we add a new edge to $F$ then $F$ is a rooted tree, and its root is the only node with no parent, namely the node 0.
    The edges of $F$ are labelled from $1$ to $p$ by construction.
    By the second invariant, the edge labelled $i$ is a proper ancestor of the edge labelled $j$ if and only if $B_j \setminus A_j \subsetneq B_i \setminus A_i$.
    Thus, by \cref{prop:isomorphism-T}, there is an isomorphism of rooted trees between $F$ and $T$ that maps the edge of $T$ which induces the separation $(A_i, B_i)$ to the edge labelled $i$ in $F$.
    This proves that the tree $F$ we return is isomorphic to $T$ and its edges are labelled according to the separations induced in $(T, \mathcal{V})$.

    Finally, let $y \in X$. When we are at line \cref{line:end-while-insert-stack}, by the third invariant, the elements of $S$ correspond to the intervals $[a, b]$ of the $B_i \cap X$ such that $y \in [a, b]$.
    At that point, if $S$ is empty, we add $y$ to $W_0$, the bag of the root of $F$. However, by \cref{lem:which-node-x}, in that case we have $y \in V_{t_x}$, i.e. $y$ is in the bag of the root node of $T$ so we put $y$ in the correct bag.
    If $S$ is not empty, then $y \notin V_{t_x}$ by \cref{lem:which-node-x}. By the first invariant, the top element $w$ of $S$ is generated by the separation $(A_i, B_i)$ with $y \in B_i$ such that $B_i \setminus A_i$ is inclusion-wise minimal. We add $y$ to the bag of the node $i$, i.e. to the bag of the top endpoint of the edge labelled $i$. By the property of the isomorphism between $F$ and $T$, this node corresponds to the top endpoint $t_i$ of the edge $e_i$ of $T$ which induces the separation $(A_i, B_i)$. By \cref{lem:which-node-x}, we have $y \in V_{t_i}$, which concludes the proof of correctness.

    We now argue about the running time of this algorithm. We already explained how to perform all steps prior to \cref{alg:compute-forest} in time $O(n + p(\ell + k))$.
    Iterating over all vertices in $X$ can be done in time $O(n)$.
    Fix $y \in X$ and let us examine the time spent while considering $y$.
    Getting to \cref{line:begin-while} takes constant time.
    Then, every iteration into the While loop at \cref{line:begin-while} also takes constant time. For every such iteration we remove an element from $L$.
    Once we are done with this While loop, we get to \cref{line:while-remove} in constant time (if we get there). Once again, every iteration into this While loop takes constant time, and for every such iteration we remove an element from $S$.
    Therefore, the time spent while considering $y$ is $O(1) + O(r_y)$, where $r_y$ is the number of elements removed from $S$ and from $L$ while considering $y$.
    Since overall we remove at most $p\ell$ elements from $L$ and $p\ell$ elements from $S$, the total running time is indeed $O(n + p(\ell + k))$.
\end{proof}

The next result explains, for every vertex $y \in V(G)$ which belongs to several bags of the tree-decomposition induced by $\mathcal{S}$, how to find the corresponding nodes. 

\begin{lemma}  \label{lem:which-not-x}
    Let $G$ be a graph and $\mathcal{S} = \{(A_1, B_1), \ldots, (A_p, B_p)\}$ be a set of nested separations of $G$.
    Let $(T, \mathcal{V})$ be the tree-decomposition of $G$ induced by $\mathcal{S}$. 
    Let $X \subseteq V(G)$ be the set of vertices which only appear in a single bag of $(T, \mathcal{V})$.
    If $y \in V(G) \setminus X$ and $t \in V(T)$ then $y \in V_t$ if and only if there is an edge incident to $t$ which induces a separation $(A_i, B_i) \in \mathcal{S}$ such that $y \in A_i \cap B_i$. 
\end{lemma}

\begin{proof}
    Suppose first that $y \in V_t$. Since $y \notin X$, there is a node $u \neq t \in V(T)$ such that $y \in V_{u}$. 
    Let $e = tt'$ be the first edge of the path from $t$ to $u$ in $T$. 
    Since $y \in V_t \cap V_u$ and $(T, \mathcal{V})$ is a tree-decomposition then $y \in V_{t'}$. 
    The edge $e$ induces a separation $(A_i, B_i) \in \mathcal{S}$ with $V_t \subseteq A_i$ and $V_{t'} \subseteq B_i$ up to renaming $A_i$ and $B_i$, and therefore $y \in A_i \cap B_i$.
    Conversely, suppose that there is an edge incident to $t$ which induces a separation $(A_i, B_i) \in \mathcal{S}$ such that $y \in A_i \cap B_i$. 
    Let $T_1, T_2$ be the two connected components of $T - e$. 
    By definition of the separations induced by $(T, \mathcal{V})$, there exist $t_1 \in V(T_1)$ and $t_2 \in V(T_2)$ such that $y \in V_{t_1} \cap V_{t_2}$. 
    Since $t_1$ and $t_2$ are not in the same connected component of $T - e$, the path between them in $T$ uses $e$ and thus $t$ is on this path. 
    Since $T$ is a tree-decomposition and $y \in V_{t_1} \cap V_{t_2}$ then $y \in V_t$.
\end{proof}

Using \cref{lem:which-not-x}, we can turn the ``partial tree-decomposition'' obtained from \cref{prop:compute-tree} into the desired tree-decomposition.

\begin{lemma} \label{lem:algo-add-not-x}
    Let $G$ be a graph and $\mathcal{S} = \{(A_1, B_1), \ldots, (A_p, B_p)\}$ be a set of nested separations of $G$.
    Let $(T, \mathcal{V})$ be the tree-decomposition of $G$ induced by $\mathcal{S}$. 
    Let $X \subseteq V(G)$ be the set of vertices which only appear in a single bag of $(T, \mathcal{V})$.
    There is an $O(n + pk)$-time algorithm which, given the graph $G$, the tree $T$ with its edges labelled so that the edge which induces the separation $(A_i, B_i)$ is labelled with $i$, the sets $A_i \cap B_i$ each of size at most $k$, and bags $(W_t)_{t \in V(T)}$ such that $W_t = V_t \cap X$ for every $t \in V(T)$, constructs the tree-decomposition $(T, \mathcal{V})$.
\end{lemma}

\begin{proof}
    We start by setting $V_t = W_t$ for every $t \in V(T)$.
    Then, we iterate over every edge $e = tt'$ of $T$ and look at the corresponding label $i \in [p]$. We add every vertex in $A_i \cap B_i$ to $V_t$ and $V_{t'}$.
    This finishes the description of the algorithm.

    We first argue that it runs in time $O(n + pk)$. 
    First, since every edge of $T$ induces a separation $(A_i, B_i)$ then $T$ has $p$ edges.
    Furthermore, every vertex in a $W_t$ is in $X$, hence is in a unique bag $W_t$, and thus the sum of the sizes of all $W_t$ is $O(n)$, so initializing every $V_t$ to $W_t$ takes time $O(n)$.
    Then, going over every edge of $T$ can be done in time $O(p)$ by a BFS for instance. When going over an edge, we can access its label $i$ in constant time, hence the set $A_i \cap B_i$ in constant time, and adding it to the two bags takes time $O(k)$.
    Thus, this algorithm indeed runs in time $O(n + pk)$.

    The correctness of the algorithm is immediate from \cref{lem:leaf-X-nonempty,lem:which-not-x}.
\end{proof}

The next result is a straightforward characterization of the torso edges.

\begin{lemma} \label{lem:which-edges-torso}
    Let $G$ be a graph and $\mathcal{S} = \{(A_1, B_1), \ldots, (A_p, B_p)\}$ be a set of nested separations of $G$.
    Let $(T, \mathcal{V})$ be the tree-decomposition of $G$ induced by $\mathcal{S}$. 
    Let $X \subseteq V(G)$ be the set of vertices which only appear in a single bag of $(T, \mathcal{V})$.
    Let $t \in V(T)$ and let $u, v \in V(G)$. Then, $uv$ is an edge of the torso at $t$ if and only if one of the following holds. \begin{itemize}
        \item $t$ is the unique node of $T$ such that $\{u, v\} \subseteq V_t$ and $uv \in E(G)$.
        \item There is an edge $e \in E(T)$ incident to $t$ which induces a separation $(A_i, B_i) \in \mathcal{S}$ such that $\{u, v\} \subseteq A_i \cap B_i$.
    \end{itemize}
\end{lemma}

\begin{proof}
    The converse implication is trivial, so we focus on proving the direct implication.
    Suppose that $uv$ is an edge of the torso at $t$. Then, $\{u, v\} \subseteq V_t$. If $t$ is the unique node of $T$ such that $\{u, v\} \subseteq V_t$ then the reason why the edge $uv$ appears in the torso at $t$ is not because $u$ and $v$ appear in the same adhesion set. Thus, $uv$ is an edge of $G[V_t]$, hence an edge of $G$.
    Otherwise, there is another node $t'$ of $T$ such that $\{u, v\} \subseteq V_{t'}$. Let $e = ts$ be the first edge of the path from $t$ to $t'$ in $T$. Since $(T, \mathcal{V})$ is a tree-decomposition and $\{u, v\} \subseteq V_t \cap V_{t'}$ then $\{u, v\} \subseteq V_s$. Let $(A_i, B_i) \in \mathcal{S}$ be the separation induced by $e$. Up to renaming $A_i$ and $B_i$, we have $V_t \subseteq A_i$ and $V_s \subseteq B_i$, so $\{u, v\} \subseteq A_i \cap B_i$.
\end{proof}

The following result states that, given a tree-decomposition $(T, \mathcal{V})$ of a graph $G$, we can efficiently compute for every edge ${xy \in E(G)}$, a node $t \in V(T)$ such that $\{x, y\} \subseteq V_t$.

\begin{lemma} \label{lem:find-common-bag}
    There is an ${O\left(m + |E(T)| + \sum_{t \in V(T)}|V_t| \right)}$-time algorithm which, given as input a graph $G$ and a tree-decomposition $(T, \mathcal{V})$ of $G$, computes for every edge $e=xy \in E(G)$ a node $t_e \in V(T)$ such that $\{x, y\} \subseteq V_{t_e}$.
\end{lemma}

\begin{proof}
    Start by rooting the tree $T$ at an arbitrary node.
    For each vertex $x \in V(G)$, denote by $r_x$ the root of the subtree $T_x$ of $T$ consisting of all nodes $t \in V(T)$ such that $x \in V_t$.
    By iterating through all sets $V_t$ for $t \in V(T)$, we can compute all roots $r_x$ for $x \in V(G)$ in time $O\left(|E(T)| + \sum_{t \in V(T)}|V_t|\right)$. 
    Then, in time $O(|E(T)|)$, we can compute for every node $t \in V(T)$ its depth in the tree $T$.

    Then, we iterate through all edges of $G$.
    Consider such an edge $xy \in E(G)$. Since $(T, \mathcal{V})$ is a tree-decomposition of $G$, there exists a node $t \in V(T)$ such that $\{x, y\} \subseteq V_t$.
    This node $t$ is a node of the subtree $T_x$ and of the subtree $T_y$, hence it is a descendant of both $r_x$ and $r_y$.
    Let $r'$ be the node among $r_x$ and $r_y$ which is furthest from the root.
    Then, $r'$ is in the path in $T$ between $r_x$ and $t$, and in the path in $T$ between $r_y$ and $t$, so $\{x, y\} \subseteq V_{r'}$.
    Therefore, it suffices to store the node $r'$ for the edge $xy$, and $r'$ can be computed in constant time.
    Overall, the total running time is indeed ${O\left(m + |E(T)| + \sum_{t \in V(T)}|V_t| \right)}$.
\end{proof}

Using \cref{lem:which-edges-torso,lem:find-common-bag}, it is simple to build the torsos of a tree-decomposition given as input.

\begin{lemma} \label{lem:algo-build-torsos}
    Let $G$ be a connected graph and $\mathcal{S} = \{(A_1, B_1), \ldots, (A_p, B_p)\}$ be a set of nested separations of $G$.
    Let $(T, \mathcal{V})$ be the tree-decomposition of $G$ induced by $\mathcal{S}$. 
    There is an $O(n + m + pk^2)$-time algorithm which, given the tree-decomposition $(T, \mathcal{V})$, a labelling of the edges of $T$ such that the edge which induces the separation $(A_i, B_i)$ is labelled with $i$, and the sets $A_i \cap B_i$ each of size at most $k$, constructs all the torsos of the tree-decomposition $(T, \mathcal{V})$.
\end{lemma}

\begin{proof}
    We construct a family of graphs $(G_t)_{t \in V(T)}$ in 3 phases. \begin{enumerate}
        \item We iterate over all nodes $t \in V(T)$ and set $V(G_t) = V_t$ for each of them.
        \item We compute for all edges $e=xy \in E(G)$ a node $t_e \in V(T)$ such that $\{x, y\} \subseteq V_{t_e}$.
        We then iterate over all edges $e = xy$ of $G$ and for each of them we add the edge $xy$ to the graph $G_{t_e}$.
        \item We iterate over all edges $e = st$ of $T$. For each such edge $e$, we look at its label $i \in [p]$ and consider the set $A_i \cap B_i$. Then, we add all edges between two elements of $A_i \cap B_i$ to both graphs $G_s$ and $G_t$.
    \end{enumerate}

    We prove that for every $t \in V(T)$, the graph $G_t$ is the torso at $t$.
    First, all graphs $G_t$ have the correct vertex set, namely $V_t$. Then, let $t \in V(T)$ and let $uv$ be an edge of the torso at $t$. By \cref{lem:which-edges-torso}, either $t$ is the unique node of $T$ such that $\{u, v\} \subseteq V_t$ and $uv \in E(G)$ or there is an edge $e \in E(T)$ incident to $t$ which induces a separation $(A_i, B_i) \in \mathcal{S}$ such that $\{u, v\} \subseteq A_i \cap B_i$.
    In the first case, we have $t = t_{uv}$ and thus we added the edge $uv$ to $G_t$ in the second phase of the algorithm. 
    In the second case, we added the edge $uv$ to $G_t$ in the third phase of the algorithm.
    Conversely, suppose that we added an edge $uv$ to $G_t$ at some point of the algorithm.
    If this edge was added during the second phase then $\{u, v\} \subseteq V_t$ and $uv$ is an edge of $G$, so $uv$ is an edge of the torso at $t$.
    Otherwise, there is an edge $e \in E(T)$ incident to $t$ which induces a separation $(A_i, B_i) \in \mathcal{S}$ such that $\{u, v\} \subseteq A_i \cap B_i$. In that case, $uv$ is again an edge of the torso at $t$ by \cref{lem:which-edges-torso}.

    We now prove that this algorithm runs in time $O(n + m + pk^2)$.
    Note that $\sum_{t \in T}|V_t| = n + \sum_{i=1}^p |A_i \cap B_i| \leq n + pk$.
    Therefore, the first phase can be performed in time $O(n + pk)$.
    By \cref{lem:find-common-bag}, the second phase can be done in time $O(n+m+pk)$.
    There are $p$ edges in $T$ so the third phase takes time $O(pk^2)$. Therefore, this algorithm runs in time $O(n + m + pk^2)$ overall.
\end{proof}

Putting together \cref{prop:compute-tree,lem:algo-add-not-x,lem:algo-build-torsos}, we can finally prove \cref{thm:build-tree+torsos}.

\begin{proof}[\textit{Proof of \cref{thm:build-tree+torsos}}.]
    Computing all sets $A_i \cap B_i$ explicitly can be done in time $O(p(\ell+k))$.
    Let $X \subseteq V(G)$ be the set of vertices which only appear in a single bag of $(T, \mathcal{V})$.
    By \cref{prop:compute-tree}, there is an algorithm which constructs $T$ and labels its edges from $1$ to $p$ such that the edge labelled $i$ induces the separation $(A_i, B_i)$ in $(T, \mathcal{V})$.
    Furthermore, it constructs bags $(W_t)_{t \in V(T)}$ such that for every $t \in V(T)$, we have $W_t = V_t \cap X$.
    This algorithm runs in time $O(n + p(\ell + k))$.
    Then, by \cref{lem:algo-add-not-x}, there is an algorithm which constructs $(T, \mathcal{V})$ in time $O(n + pk)$.
    Finally, by \cref{lem:algo-build-torsos}, there is an algorithm which builds all the torsos of $(T, \mathcal{V})$ in time $O(n + m + pk^2)$.
\end{proof}

\end{document}